\documentclass[11pt,reqno]{amsart}
\usepackage{amsmath, amsthm, amssymb}
\usepackage{enumitem}
\usepackage[foot]{amsaddr}

\usepackage{fullpage}
\usepackage[many]{tcolorbox}
\usepackage{xcolor}
\usepackage{wasysym}
\usepackage{mathtools}
\usepackage{scalerel}
\usepackage{bbm}
\usepackage[all]{xy}

\usepackage{tikz}
\usetikzlibrary{arrows,backgrounds,patterns.meta}
\usetikzlibrary{positioning,shadings,cd}
\usetikzlibrary{shapes}
\usetikzlibrary{backgrounds}
\usetikzlibrary{decorations,decorations.pathreplacing,decorations.markings,decorations.pathmorphing}
\tikzstyle{snake}=[decorate, decoration={snake, segment length=1mm, amplitude=.5mm}]
\usetikzlibrary{fit,calc,through}
\usetikzlibrary{external}
\newcommand{\tikzmath}[2][]
{\vcenter{\hbox{\begin{tikzpicture}[#1]#2\end{tikzpicture}}}
}
\newcommand{\roundNbox}[6]{
	\draw[rounded corners=5pt, very thick, #1] ($#2+(-#3,-#3)+(-#4,0)$) rectangle ($#2+(#3,#3)+(#5,0)$);
	\coordinate (ZZa) at ($#2+(-#4,0)$);
	\coordinate (ZZb) at ($#2+(#5,0)$);
	\node at ($1/2*(ZZa)+1/2*(ZZb)$) {#6};
}

\tikzset{super thick/.style={line width=3pt}}
\tikzstyle{far>}=[decoration={markings, mark=at position 0.75 with {\arrow{>}}}, postaction={decorate}]
\tikzstyle{mid>}=[decoration={markings, mark=at position 0.55 with {\arrow{>}}}, postaction={decorate}]
\tikzstyle{mid<}=[decoration={markings, mark=at position 0.55 with {\arrow{<}}}, postaction={decorate}]
\tikzset{super thick/.style={line width=3pt}}
\tikzstyle{far>}=[decoration={markings, mark=at position 0.75 with {\arrow{>}}}, postaction={decorate}]
\tikzstyle{mid>}=[decoration={markings, mark=at position 0.55 with {\arrow{>}}}, postaction={decorate}]
\tikzstyle{mid<}=[decoration={markings, mark=at position 0.55 with {\arrow{<}}}, postaction={decorate}]
\tikzstyle{knot}=[preaction={super thick, white, draw}]
\tikzstyle{coupon}=[draw, very thick, rectangle, rounded corners=5pt]
\tikzset{Rightarrow/.style={double equal sign distance,>={Implies},->},
triplecd/.style={-,preaction={draw,Rightarrow}},
quadruplecd/.style={preaction={draw,Rightarrow,
shorten >=0pt
},
shorten >=1pt,
-,double,double
distance=0.2pt}}
\tikzset{
    tripleline/.style args={[#1] in [#2] in [#3]}{
        #1,preaction={preaction={draw,#3},draw,#2}
    }
}
\tikzstyle{triple}=[tripleline={[line width=.15mm,black] in
      [line width=.7mm,white] in
      [line width=1mm,black]}] 
\tikzset{
    quadrupleline/.style args={[#1] in [#2] in [#3] in [#4]}{
        #1,preaction={preaction={preaction={draw,#4},draw,#3}, draw,#2}
    }
}
\tikzstyle{quadruple}=[quadrupleline={[line width=.3mm,white] in
      [line width=.6mm,black] in
      [line width=1.2mm,white] in
      [line width=1.5mm,black]}]

\definecolor{violet}{RGB}{148,0,211}
\definecolor{DarkGreen}{RGB}{0,150,0}
\definecolor{rufous}{HTML}{A81C07}
\newcommand{\ModuleColor}{blue}

\usepackage[plainpages=false,hypertexnames=false,pdfpagelabels]{hyperref}
\definecolor{medium-blue}{rgb}{0,0,.8}
\hypersetup{colorlinks, linkcolor={purple}, citecolor={medium-blue}, urlcolor={medium-blue}}
\newcommand{\arxiv}[1]{\href{http://arxiv.org/abs/#1}{\tt arXiv:\nolinkurl{#1}}}
\newcommand{\arXiv}[1]{\href{http://arxiv.org/abs/#1}{\tt arXiv:\nolinkurl{#1}}}

\newcommand{\doi}[1]{\href{http://dx.doi.org/#1}{{\tt DOI:#1}}}

\newcommand{\googlebooks}[1]{(preview at \href{https://books.google.com/books?id=#1}{google books})}

\DeclareMathOperator{\alt}{alt}
\DeclareMathOperator{\coev}{coev}
\DeclareMathOperator{\End}{End}
\DeclareMathOperator{\eval}{eval}
\DeclareMathOperator{\Forget}{Forget}
\DeclareMathOperator{\ev}{ev}

\DeclareMathOperator{\Hom}{Hom}
\DeclareMathOperator{\id}{id}
\DeclareMathOperator{\im}{im}

\DeclareMathOperator{\Irr}{Irr}

\DeclareMathOperator{\Tr}{Tr}
\DeclareMathOperator{\tr}{tr}
\DeclareMathOperator{\Tube}{Tube}

\newcommand{\set}[2]{\left\{#1 \middle| #2\right\}}

\newcommand{\Mod}{\mathsf{Mod}}
\newcommand{\Bim}{\mathsf{Bim}}

\newcommand{\DHR}{\mathsf{DHR}}
\newcommand{\Hilb}{\mathsf{Hilb}}

\newcommand{\Rep}{\mathsf{Rep}}

\newcommand{\QSys}{\mathsf{QSys}}
\newcommand{\CStar}{\mathsf{C^*Alg}}

\def\semicolon{;}
\def\applytolist#1{
    \expandafter\def\csname multi#1\endcsname##1{
        \def\multiack{##1}\ifx\multiack\semicolon
            \def\next{\relax}
        \else
            \csname #1\endcsname{##1}
            \def\next{\csname multi#1\endcsname}
        \fi
        \next}
    \csname multi#1\endcsname}

\def\calc#1{\expandafter\def\csname c#1\endcsname{{\mathcal #1}}}
\applytolist{calc}QWERTYUIOPLKJHGFDSAZXCVBNM;
\def\bbc#1{\expandafter\def\csname bb#1\endcsname{{\mathbb #1}}}
\applytolist{bbc}QWERTYUIOPLKJHGFDSAZXCVBNM;
\def\bfc#1{\expandafter\def\csname bf#1\endcsname{{\mathbf #1}}}
\applytolist{bfc}QWERTYUIOPLKJHGFDSAZXCVBNM;
\def\sfc#1{\expandafter\def\csname s#1\endcsname{{\sf #1}}}
\applytolist{sfc}QWERTYUIOPLKJHGFDSAZXCVBNM;
\def\fc#1{\expandafter\def\csname f#1\endcsname{{\mathfrak #1}}}
\applytolist{fc}QWERTYUIOPLKJHGFDSAZXCVBNM;
\def\rmc#1{\expandafter\def\csname rm#1\endcsname{{\mathrm #1}}}
\applytolist{rmc}QWERTYUIOPLKJHGFDSAZXCVBNM;

\numberwithin{equation}{section}

\theoremstyle{plain}
\newtheorem{thm}[equation]{Theorem}
\newtheorem*{thm*}{Theorem}

\newtheorem{lem}[equation]{Lemma}
\newtheorem{prop}[equation]{Proposition}

\newtheorem*{claim*}{Claim}

\newtheorem{thmalpha}{Theorem}

\theoremstyle{definition}
\newtheorem{defn}[equation]{Definition}

\newtheorem*{trick*}{Trick}
\newtheorem{construction}[equation]{Construction}
\newtheorem{nota}[equation]{Notation}

\newtheorem{ex}[equation]{Example}

\newtheorem{subex}[equation]{Sub-example}

\newtheorem{rem}[equation]{Remark}

\title{Holography for bulk-boundary local topological order}
\date{\today}
\begin{document}
\author{Corey Jones$^1$}
\address{$^1$ Department of Mathematics, North Carolina State University, Raleigh, NC 27695, USA}
\author{Pieter Naaijkens$^2$}
\address{$^2$ School of Mathematics, Cardiff University, Cardiff, CF24 4AG, United Kingdom}
\author{David Penneys$^3$}
\address{$^3$ Department of Mathematics, The Ohio State University, Columbus, OH 43210, USA}

\begin{abstract}
In our previous article [\arXiv{2307.12552}], we introduced local topological order (LTO) axioms for quantum spin systems which allowed us to define a physical boundary (associated to a cut of the lattice) manifested by a net of boundary algebras in one dimension lower.
This gives a formal setting for topological holography, where the braided tensor category of DHR bimodules of the physical boundary algebra captures the bulk topological order.

In this article, we extend the LTO axioms to quantum spin systems equipped with a topological boundary (domain wall with the trivial phase), again producing a physical boundary algebra for the bulk-boundary system, whose category of (topological) boundary DHR bimodules recovers the topological boundary order.
We perform this analysis in explicit detail for Levin-Wen and Walker-Wang bulk-boundary systems.

Along the way, we introduce a 2D braided categorical net of algebras built from a unitary braided fusion category (UBFC).
Such nets arise as boundary algebras of Walker-Wang models. 
We consider the canonical state on this braided categorical net corresponding to the standard topological boundary for the Walker-Wang model.
Interestingly, in this state, the cone von Neumann algebras are type I with finite dimensional centers, in contrast with the type II and III cone von Neumann algebras from the Levin-Wen models studied in [\arXiv{2307.12552}].
The superselection sectors recover the underlying unitary category of our UBFC, and it was recently proven in [\arXiv{2609.20725}] that the superselection category also captures the fusion and braiding.
\end{abstract}
\maketitle
\tableofcontents

\section{Introduction}

Topologically ordered spin systems are fundamental examples of quantum many-body systems whose ground states exhibit long-range entanglement. 
A particularly interesting class of these systems exhibit \textit{local topological order} (LTO), defined in terms of error correcting properties of their local ground state spaces. 
This class includes all string-net models \cite{MR1611329,MR1951039,PhysRevB.71.045110,PhysRevB.103.195155}, and thus the class of LTOs exhibits all possible non-chiral topological orders in (2+1)D. 

In our previous work~\cite{MR4945955}, we introduced additional LTO-type axioms, which we review in detail in~\S\ref{sec:LTO}.
Roughly speaking, the first LTO axiom is essentially the same as the \textit{local topological quantum order} axioms of~\cite{MR2742836} and says that ground states are locally indistinguishable.
The remaining three allow for a well-defined notion of a \textit{boundary algebra} related to a \textit{physical boundary}.
The physical boundary is obtained by cutting the system using a hyperplane.\footnote{In practice, the construction can be thought of as cutting off half of the system, and forgetting about the interaction terms across the cut. }
We can define a corresponding boundary algebra, which essentially can be thought of as operators creating excitations (`errors') at the boundary, but not in the bulk.
The LTO axioms guarantee that these boundary algebras are well-defined and have the structure of a local net of $\rm{C}^*$-algebras in one dimension lower than the original LTO system.
It allows for a parameterization of boundary states of the system, and in (2+1)D it captures the bulk topological order via \emph{DHR bimodules}~\cite{MR4814692}, which are bimodules which can be localized in arbitrary intervals along the boundary, and should be thought of as analogous to DHR superselection sectors.
The relation of the (physical) boundary theory to the bulk topological order manifests a kind of topological holographic principle in the spirit of~\cite{PhysRevB.107.155136,2310.05790}.

In this article, we extend this story to the case of \textit{topological boundaries} of locally topologically ordered spin systems in 2+1 and 3+1 dimensions, introducing the notion of \textit{boundary LTO}.
The topological boundary comes from a boundary of the underlying (bulk) lattice on which the system is defined, and can be thought of as a domain wall between a topologically ordered system and the trivial phase \cite{MR2942952,MR4640433}.
A boundary LTO is then a system that satisfies the \textit{boundary LTO axioms}, which are appropriate modifications of the LTO axioms designed to capture the standard topological boundaries of string-net models. 
We again cut the bulk and boundary by a hyperplane $\mathcal{K}$, and construct a (physical) boundary algebra for the topological bulk-boundary system described by a boundary LTO.
In 1D and 2D this looks as follows, where
red denotes the bulk, black denotes the topological boundary, and the cut corresponding to the physical boundary is denoted in blue.
$$
\tikzmath{
\draw[thick, step=.75, thick, red] (.01,0.25) grid (3.5,2.75);
\draw[thick, black] (0,.25) -- (0,2.75);
\draw[thick, cyan] (-.5,1.875) node[above]{$\scriptstyle \cK$} -- (3.5,1.875);
}
\qquad\qquad\qquad\qquad\qquad
\tikzmath{
\draw[thick, red, step=.75, xshift=-.3cm, yshift=-.3cm] (1.75,0.25) grid (3.5,3.5);
\draw[thick, red, step=.75, knot, xshift=-.15cm, yshift=-.15cm] (1.75,0.25) grid (3.5,3.5);
\draw[thick, step=.75, knot] (1.75,0.25) grid (3.5,3.5);
\filldraw[draw=cyan, thick, fill=cyan!30, rounded corners=5pt]
(1.75,0) --  (2.1,.2) -- (2.1,3.6) -- (1.4,3.2) -- (1.4,-.2) -- (1.75,0);
\draw[thick, red, step=.75, knot, xshift=-.3cm, yshift=-.3cm] (0.25,0.25) grid (2,3.5);
\draw[thick, red, step=.75, knot, xshift=-.15cm, yshift=-.15cm] (0.25,0.25) grid (2,3.5);
\draw[thick, step=.75, knot] (0.25,0.25) grid (2,3.5);
\foreach \x in {.75,1.5,3}{
\foreach \y in {.75,1.5,2.25,3}{
\draw[thick, red] (\x,\y) -- ($ (\x,\y) + (-.45,-.45) $);
}}
\foreach \y in {.75,1.5,2.25,3}{
\draw[thick, red] (2.25,\y) -- ($ (2.25,\y) + (-.15,-.15) $);
}
\node[cyan] at (2,3.7) {$\scriptstyle\cK$};
}
$$
With a physical boundary we will always mean a boundary obtained by cutting the system (with or without topological boundary) in half.
To avoid the awkward terminology of `boundary LTO boundary algebra', we use a superscript $\partial$ to distinguish it from the LTO boundary algebras studied in~\cite{MR4945955} whenever this is not clear from the context.
On this algebra $\fB^\partial$ we define a boundary DHR category $\text{DHR}^{\partial}$ meant to capture the structure of topological boundary excitations and their relation to the bulk. 
Unlike the category of DHR bimodules of the LTO boundary algebra, $\text{DHR}^\partial$ does not have a natural braiding in (2+1)D, but it will have a braiding in ($n$+1)D when $n>2$.

We illustrate the construction with two classes of examples.
Our first family of examples consists of 1D topological boundaries for 2D Levin-Wen models associated to a unitary fusion category (UFC) $\cC$ and a unitary module category $\cM$ (pictured left above) \cite{MR2942952}.
The second family of main examples comprises 2D topological boundaries for the 3D Walker-Wang model associated to a unitary braided fusion category (UBFC) $\cA$ and a unitary $\cA$-enriched UFC $\cX$ \cite{MR4640433,2305.14068} (pictured right above).

Our first result analyzes the boundary LTOs of Levin-Wen models. 
For a Levin-Wen model built from a string-net for some UFC $\cC$, we study the case of standard topological boundaries given microscopically by a $\cC$-module category $\cM_\cC$~\cite{MR2942952}. 
In this case, the boundary algebras associated to the bulk system are fusion spin chains, and the boundary algebras associated to the bulk-boundary system are natural `finite index' (in the sense of \cite{MR0696688}) extensions of half-infinite fusion spin chains called \emph{fusion module spin chains}, which is precisely the setting of classical subfactor theory \cite{MR1473221,MR1642584} (see also \cite{MR4227743} for how these fusion module algebras arise). 
We may thus apply existing operator algebraic machinery to deduce that the boundary DHR category is the dual category $\End(\cM_\cC)$ as expected. 

To state this result as a formal theorem, let $\cC$ be a UFC and $\cM_\cC$ an indecomposable $\cC$-module category. 
We have a natural boundary LTO on the 2D square lattice whose bulk LTO is a version of the Levin-Wen model studied in \cite{PhysRevB.71.045110}, and whose boundary projections are defined in terms of $\cM$ as in \cite{MR2942952}. 
We denote the (physical) boundary algebra of this bulk-boundary system by $\fM^\partial$. 
Using standard subfactor techniques (which we summarize in Appendix \ref{appendix:AFActionsOfUmFCs}), we demonstrate the following in~\S\ref{sec:BoundaryLTOforBoundaryLW}.

\begin{thmalpha}
The Levin-Wen bulk-boundary system satisfies the boundary LTO axioms.
The boundary algebra $\mathfrak{M}^\partial$ is a fusion module spin chain, and its boundary DHR bimodules are given by the UFC $\DHR^{\partial}(\mathfrak{M}^\partial)\cong \End(\cM_\cC)$.
\end{thmalpha}

Our next main topic is the analysis of (3+1)D Walker-Wang LTOs and their boundary LTOs, providing the first worked example of the LTO axioms beyond (2+1)D.
Associated to a unitary \emph{braided} fusion category (UBFC) $\cB$, the Walker-Wang model is a (3+1)D string-net model which is easily seen to satisfy the LTO axioms (see Theorem \ref{thm:LTQO-WW}).
For (2+1)D systems, the boundary algebras are generally given in terms of fusion spin chains, which can be defined straightforwardly in terms of a unitary fusion category because the 1D boundary has a natural ordering of the points, giving a natural order to take tensor products in the fusion category.
For the 2D boundary of the $\cB$ Walker-Wang model, this is no longer the case.
Using the braiding on the UBFC $\cB$ and a (strong) tensor generating object $X\in \mathcal{B}$, we define a two-dimensional version of a fusion spin chain which we call a \textit{braided categorical net}.
The boundary algebra for the $\cB$ Walker-Wang model can naturally be identified with such a braided categorical net, where $X$ is the sum over simples.

Braided categorical nets are interesting mathematical objects in their own right. 
For example, when a braided categorical net is built from a Drinfeld center $Z(\cC)$, which is a UBFC if $\cC$ is UFC, the net of algebras is \emph{bounded spread isomorphic} to the net of local operators  in the `flux sector' of a Levin-Wen model \cite{MR4808260} (see also \S\ref{sec:BoundedSpreadIso} below), which is purely 2D. 
Transporting the ground state of the Levin-Wen model across this isomorphism yields an interesting topological boundary state on the $Z(\cC)$ braided categorial net, which is generalizable to arbitrary UBFCs $\mathcal{B}$. 
This gives a new example of a \textit{duality} for 2D abstract quantum spin systems \cite{MR5041853,MR5011952}.

Given such a braided categorical net in 2D together with a distinguished state, one can try to study its superselection theory in the Doplicher--Haag--Roberts approach~\cite{MR2804555,MR3617688,MR4362722}.
Of central importance are the \textit{cone algebras}, the von Neumann algebras generated (in the reference state) by observables localized in a cone-like region.
The braided categorical nets as obtained above provide an interesting new feature in comparison with the nets of type $\rm II$ and type $\rm III$ cone von Neumann algebras from Levin-Wen models~\cite{MR4945955} or from Quantum Double models~\cite{MR4721705}.

\begin{thmalpha}
The cone von Neumann algebras of the $\cB$ braided categorical net in this topological boundary state are type $\rmI$ von Neumann algebras with finite dimensional centers.
\end{thmalpha}

The extension of superselection theory to the non-factorial setting appears in \cite{MR4927814}.
In contrast to the factorial case, where type I algebras lead to a trivial superselection theory~\cite{MR4426734}, the braided categorical nets with their natural states described above do have non-trivial sectors.
To our knowledge, this gives the first examples of nets of type $\rmI$ von Neumann algebras with interesting superselection theory.

\begin{thmalpha}
The $\rmW^*$-category\footnote{We do not analyze the monoidal or braided structures on this category.
These were recently shown to come from those of $\mathcal{B}$ in \cite{2609.20725}.} 
of superselection sectors of the $\cB$ braided categorical net is unitarily equivalent to $\Hilb(\cB)$, the unitary completion of $\cB$ in the sense of \cite{2411.01678} (see also \cite{MR3509018,MR3687214}).
\end{thmalpha}
The proof of these two theorems can be found in~\S\ref{sec:SSSofBraidedCategoricalNet}.

Returning to our main line of investigation, the universality classes of topological boundaries of $\mathcal{B}$ Walker-Wang models correspond to $\mathcal{B}$-enriched fusion categories \cite{MR4640433}.
A $\cB$-enriched fusion category is a UFC $\cX$ equipped with a braided central functor from $\mathcal{B}$ to $Z(\cX)$ \cite{MR3866911,MR3578212,MR3961709}.
Microscopic models for these boundaries were introduced in \cite{MR4640433,2305.14068}. 
We show in \S\ref{sec:WW} that these bulk-boundary models are boundary LTOs, and that their physical boundary algebras assemble into a braided categorical net $\fX^\partial$. 

\begin{thmalpha}
The Walker-Wang bulk-boundary system for $\cB\to Z(\cX)$ is a boundary LTO.
The boundary algebra $\fX^\partial$ is a braided categorical net, and
its category of boundary DHR bimodules $\DHR^{\partial}(\mathfrak{X}^\partial)$ is unitarily braided equivalent to the \emph{enriched center} $Z^\cB(\cX)$, the M\"uger centralizer~\cite{MR1990929} of $\cB$ in $Z(\cX)$. 
\end{thmalpha} 

Our analysis is performed using the folding trick, leading to a new inclusion of von Neumann algebras which could reasonably be called the \emph{enriched symmetric enveloping algebra} of an `enriched subfactor' in the spirit of \cite{MR1302385}.
As a corollary, we also obtain a characterization of the \emph{bulk} DHR bimodules for Walker-Wang models.

\begin{thmalpha}
The $\cB$ Walker-Wang model is an LTO, and the DHR category $\DHR(\fB)$ of its boundary algebra $\fB$ is unitarily braided equivalent to $Z_2(\cB)$, the M\"uger center of $\cB$. 
\end{thmalpha}

The physical interpretation of this result in the context of classifying topological order is as follows.
The main application of DHR bimodules for (physical) boundary algebras for LTOs in 2D is to give a model-independent way to recover the bulk topological order.
We again find that the boundary DHR bimodules for the $\cX$ topological boundary for the $\cB$ Walker-Wang model
recover the topological boundary excitations.
However, for the (3+1)D $\cB$ Walker-Wang model alone, the DHR bimodules fall short in this regard.
Indeed, the topological order of a (3+1)D system is a braided fusion 2-category, whereas the DHR bimodules merely form a braided (and in (3+1)D, \emph{symmetric}) fusion category. 
Our bulk DHR only recovers information concerning \emph{point defects} of the underlying topological quantum field theory, missing the line defects. 
However, if we truncate the boundary algebra to obtain a bulk-boundary system for a fusion spin chain, our result says the category of boundary DHR bimodules is $\cB$ as a braided fusion category. 
From this, we can take $Z(\Mod(\cB))$ to again describe the entire bulk topological order. 
However, this description is somewhat unsatisfactory, and it would be desirable to extend the notion of DHR bimodule to directly construct the braided fusion 2-category of all low energy excitations.

We recall the basic setting and main results of~\cite{MR4945955} in \S\ref{sec:Background}, before we introduce the notion of boundary LTO and the category of boundary DHR bimodules.
These notions are analyzed in \S\ref{sec:BoundaryLW} for the Levin-Wen model with gapped boundary.
Braided categorical nets are defined (and their superselection sectors identified) in~\S\ref{sec:BraidedCategoricalNets}. 
Finally, we consider the (3+1)D Walker-Wang model in~\S\ref{sec:WW} and identify its boundary algebras and corresponding (boundary) DHR bimodules.
Some useful techniques from subfactor theory are collected in Appendix~\ref{appendix:AFActionsOfUmFCs}.

\section{Background}
\label{sec:Background}

\subsection{Nets of algebras and local topological order axioms}
\label{sec:LTO}
We consider \emph{abstract quantum spin systems} given by a (local) net of algebras defined on some discrete set $\cL$.
In our applications, typically $\cL = \bbZ^d$.

\begin{defn}[Net of algebras]
Let $\cL$ be a discrete set and $\fA$ a unital $\rmC^*$-algebra.
A \emph{net of algebras} is a mapping which assigns a unital $\rmC^*$-subalgebra $\fA(\Lambda) \subset \fA$ to each finite subset\footnote{It is enough to only specify $\fA(\Lambda)$ for ``nice'' regions $\Lambda$, e.g. all rectangles.} $\Lambda \subset \cL$ such that
\begin{enumerate}
    \item $\fA(\emptyset) = \bbC 1_\fA$;
    \item (Isotony) $\fA(\Lambda) \subset \fA(\Delta)$ if $\Lambda \subset \Delta$;
    \item (Locality) $[\fA(\Lambda), \fA(\Delta)] = 0$ if $\Lambda \cap \Delta = \emptyset$;
    \item (Quasi-locality/density) $\bigcup_\Lambda \fA(\Lambda)$ is norm-dense in $\fA$.
\end{enumerate}
\end{defn}
We will also refer to $\fA$ as the \emph{quasi-local algebra}.
Physically, it describes those operators that can be approximated arbitrarily well (in norm) by strictly local operators.
For any (not-necessarily finite) $\Lambda \subset \cL$ we can define $\fA(\Lambda)$ as the norm-closure of $\bigcup_{\Delta \subset \Lambda} \fA(\Lambda)$, where the union is over all finite subsets contained in $\Lambda$, or in other words, $\fA(\Lambda) := \varinjlim_{\Delta \subset \Lambda} \fA(\Delta)$.
Note that $\fA(\cL) = \fA$.

For a net of algebras $\fA$ on a $\bbZ^d$ lattice, a net of projections is an assignment of a projection $p_\Lambda\in \fA(\Lambda)$ for every bounded rectangle $\Lambda$.
We furthermore demand that $\Lambda \mapsto p_\Lambda$ is ordered by reverse inclusion, i.e. $p_\Delta \leq p_\Lambda$ if $\Lambda \subset \Delta$.
In our examples there always is a natural translation symmetry on $\cL$, inducing a translation symmetry $\tau_x(\fA(\Lambda)) = \fA(\Lambda+x)$ on $\fA$, and we will assume that $p_\Lambda$ is translation invariant.

We say a rectangle $\Lambda$ is 
\begin{itemize}
\item
\emph{completely surrounded} by another rectangle $\Delta$ by \emph{surrounding constant} $s>0$, denoted $\Lambda \ll_s\Delta$, if $\Lambda\subset \Delta$, and the $\ell^\infty$ ball of radius $s$ about every vertex $v\in\Lambda$ (a rectangle centered at $v$ with side lengths $2s$) fits inside $\Delta$, 
and
\item
\emph{surrounded} by another rectangle $\Delta$ by \emph{surrounding constant} $s>0$, denoted  $\Lambda \Subset_s \Delta$
if $\Lambda\subset \Delta$, the intersection $\partial \Lambda\cap \partial \Delta$ occurs in a hyperplane of codimension one along a standard axis in our $\bbZ^d$ lattice, and for all other points $v\in \Delta \setminus \Lambda$, an $\ell^\infty$ ball of radius $s$ containing $v$ fits inside $\Delta$.
\end{itemize}
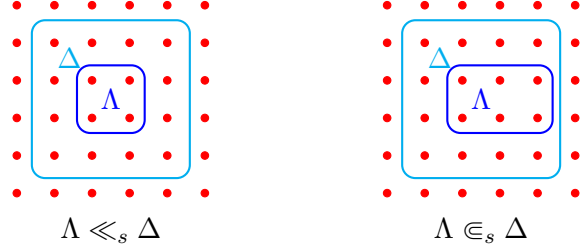
$$
\tikzmath{
\foreach \y in {0,.5,...,2.5}{
\foreach \x in {0,.5,...,2.5}{
\filldraw[red] (\x,\y) circle (.05cm);
}
}
\draw[thick, blue, rounded corners=5pt] (.8,.8) rectangle (1.7,1.7);
\draw[thick, cyan, rounded corners=5pt] (.2,.2) rectangle (2.3,2.3);
\node[cyan] at (.7,1.8) {$\Delta$};
\node[blue] at (1.25,1.25) {$\Lambda$};
\node at (1.25,-.5) {$\Lambda \ll_s \Delta$};
}
\qquad\qquad\qquad
\tikzmath{
\foreach \y in {0,.5,...,2.5}{
\foreach \x in {0,.5,...,2.5}{
\filldraw[red] (\x,\y) circle (.05cm);
}
}
\draw[thick, blue, rounded corners=5pt] (.8,.8) rectangle (2.2,1.7);
\draw[thick, cyan, rounded corners=5pt] (.2,.2) rectangle (2.3,2.3);
\node[cyan] at (.7,1.8) {$\Delta$};
\node[blue] at (1.25,1.25) {$\Lambda$};
\node at (1.25,-.5) {$\Lambda \Subset_s \Delta$};
}
$$
Given $\Lambda \Subset_s \Delta$, we define a unital $*$-algebra of `boundary error operators' by
$$
\fB(\Lambda \Subset_s \Delta)
:=
\set{xp_\Delta}{x\in p_{\Lambda} \fA(\Lambda)p_\Lambda\text{ and }xp_{\Delta'}=p_{\Delta'}x\text{ whenever }\Lambda \Subset_s \Delta'\text{ with }\partial \Lambda \cap \partial \Delta' = \partial \Lambda \cap \partial \Delta}.
$$
The local topological order axioms \cite{MR4945955} for a net of algebras $\fA$ equipped with a net of projections $(p_\Lambda)$ are as follows.
There is a `sufficient largeness' constant $r>0$ and a 
`surrounding constant' $s>0$ such that
\begin{enumerate}[label=\textup{(LTO\arabic*)}]
\item 
\label{LTO:QECC}
Whenever $\Lambda \ll_s \Delta$, $p_\Delta \fA(\Lambda)p_\Delta = \bbC p_\Delta$.
\item 
\label{LTO:Boundary}
Whenever $\Lambda \Subset_s \Delta$, $p_\Delta \fA(\Lambda)p_\Delta = \fB(\Lambda \Subset_s \Delta) p_\Delta$.
\item
\label{LTO:Surjective}
Whenever
$\Lambda_1\subset  \Lambda_2\Subset_s \Delta$ with 
$\partial \Lambda_1 \cap \partial\Delta= \partial \Lambda_2 \cap \partial \Delta$,
$\fB(\Lambda_1 \Subset_s \Delta)=\fB(\Lambda_2 \Subset_s \Delta)$.
\item
\label{LTO:Injective}
Whenever $\Lambda \Subset_s \Delta_1\subset \Delta_2$ with 
$\partial \Lambda \cap \partial\Delta_1= \partial \Lambda \cap \partial \Delta_2$,
the map $x\mapsto xp_{\Delta_2}$ is injective.
\end{enumerate}
Here all rectangles are assumed to have diameter at least $r$.
The constants $r$ and $s$ depend on the specific model under consideration.
One can define the LTO axioms for more general regions (and show that they hold for e.g. Levin--Wen models)~\cite{2605.10693}, but for simplicity we will restrict to rectangular regions.

Axioms \ref{LTO:Surjective} and \ref{LTO:Injective} ensure that the boundary error operators $\fB(\Lambda \Subset_s \Delta)$ only depend on the intersection $\partial\Lambda \cap \partial \Delta$ and not on the choices of $\Lambda,\Delta$ beyond them being sufficiently large with $\Lambda \Subset_s \Delta$.

We can think of our axioms in terms of error correcting codes.
The first axiom is related to the LTQO conditions~\cite{MR2742836,MR2842961} and implies that for tensor product nets of algebras, the projections associated to sufficiently large rectangles project onto a code space that corrects against errors localized \emph{away} from the boundary of the rectangle. 
The second axiom ensures all localized errors whose supports only meet the boundary in a small region are correctable to errors localized near where the support meets the boundary. 
The algebras $\fB(\Lambda \Subset_s \Delta)$ collect these local error operators into a net of algebras on a rectangle lattice one dimension smaller, as we recall below.
The other two axioms ensure that the resulting nets have nice properties (e.g., being local in a suitable sense), which we will explain below.
The operators in $\fB(\Lambda \Subset_s \Delta)$ are related to the `patch operators' in \cite{PhysRevB.107.155136,2310.05790} and our motivation to study them is the same, namely the holographic bulk-boundary duality.

\begin{ex}
Kitaev's Toric Code \cite{MR1611329} was shown to satisfy \ref{LTO:QECC}--\ref{LTO:Injective} in~\cite[\S 3]{MR4945955}.
\end{ex}

\begin{ex}
Kitaev's Quantum Double model \cite{MR1951039} was shown to satisfy \ref{LTO:QECC}--\ref{LTO:Injective} in~\cite[\S III]{MR4814524}.
\end{ex}

\begin{ex}
The unitary tensor category version of the Levin-Wen model from \cite{MR3204497,2305.14068} 
was shown to satisfy \ref{LTO:QECC}--\ref{LTO:Injective} in \cite[\S 4]{MR4945955}.
\end{ex}

\subsection{The canonical state and the net of boundary error operators}
\label{sec:CanonicalStateAndBoundaryAlgebra}

By \ref{LTO:QECC}, we get a canonical state $\psi$ on the quasi-local algebra $\fA$ determined by taking norm limits under the formula
$$
p_\Delta xp_\Delta = \psi(x) p_\Delta
\qquad\qquad\qquad
x\in \fA(\Lambda)\text{ and }\Lambda\ll_s\Delta.
$$
It was shown in \cite[Cor.~2.24]{MR4945955} that if $\Delta \mapsto p_\Delta$ is translation invariant, $\psi$ is the unique translation invariant state satisfying $\psi(p_\Lambda)=1$ for all $\Lambda$, and thus $\psi$ is pure.

Now given a choice of codimension 1 sub-lattice $\cK\subset \cL$ and a distinguished choice of half-space $\bbH$ of $\cL$ to one side of $\cK$, the  axioms \ref{LTO:QECC}--\ref{LTO:Injective} construct a net of algebras $\fB$ on $\cK$ called the \emph{boundary algebra} together with a quantum channel $\bbE: \fA({\bbH}) \to \fB$.

The main idea is that given $\Lambda\Subset_s \Delta\subset \bbH$ with $s$ sufficiently large
and $I:=\partial \Lambda\cap \partial\Delta\subset \cK$,
the algebra
$
\fB(\Lambda\Subset_s \Delta)
$
only depends on the boundary rectangle $I$ and sites nearby in $\Lambda$, not the specific choices of $\Lambda$ and $\Delta$.
Indeed, choosing $\Lambda_I\subset \Lambda$ to be the smallest sufficiently large rectangle with $\partial \Lambda_I\cap \cK= I$
and $\Delta_I\subset \Delta$ to be the smallest rectangle with $\Lambda_I\Subset_s \Delta_I$ with $\partial \Lambda_I\cap \partial\Delta_I = I$,
we define the \emph{boundary algebra}
$$
\fB(I):= \fB(\Lambda_I\Subset_s \Delta_I).
$$
It was shown in~\cite[Construction 2.28]{MR4945955} that $I\mapsto \fB(I)$ defines a net of algebras on $\cK$.

We now observe that we get a quantum channel $\bbE$ from 
$\fA(\bbH)$ to $\fB=\varinjlim \fB(I)$ 
by taking norm limits under the formula
$$
p_\Delta x p_\Delta = \bbE(x)p_\Delta
\qquad\qquad\qquad
x\in \fA(\Lambda),\,\,
\Lambda\Subset_s \Delta\subset \bbH,\text{ and }
\partial \Lambda\cap \partial\Delta\subset \cK.
$$
Moreover, for any $x\in \fA({\bbH})$ supported sufficiently far away from $\cK$, $\bbE(x)=\psi(x)$.
We thus get a \emph{state-based approach} to boundary theories, where any state $\phi$ on $\fB$ can be uniquely extended to $\fA(\bbH)$ by $\phi\circ \bbE$, and $\phi\circ \bbE$ will still restrict to $\psi$ sufficiently far away from $\cK$.

\begin{ex}
\label{ex:fusion_spin}
In \cite[Thm.~4.8 and Rem.~4.9]{MR4945955}, it was shown that the boundary algebra for the unitary tensor category version of the Levin-Wen model from \cite{MR3204497,2305.14068} 
is a \emph{fusion spin chain}.

In more detail, consider a unitary fusion category (UFC) $\cC$ equipped with a \emph{strong tensor generator} $X$, meaning that there exists an $n$ such that every simple object in $\cC$ is isomorphic to a summand of $X^{\otimes n}$.
For a finite interval $I \subset \bbZ$, we define the finite dimensional $\rmC^*$-algebra $\fF(I) := \End_\cC(X^{\otimes |I|})$.
If $I \subset J$ there is an obvious inclusion of the corresponding algebras, by tensoring with copies of $\id_X$.
Note that this is compatible with the `geometric' locality of $\bbZ$, and hence we get a net of algebras $\fF$ called the \emph{fusion spin chain} associated to $(\cC,X)$.

One can then build a Levin-Wen model on a $\bbZ^2$ lattice where the local Hilbert spaces are given by
$$
\cC(X\otimes X \to X\otimes X)
$$
equipped with the skein module inner product (see \S\ref{sec:BoundaryLW} below).
The boundary algebra along a 1D hyperplane will then be bounded spread isomorphic to the fusion spin chain $\fF$, which can be seen using a skein module argument as in \cite[Thm.~4.8]{MR4945955}.
Indeed, by semisimplicity, whenever $|I|\geq n$, $\fB(I)=\fF(I)$.
We refer the reader to \S\ref{sec:BoundaryLW} and Lemma \ref{lem:SemisimpleCommutation} below for more details.
\end{ex}

\subsection{DHR bimodules}
DHR bimodules where introduced in~\cite{MR4814692} as an analog of the DHR superselection theory~\cite{MR1405610} for abstract spin systems.
In~\cite{MR4945955} it was shown that the DHR bimodules of a `boundary net' can be thought of as being holographically dual to the bulk topological order in Levin-Wen models or the toric code.
We recall the main ideas here, giving a fairly complete sketch of the main argument.

In this section, $A,B$ will denote unital $\rmC^*$-algebras.

\begin{defn}
A right Hilbert $\rmC^*$ $B$-module $Y_B$ is a right $B$-module equipped with a right $B$-valued inner product.
Moreover, we require that $Y$ is complete with respect to the norm $\|\xi\|:=\|\langle \xi|\xi\rangle_B\|_B^{1/2}$.

A bounded operator $T:Y_B\to Y_B$ is called \emph{adjointable} if there is another operator $T^\dag:Y_B\to Y_B$ called the \emph{adjoint} such that
$\langle \eta|T\xi\rangle_B = \langle T^\dag \eta|\xi\rangle_B$ for all $\eta,\xi\in Y_B$.
Adjoints are unique when they exist, and they are automatically bounded by the Closed Graph Theorem.
We denote the space of adjointable operators by $\End(Y_B)$, which is again a $\rmC^*$-algebra.

A Hilbert $\rmC^*$ $A-B$ bimodule (sometimes called a \emph{correspondence}) is a right Hilbert module $Y_B$ together with a left $A$-action by adjointable operators.
We denote the category of $A-B$ bimodules by $\Bim(A,B)$.
We refer the reader to \cite[\S8]{MR2111973} (see also \cite{MR355613,MR0367670} and \cite[\S2]{MR4419534}) for the definition of the Connes fusion relative tensor product for bimodules, which we will denote by $\boxtimes$, and further details.
This makes $\Bim(A)$ into a $\rm{C}^*$ tensor category.
\end{defn}

\begin{defn}
Suppose $Y_B$ is a right Hilbert $B$-module.
Each $\eta\in Y$ gives a right $B$-linear creation operator $L_\eta: B\to Y$ given by $b\mapsto \eta b$, which is adjointable with $L_\eta^\dag\xi=\langle \eta|\xi\rangle_B$.

A right Hilbert $B$-module $Y_B$ is called \emph{finitely generated projective (fgp)} if there is a finite set $\{\beta_i\}_{i=1}^n\subset Y_B$ called a $Y_B$-\emph{basis} such that $\sum_iL_{\beta_i}L_{\beta_i}^\dag=1$.
Such a $Y_B$-basis gives an explicit isomorphism $Y_B \cong pB^n_B$ where
$p\in M_n(B)$ is the orthogonal projection whose entries are given by
$$
p_{ij} :=\langle \beta_i|\beta_j\rangle_B.
$$
Moreover, if ${}_AY_B$ carries a left $A$-action by adjointable operators, the corresponding left $A$-action on $B^n$ given by
$$
(a\cdot (b_j))_i
:= 
\sum_j \langle \beta_i | a\beta_j\rangle_B b_j
$$
promotes the above right $B$-module isomorphism to an $A-B$ bimodule isomorphism.
\end{defn}

\begin{defn}
Suppose $\fA$ is a net of algebras on a lattice $\cL$ and ${}_\fA Y_\fA\in \Bim(\fA)$.
We say that a vector $\xi\in Y$ is \emph{localized} in a rectangle $\Lambda$ if 
$$
a\xi = \xi a
\qquad\qquad\qquad
\forall\, a\in \fA(\Delta),
\quad
\forall\, \Delta\subset\Lambda^c.
$$
\end{defn}

Recall that if $B\subset A$ is a unital inclusion of $\rmC^*$-algebras, the \emph{centralizer} of $B\subset A$ is 
$$
Z_A(B):=\set{a\in A}{ab=ba\text{ for all }b\in B}.
$$

\begin{lem}
\label{lem:IPLocalizedInCentralizer}
If $\eta,\xi\in Y\in \Bim(\fA)$ are both localized in $\Lambda$ and $a\in \fA(\Lambda)$, then
$\langle \eta|a\xi\rangle_{\fA} \in Z_{\fA}(\fA(\Delta))$ for all $\Delta\subset \Lambda^c$.
\end{lem}
\begin{proof}
For all $x\in \fA(\Delta)$,
$
x\langle \eta|a\xi\rangle_{\fA}
=
\langle \eta x^*|a\xi\rangle_{\fA}
=
\langle x^*\eta|a\xi\rangle_{\fA}
=
\langle \eta|xa\xi\rangle_{\fA}
=
\langle \eta|a\xi x\rangle_{\fA}
=
\langle \eta|a\xi \rangle_{\fA}x
$.
\end{proof}

If we know our net of algebras satisfies algebraic Haag duality (see the next definition), the above lemma allows us to conclude that $\langle \eta|\xi\rangle_{\fA}$ actually lies in the algebra $\fA(\Lambda)$.

\begin{defn}[\cite{MR4814692}]
A net of algebras $\fA$ satisfies \emph{algebraic Haag duality}
if for every rectangle $\Lambda$, 
$Z_{\fA}(\fA(\Lambda^c))=\fA(\Lambda)$, where $\fA(\Lambda^c)\subset \fA$ is the $\rmC^*$-algebra generated by all the $\fA(\Delta)$ such that $\Delta\subset \Lambda^c$. 
\end{defn}

\begin{defn}[\cite{MR4814692}]
\label{def:LocalizableBimodule}
Suppose $\fA$ is a net of algebras on our lattice $\cL$
and ${}_\fA Y_\fA\in \Bim(\fA)$ is fgp.
For a rectangle $\Lambda\subset \cL$, we say that ${}_\fA Y_\fA$ is \emph{localizable in $\Lambda$}
if there is a $Y_\fA$ basis localized in $\Lambda$.

We say ${}_\fA Y_\fA$ is \emph{localizable} or a \emph{DHR bimodule} if for any $\Lambda$ sufficiently large, ${}_\fA Y_\fA$ is localizable in $\Lambda$.
Observe that DHR bimodules are closed under composition of bimodules, and thus form a full $\rmC^*$ tensor subcategory of $\Bim(\fA)$, denoted $\DHR(\fA)$.
\end{defn}

One uses this localizability of DHR bimodules to define a unitary braiding for $\DHR(\fA)$ similar to the strategy in \cite{MR4753059}.

\begin{defn}
Suppose that $\fA$ satisfies algebraic Haag duality.
Suppose $Y,Z\in \DHR(\fA)$.
Choose projective bases $\{b_{i}\}, \{c_{j}\}$ for $Y,Z$ respectively which are localized in sufficiently large rectangles $\Lambda,\Delta$ which are sufficiently far apart in $\cL$ (see \cite{MR4814692} for the precise details). 
We define 
$$
u^{\Lambda,\Delta}_{Y,Z}: Y\boxtimes_\fA Z\rightarrow Z\boxtimes_\fA Y
\qquad\qquad\text{by}\qquad\qquad
\sum_{i,j}b_{i}\boxtimes c_{j}a_{ij}
\mapsto
\sum_{i,j}c_{j}\boxtimes b_{i} a_{ij}.
$$
One verifies this formula is a well-defined unitary $\fA-\fA$ bimodule map which is independent of the choices of $\Lambda,\Delta$ (provided they are sufficiently large and separated), and also independent of the choices of projective bases.
(When our lattice $\cL$ is 1-dimensional, we must require that our two rectangles are always linearly ordered in the same way, or we may accidentally get the reverse braiding.)
\end{defn}

To calculate $\DHR(\fA)$ and its braiding in practice, we may use the following equivalent characterization of a DHR bimodule (for a similar analysis, see \cite{2504.06094}).

\begin{lem}[Local triviality of DHR bimodules]
\label{lem:DHRBimoduleCharacterization}
A bimodule ${}_\fA Y_\fA\in \Bim(\fA)$ lies in $\DHR(\fA)$ if and only if for any sufficiently large $\Lambda$, the restriction
$
{}_{\fA(\Lambda^c)} Y_\fA
$
is isomorphic to a summand of a finite direct sum of copies of ${}_{\fA(\Lambda^c)}\fA_\fA$.
\end{lem}
\begin{proof}
Suppose ${}_\fA Y_\fA\in\DHR(\fA)$ and $\{\beta_i\}_{i=1}^n$ is a $Y_\fA$-basis localized in $\Lambda$, i.e., $x\beta_i=\beta_ix$ for all $x\in \fA(\Lambda^c)$.
This condition is identical to saying that the creation operator $L_{\beta_i}: \fA_\fA\to Y_\fA$ is left $\fA(\Lambda^c)$-linear.
Its adjoint $L_{\beta_i}^\dag:Y_\fA\to \fA_\fA$ is also left 
$\fA(\Lambda^c)$-linear, and thus the composite 
$$
L^\dag_{\beta_i}L_{\beta_j} 
\in 
\fA(\Lambda^c)'\cap \fA
:=
\set{a\in \fA}{ax=xa\text{ for all }x\in \fA(\Lambda^c)}.
$$
This implies that the projector $p\in \End(\fA^{\oplus n}_\fA)$ given by $p_{ij}=\langle \beta_i|\beta_j\rangle_\fA$ whose image is isomorphic to $Y_\fA$
satisfies 
$$
p\in \fA(\Lambda^c)'\cap M_n(\fA)
=
\End\left({}_{\fA(\Lambda^c)}\fA^{\oplus n}_\fA\right).
$$
Thus the restriction
$
{}_{\fA(\Lambda^c)} Y_\fA
$
is isomorphic to a summand of a finite direct sum of copies of ${}_{\fA(\Lambda^c)}\fA_\fA$.

Conversely, suppose that for any sufficiently large $\Lambda$, ${}_{\fA(\Lambda^c)}Y_\fA$ is isomorphic to a summand of a finite direct sum of copies of ${}_{\fA(\Lambda^c)} \fA_\fA$.
This exactly says we have a projector $p\in \End(\fA^{\oplus n}_\fA)$ which commutes with the left $\fA(\Lambda^c)$-action, and we may identify $Y\cong p\fA^{\oplus n}$.
Set $\beta_i:= pe_i$ where $e_i$ has $i$-th entry 1 and all other entries zero.
Then clearly $\{\beta_i\}$ is a $p\fA^{\oplus n}_\fA$-basis such that $x\beta_i=\beta_i x$ for all $x\in \fA(\Lambda^c)$.
Thus $Y$ is localizable in every sufficiently large $\Lambda$.
\end{proof}

We may use the characterization of DHR bimodules from Lemma~\ref{lem:DHRBimoduleCharacterization} to compute $\DHR(\fA)$ if we understand how such bimodules restrict to $\fA(\Lambda^c)$-bimodules.
The following lemma which relies on dualizability is useful for exactly this type of computation.

\begin{lem}[2 out of 3]
\label{lem:2outof3}
Suppose $\cE\subset \Bim(A)$ is a fully faithfully embedded and replete unitary multifusion subcategory. 
If $X,Z\in \Bim(A)$ such that
two out of three of 
$X,Z,X\boxtimes_AZ$
lie in $\cE$, then the third also lies in $\cE$
\end{lem}
\begin{proof}
Clearly $X,Z\in\cE$ implies $X\boxtimes_A Z\in\cE$.
We prove $X,X\boxtimes_AZ\in\cE$ implies $Z\in\cE$, and the remaining case is similar.
Observe that ${}_AA_A$ is always a sub-bimodule of $X^\vee \boxtimes_A X\in\cE$.
Hence $Z$ is a sub-bimodule of
$X^\vee \boxtimes_A X\boxtimes_A Z\in\cE$.
Since $\cE$ is idempotent complete, $Z\in \cE$.
\end{proof}

\begin{ex}[{\cite[Thm.~C]{MR4814692}}]
\label{ex:fusionSpinChainDHR}
Consider the fusion spin chain $\fF$ build from a UFC with a strong tensor generator $X\in\cC$ (i.e., there is an $n\in\bbN$ such that each simple $c\in\Irr(\cC)$ is a subobject of $X^{\otimes n}$).
It is straightforward to construct a DHR bimodule $Y^z$ for every $z\in Z(\cC)$ via an inductive limit construction.
The main point is that we pick an interval $I$ long enough so that every simple of $\cC$ appears as a summand of $X^{\otimes |I|}$,
and then subfactor techniques 
for the inclusion of AF $\rmC^*$-algebras $\fF(I^c)\subset \fF$
can be used to construct DHR bimodules.
We will not explain this here, as it is entirely analogous to the module category version in \S\ref{sec:BoundaryDHRforBoundaryLW} below; instead we refer the reader to \cite[\S4]{MR4814692}.
It was then shown that $z\mapsto Y^z$ is a fully faithful braided tensor functor using the Ocneanu Compactness Theorem \cite{MR1473221} (see also \cite{MR4916103} and \cite{MR4923706}).
We give a generalization of this construction and the proof that the functor is fully faithful in Appendix~\ref{appendix:AFActionsOfUmFCs} below.

To see that this functor is essentially surjective, for an arbitrary $Y\in \DHR(\fF)$,
${}_{\fF(I^c)}Y_\fF$ is a summand of a direct sum of the trivial bimodule ${}_{\fF(I^c)}\fF_\fF$ by Lemma \ref{lem:DHRBimoduleCharacterization}.
We then restrict the right action to see that
${}_{\fF(I^c)}Y_{\fF(I^c)}$ is a summand of a direct sum of the trivial bimodule ${}_{\fF(I^c)}\fF_{\fF(I^c)}$.
These trivial bimodules lie in the essential image of $\cC\boxtimes \cC^{\rm mp}$ in $\Bim(\fF(I^c))$ by using the \emph{Q-system realization} technique from \cite{MR4419534}.
One then induces back to $\Bim(\fF)$ using Lemma \ref{lem:2outof3} to conclude that ${}_\fF Y_\fF$ is in the essential image of $Z(\cC)$.
Here, we have used the fact that one can embed the unitary multifusion category
$$
\cE = 
\begin{pmatrix}
\cC\boxtimes \cC^{\rm mp} & \cC
\\
\cC^{\rm op} & Z(\cC)
\end{pmatrix}
$$
fully faithfully into $\Bim(\fF(I^c)\oplus \fF)$.
Thus Lemma \ref{lem:2outof3} is applied twice with $A=\fF(I^c)\oplus \fF$.
\end{ex}

\subsection{Boundary LTO axioms for a net of algebras on a lattice with boundary}

We now consider a net of algebras $\fA$ on a lattice with boundary.
We consider the lattice $\cL=\bbZ^d\times \bbN$ 
where $\partial\cL:=\bbZ^d\times \{0\}$ is considered as a distinguished codimension 1 boundary.
We write $\cL^\circ:=\cL\setminus \partial\cL = \bbZ^d\times \{1,2,3,4,\dots\}$.

As above, our net of algebras $\fA$ on our lattice with boundary is equipped with a net of projections $\Lambda\mapsto p_\Lambda$ ordered by reverse inclusion.
For rectangles suitably far away from $\partial \cL$, we require the above \ref{LTO:QECC}--\ref{LTO:Injective} axioms for the bulk.
For rectangles along the boundary, we must now allow for surrounding rectangles to intersect on two adjacent sides, rather than along only one side.

\begin{defn}
We say a rectangle $\Lambda$ with $\partial\Lambda\cap \partial \cL \neq \emptyset$ is:
\begin{itemize}
\item 
\emph{completely} $\partial$-\emph{surrounded} by $\Delta$ with surrounding constant $s$, denoted $\Lambda \ll_s^\partial \Delta$, if
$\Lambda\subset \Delta$, $\partial\Lambda\cap \partial \Delta$ is a rectangle one dimension lower which lies in $\partial \cL$, and for all other points $v\in\Delta$ not in $\Lambda$, an $\ell^\infty$ ball of radius $s$ containing $v$ fits inside $\Delta$.

Alternatively, $\Lambda \ll_s^\partial \Delta$ if and only if $\Lambda \Subset_s \Delta$ with $\emptyset \neq\partial\Lambda\cap \partial\Delta \subset \partial \cL$.
\item 
$\partial$-\emph{surrounded} by $\Delta$ with surrounding constant $s$, denoted $\Lambda \Subset_s^\partial \Delta$, if $\Lambda\subset \Delta$, $\partial\Lambda\cap \partial \Delta$ is the union of two rectangles of one dimension lower that intersect each other, one of which lies in $\partial \cL$, and for all other points $v\in\Delta$ not in $\Lambda$, an $\ell^\infty$ ball of radius $s$ containing $v$ fits inside $\Delta$.
\end{itemize}
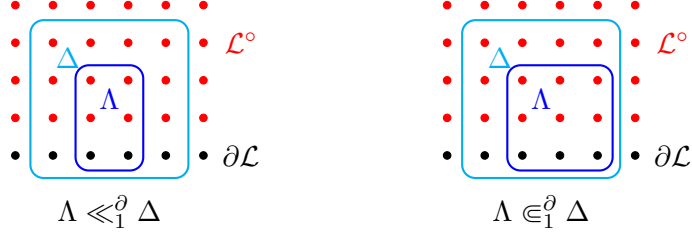
$$
\tikzmath{
\foreach \x in {0,.5,...,2.5}{
\filldraw (\x,.5) circle (.05cm);
}
\foreach \y in {1,1.5,...,2.5}{
\foreach \x in {0,.5,...,2.5}{
\filldraw[red] (\x,\y) circle (.05cm);
}
}
\draw[thick, blue, rounded corners=5pt] (.8,.3) rectangle (1.7,1.7);
\draw[thick, cyan, rounded corners=5pt] (.2,.2) rectangle (2.3,2.3);
\node[cyan] at (.7,1.8) {$\Delta$};
\node[blue] at (1.25,1.25) {$\Lambda$};
\node at (3,.5) {$\partial \cL$};
\node[red] at (3,2) {$\cL^\circ$};
\node at (1.25,-.25) {$\Lambda \ll_1^\partial \Delta$};
}
\qquad\qquad\qquad
\tikzmath{
\foreach \x in {0,.5,...,2.5}{
\filldraw (\x,.5) circle (.05cm);
}
\foreach \y in {1,1.5,...,2.5}{
\foreach \x in {0,.5,...,2.5}{
\filldraw[red] (\x,\y) circle (.05cm);
}
}
\draw[thick, blue, rounded corners=5pt] (.8,.3) rectangle (2.2,1.7);
\draw[thick, cyan, rounded corners=5pt] (.2,.2) rectangle (2.3,2.3);
\node[cyan] at (.7,1.8) {$\Delta$};
\node[blue] at (1.25,1.25) {$\Lambda$};
\node at (3,.5) {$\partial \cL$};
\node[red] at (3,2) {$\cL^\circ$};
\node at (1.25,-.25) {$\Lambda \Subset_1^\partial \Delta$};
}
$$
\end{defn}
For a rectangle $\Lambda$ intersecting $\partial\cL$ and a rectangle $\Delta$ with $\Lambda \Subset_s^\partial \Delta$, we introduce the `corner' boundary algebra
$\fB^\partial(\Lambda \Subset_s^\partial \Delta)$ defined exactly as in the case $\fB(\Lambda \Subset_s \Delta)$.
The following \emph{boundary LTO axioms} ensure that this algebra is independent of choices beyond the intersection of the rectangles (see~\S\ref{sec:BoundaryAlgforBoundaryLTO} below for an outline of the argument).
\begin{enumerate}[label=\textup{($\partial$LTO\arabic*)}]
\item 
\label{BoundaryLTO:QECC}
Whenever $\Lambda \ll_s^\partial \Delta$, $p_\Delta \fA(\Lambda)p_\Delta = \bbC p_\Delta$.
\item 
\label{BoundaryLTO:Boundary}
Whenever $\Lambda \Subset_s^\partial \Delta$, $p_\Delta \fA(\Lambda)p_\Delta = \fB^\partial(\Lambda \Subset_s^\partial \Delta) p_\Delta$.
\item
\label{BoundaryLTO:Surjective}
Whenever
$\Lambda_1\subset  \Lambda_2\Subset_s^\partial \Delta$ with 
$\partial \Lambda_1 \cap \partial\Delta = \partial \Lambda_2 \cap \partial \Delta$,
$\fB^\partial(\Lambda_1 \Subset_s^\partial \Delta)=\fB^\partial(\Lambda_2 \Subset_s^\partial \Delta)$.
\item
\label{BoundaryLTO:Injective}
Whenever $\Lambda \Subset_s^\partial \Delta_1\subset \Delta_2$ with 
$\partial \Lambda \cap \partial\Delta_1= \partial \Lambda \cap \partial \Delta_2$,
the map $x\mapsto xp_{\Delta_2}$ is injective.
\end{enumerate}

\begin{ex}
Kitaev's Toric Code with smooth boundary satisfies \ref{BoundaryLTO:QECC}--\ref{BoundaryLTO:Injective}.
In this example the degrees of freedom live on the edges instead of on $\mathbb{Z}^2$ vertices (cf. \S3.1 of~\cite{MR4945955}), as is the standard convention for the toric code.
Indeed, for rectangles $\Lambda$ such that $\partial \Lambda$ meets the smooth boundary, the projection $p_\Lambda$ is still given by
$$
p_\Lambda 
:= 
\prod_{s\subset \Lambda} \left(\frac{1+A_s}{2}\right)
\prod_{p\subset \Lambda} \left(\frac{1+B_p}{2}\right),
$$
for star and plaquette operators $A_s,B_p$ respectively.
However, we now have more terms $A_s$ and $B_p$ than before in our boundary Hamiltonian.
Indeed, every boundary star and plaquette
$$
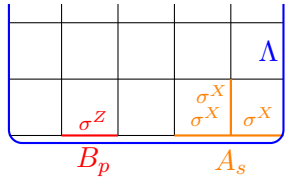
\tikzmath{
\foreach \y in {-.75,-1.5,-2.25}{
\draw (.05,\y) -- (3.7,\y);
}
\foreach \x in {.75,1.5,2.25,3}{
    \draw (\x,-.5) -- (\x,-2.25);
}
\draw[thick, red] (.75,-2.25) -- (1.5,-2.25);
\node[red] at (1.175,-2.05) {$\scriptstyle \sigma^Z$};
\node[red] at (1.175,-2.6) {$B_p$};
\draw[thick, orange] (3,-1.5) -- (3,-2.25);
\draw[thick, orange] (2.25,-2.25) -- (3.7,-2.25);
\node[orange] at (2.75,-1.7) {$\scriptstyle \sigma^X$};
\node[orange] at (2.675,-2) {$\scriptstyle \sigma^X$};
\node[orange] at (3.375,-2) {$\scriptstyle \sigma^X$};
\node[orange] at (3,-2.6) {$A_s$};
\draw[thick, blue, rounded corners=5pt] (3.7,-.5) -- (3.7,-2.35) -- (.05,-2.35) -- (.05,-.5);
\node[blue] at (3.5,-1.125) {$\Lambda$}; 
}
$$
is included in the boundary Hamiltonian, so the argument in \cite[\S 3]{MR4945955} based on \cite{MR2345476} now implies that \ref{BoundaryLTO:QECC} holds.
(See also \cite[Proof of Thm.~4.2]{MR4650344}.)
More precisely, one can consider a simple tensor $a$ of Pauli matrices, and consider $p_\Delta a p_\Delta$ for a suitable rectangle $\Delta$.
The $A_s$ and $B_p$ terms from the Hamiltonian can be freely absorbed into and created from $p_\Delta$.
Using the (anti-)commutation properties of the Pauli matrices, it follows that $p_\Delta a p_\Delta = 0$ unless $a$ is a product of $A_s$'s and $B_p$'s (from the bulk or boundary Hamiltonian).

The axioms \ref{BoundaryLTO:Boundary}--\ref{BoundaryLTO:Injective} can be similarly adapted from the proofs in \cite[\S 3]{MR4945955}.
The result depends on if the touching boundaries of $\Lambda$ and $\Delta$ in the bulk are along a `rough' or a `smooth' (physical) cut in the bulk.
For concreteness, we consider the rough case:
$$
\tikzmath{
\foreach \y in {1.5,0.75,0,-.75,-1.5,-2.25}{
\draw (-1.1,\y) -- (3.7,\y);
}
\foreach \x in {-0.75,0,.75,1.5,2.25,3}{
    \draw (\x,1.65) -- (\x,-2.25);
}
\node at (4.0,-2.25) {$\scriptstyle \partial \cL$};

\draw[thick, orange] (3,-1.5) -- (3.7,-1.5);
\node[orange] at (3.25,-1.25) {$\scriptstyle \sigma^X$};
\node[orange] at (4.0,-1.45) {$C_\ell$};
\draw[thick, red] (3.7,-0.75) -- (3.0,-0.75) -- (3,0) -- (3.7,0);
\node[red] at (4.0,-0.4) {$D_p$};
\node[red] at (3.25,-.5) {$\scriptstyle \sigma^Z$};
\node[red] at (3.25,0.25) {$\scriptstyle \sigma^Z$};
\node[red] at (2.75,-0.35) {$\scriptstyle \sigma^Z$};

\draw[thick, cyan, rounded corners=5pt] (-.85,-2.45) rectangle (3.6,1.6);
\node[cyan] at (-0.5,1.275) {$\Delta$}; 
\draw[thick, blue, rounded corners=5pt] (-.10,-2.35) rectangle (3.5,0.85);
\node[blue] at (0.5,0.475) {$\Lambda$}; 
}
$$
The operators $C_\ell$ and $D_p$ are partial star and plaquette operators as depicted above whose support is contained in $\Lambda$.
Write $I$ for the horizontal links along the rough cut, i.e.\ the \textit{right side} boundary of $\partial \Lambda \cap \partial \Delta$.
We denote $\widetilde{I}$ for the union of $I$ together with the vertical links directly to the left of it.
By again systematically checking commutation relations with Hamiltonian terms, we find that
\[
    \fB^\partial(\Lambda \Subset^\partial_s \Delta) \cong \mathrm{C}^*\left(C_\ell,\, D_p \, \middle| \, \ell \subset I \setminus \partial\cL,\, p \subset \widetilde{I} \right).
\]
Note that the horizontal link $\ell$ at the boundary $\partial\cL$ is excluded due to the boundary Hamiltonian term $\sigma^Z$.
The case of the smooth (physical) bulk cut is handled analogously.

Similarly, Kitaev's Toric Code with rough topological boundary also satisfies \ref{BoundaryLTO:QECC}--\ref{BoundaryLTO:Injective}.
In either case, one can define a net of algebras assigned to intervals of the half-line given by the cut $\cK$, including for intervals away from $\partial \cL$, by matching the `corner' boundary algebra with the `bulk' boundary algebra coming from the LTO axioms, cf. a similar construction for the Levin--Wen model below in~\S\ref{sec:BoundaryLTOforBoundaryLW}.
\end{ex}

\begin{ex}
We will see in \S\ref{sec:BoundaryLTOforBoundaryLW} below that the 1D boundary for the 2D Levin-Wen model for $\cC$ coming from an indecomposable module category $\cM_\cC$ satisfies  \ref{BoundaryLTO:QECC}--\ref{BoundaryLTO:Injective}.
\end{ex}

\subsection{Boundary algebras and DHR bimodules for boundary LTOs}
\label{sec:BoundaryAlgforBoundaryLTO}

We now extend the construction of the boundary net $\fB$ and the quantum channel $\bbE: \fA({\bbH}) \to \fB$ to boundary LTOs.
As before, we work with $\cL=\bbZ^d\times \bbN$ with boundary $\partial \cL = \bbZ^d\times \{0\}$.
We pick a distinguished codimension 1 sublattice $\cK\subset \cL$ by fixing one of the first $d$ coordinates in $\cL$, so that $\cK$ is not completely contained in $\partial \cL$.
Writing $\partial \cK$ to denote the subset of $\cK$ with the last coordinate set to $0$, we have $\partial \cK\subset \partial\cL$.
For example, we may take
$\cK=\bbZ^{d-1}\times \{0\}\times \bbN$.
for which $\partial \cK = \bbZ^{d-1}\times \{0\}\times \{0\}$.
$$
\tikzmath{
\foreach \x in {0,.5,1,1.5,2.5}{
\filldraw (\x,.5) circle (.05cm);
}
\foreach \y in {1,1.5,...,2.5}{
\foreach \x in {0,.5,1,1.5,2.5}{
\filldraw[red] (\x,\y) circle (.05cm);
}
}
\foreach \y in {1,1.5,...,2.5}{
\filldraw[blue] (2,\y) circle (.05cm);
}
\filldraw[blue!50!red!50] (2,.5) circle (.05cm);
\node[blue!50!red!50] at (2,0) {$\partial \cK$};
\node[blue] at (2,3) {$\cK$};
\node at (3,.5) {$\partial \cL$};
\draw[thick, cyan, rounded corners=5pt] (-.2,.3) -- (2.3,.3) -- (2.3,2.7);
\node[cyan] at (.75,1.75) {$\bbQ$};
}
$$
Choose a `quadrant' $\bbQ\subset \cL$ on one side of $\cK$ so that $\cK\subset \partial \bbQ$ as in the diagram above.
(Here, we use $\partial$ to denote the most external sites.)
We now construct a boundary net of algebras $\fB^\partial$ on $\cK$ together with a quantum channel $\fA(\bbQ) \to \fB^\partial$.

Similar to \cite[Lem.~2.13]{MR4945955}, the axioms \ref{BoundaryLTO:QECC}--\ref{BoundaryLTO:Injective} above ensure that $\fB^\partial(\Lambda \Subset_s^\partial \Delta)$ is independent of the choice of $\Lambda \Subset_s^\partial \Delta\subset \bbQ$ beyond the intersection $I:=\partial \Lambda \cap \partial \Delta\cap  \cK$ when $\Lambda\cap \partial \cL\neq \emptyset$.
Indeed, we may choose $\Lambda_I\subset \Lambda$ to be the smallest sufficiently large rectangle with $\partial \Lambda_I \cap \cK=I$ and $\Delta_I\subset \Delta$ to be the smallest rectangle with $\Lambda_I\Subset_s^\partial \Delta_I$ and $\partial\Lambda_I\cap \partial\Delta_I = I$.
We then define 
$$
\fB^\partial(I):= \fB^\partial(\Lambda_I\Subset^\partial_s\Delta_I).
$$
For rectangles $\Lambda\Subset_s \Delta \subset \bbQ$ with 
$I=\partial\Lambda\cap \partial \Delta \subset \cK$ and $\Lambda \cap \partial \cL=\emptyset$,
we define $\fB^\partial(I)=\fB(I)$ as before.
The map $I\mapsto \fB^\partial(I)$ now defines a net of algebras with boundary on $\cK$.

We now introduce the notion of a \emph{boundary DHR bimodule} for a net of algebras $\fA$ on a lattice $\cL$ with boundary $\partial \cL$, which is localized near the boundary $\partial \cL$.
(We will use this notion for the boundary algebra $\fB^\partial$ on $\cK$, which has boundary $\partial \cK = \cK\cap \partial \cL$.)

\begin{defn}
\label{defn:BoundaryDHRBimodule}
A \emph{boundary DHR bimodule} for $\fA$ is an $\fA-\fA$ Hilbert $\rmC^*$-bimodule ${}_{\fA} Y_{\fA}$ such that 
for any sufficiently large rectangle $\Lambda\subset \cL$ which meets $\partial \cL$, $Y$ is localized in $\Lambda$ in the sense of Definition~\ref{def:LocalizableBimodule} above.
Unpacking, this means there is a finite $Y_{\fA}$-basis $\{\beta_i\}\subset Y$ such that $\beta_i a=a\beta_i$ for all 
$$
a\in \fA(\Lambda^c)
:= 
\varinjlim_{\Delta\subset \Lambda^c} \fA(\Delta)
\subset 
\fA.
$$
That is, the basis elements $\beta_i$ commute with all the elements of $\fA$ which are in the subalgebra supported on $\Lambda^c$.

We denote the full subcategory of boundary DHR-bimodules by $\DHR^\partial(\fA)\subset \Bim(\fA)$.
Observe that $\DHR^\partial(\fA)$ is closed under composition of bimodules and thus forms a full $\rmC^*$ tensor subcategory.
\end{defn}

Similar to Lemma \ref{lem:DHRBimoduleCharacterization} above, we have the following characterization of boundary DHR bimodules whose proof is identical and therefore omitted.

\begin{lem}[Local triviality of boundary DHR bimodules]
\label{lem:BoundaryDHRBimoduleCharacterization}
A bimodule ${}_{\fA} Y_{\fA}\in \Bim(\fA)$ lies in $\DHR^\partial(\fA)$ if and only if for any sufficiently large $\Lambda$ which meets $\partial\cL$, the restriction
$
{}_{\fA(\Lambda^c)} Y_{\fA}
$
is isomorphic to a summand of a finite direct sum of copies of ${}_{\fA(\Lambda^c)}\fA_{\fA}$.
\end{lem}

Even though our net of algebras on our lattice with boundary may no longer satisfy algebraic Haag duality (e.g., the fusion module categorical net in \S\ref{sec:BoundaryDHRforBoundaryLW} below does not have this property), we still have a notion of \emph{boundary algebraic Haag duality}.
(There are obvious weaker/bounded spread variants of this notion, which we leave to the reader.)

\begin{defn}
A net of algebras $\fA$ on a lattice with boundary satisfies \emph{boundary algebraic Haag duality} if for every rectangle $\Delta$ meeting the boundary $\partial \cL$,
$Z_{\fA}(\fA(\Delta^c))=\fA(\Delta)$,
where again $\fA(\Delta^c)$ is the subalgebra of $\fA$ generated by the $\fA(\Lambda)$ with $\Lambda\subset \Delta^c$.
\end{defn}

The proof of the following lemma is identical to Lemma \ref{lem:IPLocalizedInCentralizer}.

\begin{lem}
\label{lem:BoundaryIPLocalizedInCentralizer}
If $\eta,\xi\in Y\in \Bim^\partial(\fA)$ are both localized in $\Delta$ which meets $\partial\cL$ and $a\in \fA(\Delta)$, then
$\langle \eta|a\xi\rangle_{\fA} \in Z_{\fA}(\fA(\Lambda))$ for all $\Lambda\subset \Delta^c$.
In particular, if $\fA$ satisfies boundary algebraic Haag duality, then $\langle \eta|a\xi\rangle_{\fA}\in \fA(\Delta)$.
\end{lem}

The tensor category $\DHR^\partial(\fA)$ does not have a braiding. 
We expect that one can adapt the definition of the braiding on $\DHR(\fA)$ to get a canonical \emph{half-braiding} for the image of $\DHR(\fA)$ in $\DHR^\partial(\fA)$ with $\DHR^\partial(\fA)$.
We leave this exploration for a future article as it would take us too far afield.

\section{Boundary LTO for Levin-Wen}
\label{sec:BoundaryLTOforBoundaryLW}

In this section, we construct a boundary LTO from the boundary Levin-Wen model \cite{MR2942952} built from a unitary fusion category (UFC) $\cC$ and an indecomposable unitary module category $\cM_\cC$.
We work with a unitary tensor category version of this model along the lines of \cite{MR3204497,2305.14068} 
rather than the 6j symbol model.
To do so, let us first recall our notation.
We denote the unique unitary spherical dual functor of $\cC$ by $\vee$, and we write $d_c:=\tr^\cC(\id_c)$.
We write $\Irr(\cC)$ for a set of representatives for the simple objects of $\cC$, and we write $D_\cC:= \sum_{c\in\Irr(\cC)} d_c^2$.
For our module category $\cM$, we write $\Irr(\cM)$ for a set of representatives for the simple objects of $\cM$.

Our indecomposable unitary module category $\cM$ is equipped with a unitary \emph{module trace} $\Tr^\cM$, which consists of a family of faithful positive linear functionals
$\Tr^\cM_m: \cM(m\to m)\to \bbC$ for each $m\in\cM$ satisfying
$$
\Tr^\cM_m(g\circ f) = \Tr^\cM_n(f\circ g)
\qquad\qquad\qquad
\forall\, m,n\in\cM,
\quad
\forall\, f:m\to n,
\quad g: n \to m,
$$
and for all $h: m\lhd c \to m\lhd c$, we have
$$
\Tr^\cM_{m\lhd c}\left(
\tikzmath{
\begin{scope}
\clip[rounded corners=5pt] (-.7,.7) rectangle (0,-.7);
\fill[gray!50] (-.2,-.7) rectangle (-.7,.7);
\end{scope}
\draw[thick, \ModuleColor] (-.2,.3) --node[left]{$\scriptstyle m$} (-.2,.7);
\draw[thick, \ModuleColor] (-.2,-.3) --node[left]{$\scriptstyle m$} (-.2,-.7);
\draw (.2,.3) --node[right]{$\scriptstyle c$} (.2,.7);
\draw (.2,-.3) --node[right]{$\scriptstyle c$} (.2,-.7);
\roundNbox{fill=white}{(0,0)}{.3}{.1}{.1}{$h$}
}
\right)
=
\Tr^\cM_m\left(
\tikzmath{
\begin{scope}
\clip[rounded corners=5pt] (-.7,.7) rectangle (0,-.7);
\fill[gray!50] (-.2,-.7) rectangle (-.7,.7);
\end{scope}
\draw[thick, \ModuleColor] (-.2,.3) --node[left]{$\scriptstyle m$} (-.2,.7);
\draw[thick, \ModuleColor] (-.2,-.3) --node[left]{$\scriptstyle m$} (-.2,-.7);
\draw (.2,.3) arc(180:0:.2cm) --node[right]{$\scriptstyle c^\vee$} (.6,-.3) arc(0:-180:.2cm);
\roundNbox{fill=white}{(0,0)}{.3}{.1}{.1}{$h$}
}
\right)
\qquad\qquad
\tikzmath{
\draw[rounded corners=5pt, dotted] (-.3,-.3) rectangle (.3,.3);
}
=
\cC,
\qquad
\tikzmath{
\fill[rounded corners=5pt, gray!50] (-.3,-.3) rectangle (.3,.3);
}
=
\End(\cM_\cC).
$$
By \cite{MR3019263}, $\Tr^\cM$ is unique up to positive scalar.
This scalar will be chosen judiciously later on; the curious reader may read Appendix \ref{appendix:AFActionsOfUmFCs} first if they desire.
We write $d_m:= \Tr^\cM_m(\id_m)$, and note that in our conventions,  $\sum_{m\in\Irr(\cM)} d_m^2 \neq D_\cC$ in general.

\subsection{The boundary Levin-Wen model}
\label{sec:BoundaryLW}

For simplicity we will only consider the model on a square lattice in two dimensions with a 1D boundary.
The Hilbert space can be visualized as follows, where the black edges carry labels from $\Irr(\cC)$ and the blue edges carry labels from $\Irr(\cM)$.
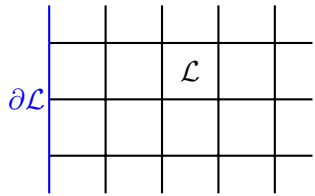
$$
\tikzmath{
\draw[thick,step=.75, thick] (0.01,0.25) grid (3.5,2.75);
\draw[thick, blue] (.0,0.25) -- (.0,2.75);
\node at (1.875,1.875) {$\cL$};
\node[blue] at (-.3,1.5) {$\partial\cL$};
}
$$
Here, we read from bottom left to top right.
The total Hilbert space is the tensor product of local Hilbert spaces:
\begin{align}
\label{eq:LWBulkSpace}
\tikzmath{
\draw[thick] (-.5,0) node[left]{$\scriptstyle a_1$} -- (.5,0) node[right]{$\scriptstyle a_2$};
\draw[thick] (0,-.5) node[below]{$\scriptstyle b_1$} -- (0,.5) node[above]{$\scriptstyle b_2$};
\filldraw (0,0) circle (.05cm);
\draw[blue!50, very thin] (-.5,.5) -- (.5,-.5);
{\draw[blue!50, very thin, -stealth] ($ (-.3,.3) - (.1,.1)$) to ($ (-.3,.3) + (.1,.1)$);}
{\draw[blue!50, very thin, -stealth] ($ (.3,-.3) - (.1,.1)$) to ($ (.3,-.3) + (.1,.1)$);}
}
\qquad
\longleftrightarrow
\qquad
\cH_v 
:&= 
\bigoplus_{
a_i,b_i\in \Irr(\cC)
}
\cC(a_1 b_1\to b_2 a_2)
=
\cC(X^{\otimes 2}\to X^{\otimes 2})
\\
\label{eq:LWBoundarySpace}
\tikzmath{
\draw[thick, blue] (0,-.5) node[below]{$\scriptstyle m_1$} -- (0,.5) node[above]{$\scriptstyle m_2$};
\draw[thick] (0,0) -- (.5,0) node[right]{$\scriptstyle c$};
\filldraw[blue] (0,0) circle (.05cm);
\draw[blue!50, very thin] (-.5,.5) -- (.5,-.5);
{\draw[blue!50, very thin, -stealth] ($ (-.3,.3) - (.1,.1)$) to ($ (-.3,.3) + (.1,.1)$);}
{\draw[blue!50, very thin, -stealth] ($ (.3,-.3) - (.1,.1)$) to ($ (.3,-.3) + (.1,.1)$);}
}
\qquad
\longleftrightarrow
\qquad
\cH^\partial_v 
:&= 
\bigoplus_{
\substack{
c\in\Irr(\cC)
\\
m_i\in \Irr(\cM)
}}
\cM(m_1\to m_2\lhd c)
=
\cM(W\to W\lhd X)
\end{align}
where the direct sum is orthogonal.
Above, we use the shorthand notation 
$X:= \bigoplus_{c\in\Irr(\cC)} c$
and
$W:= \bigoplus_{m\in\Irr(\cM)} m$.
The spaces $\cH_v,\cH^\partial_v$ are equipped with the \emph{skein-module inner products}.
\begin{align*}
\left\langle
\tikzmath{
\draw[thick] (-.5,0) node[left]{$\scriptstyle a_1$} -- (.5,0) node[right]{$\scriptstyle a_2$};
\draw[thick] (0,-.5) node[below]{$\scriptstyle b_1$} -- (0,.5) node[above]{$\scriptstyle b_2$};
\node at (.2,-.2) {$\scriptstyle \xi$};
}
\middle|
\tikzmath{
\draw[thick] (-.5,0) node[left]{$\scriptstyle a_1'$} -- (.5,0) node[right]{$\scriptstyle a_2'$};
\draw[thick] (0,-.5) node[below]{$\scriptstyle b_1'$} -- (0,.5) node[above]{$\scriptstyle b_2'$};
\node at (.2,-.2) {$\scriptstyle \xi'$};
}
\right\rangle
&=
\left(\prod_{i=1,2}
\frac{
\delta_{a_i=a_i'}
\delta_{b_i=b_i'}
}{\sqrt{d_{a_i}d_{b_i}}}
\right)
\cdot
\tr^\cC(\xi^\dag\circ \xi')
\\
\left\langle
\tikzmath{
\draw[thick, blue] (0,-.5) node[below]{$\scriptstyle m_1$} -- (0,.5) node[above]{$\scriptstyle m_2$};
\draw[thick] (0,0) -- (.5,0) node[right]{$\scriptstyle c$};
\node[blue] at (-.2,0) {$\scriptstyle \eta$};
}
\middle|
\tikzmath{
\draw[thick, blue] (0,-.5) node[below]{$\scriptstyle m_1'$} -- (0,.5) node[above]{$\scriptstyle m_2'$};
\draw[thick] (0,0) -- (.5,0) node[right]{$\scriptstyle c'$};
\node[blue] at (-.2,0) {$\scriptstyle \eta'$};
}
\right\rangle
&=
\frac{\delta_{c=c'}}{\sqrt{d_c}}
\left(\prod_{i=1,2}
\frac{\delta_{m_i=m_i'}}{\sqrt{d_{m_i}}}
\right)
\cdot
\Tr^\cM(\eta^\dag\circ \eta')
\end{align*}

Consider now a 2D rectangle $\Lambda$ in our lattice $\cL$.
We consider the canonical spin system from this setup, i.e., we define local algebras
$\fA(\Lambda):= \bigotimes_{v\in \Lambda} B(\cH_v)$.
We set $\fA:= \varinjlim \fA(\Lambda) = \bigotimes_v B(\cH_v)$.
There are edge terms $A_\ell$ which enforce that the simple labels on connected edges match, e.g. for two neighboring vertices linked by a horizontal edge $\ell$ we have
\[
A_\ell \left(
\tikzmath{
\draw[thick] (-.5,0) node[left]{$\scriptstyle a_1$} -- (.5,0) node[right]{$\scriptstyle a_2$};
\draw[thick] (0,-.5) node[below]{$\scriptstyle b_1$} -- (0,.5) node[above]{$\scriptstyle b_2$};
\node at (.2,-.2) {$\scriptstyle \xi$};
}
\otimes
\tikzmath{
\draw[thick] (-.5,0) node[left]{$\scriptstyle a_1'$} -- (.5,0) node[right]{$\scriptstyle a_2'$};
\draw[thick] (0,-.5) node[below]{$\scriptstyle b_1'$} -- (0,.5) node[above]{$\scriptstyle b_2'$};
\node at (.2,-.2) {$\scriptstyle \xi'$};
}
\right)
:=
\delta_{a_2 = a_1'}
\left(
\tikzmath{
\draw[thick] (-.5,0) node[left]{$\scriptstyle a_1$} -- (.5,0) node[above]{$\scriptstyle a_2$};
\draw[thick] (0,-.5) node[below]{$\scriptstyle b_1$} -- (0,.5) node[above]{$\scriptstyle b_2$};
\node at (.2,-.2) {$\scriptstyle \xi$};
\draw[thick] (0.5,0) -- (1.5,0) node[right]{$\scriptstyle a_2'$};
\draw[thick] (1,-.5) node[below]{$\scriptstyle b_1'$} -- (1,.5) node[above]{$\scriptstyle b_2'$};
\node at (1.2,-.2) {$\scriptstyle \xi'$};
}
\right)
,
\]
and similarly for the boundary edges.
We will primarily be interested in the subspace spanned by tensors where the labels match.
There are plaquette terms $B_p$ for every square in the lattice, defined on the subspace of matching edges as pictured below (and zero on the orthogonal complement of this space).
\begin{equation}
\label{eq:LWPlaquette}
\frac{1}{D_\cC}
\sum_{a\in \Irr(\cC)}d_c\cdot
\tikzmath{
\filldraw[knot, thick, cyan, rounded corners=5pt, fill=cyan!20] (.15,.15) rectangle (.85,.85);
\node[cyan] at (.25,.5) {$\scriptstyle c$};
\draw[thick] (-.5,-.5) grid (1.5,1.5);
}
\qquad\qquad\text{or}\qquad\qquad
\frac{1}{D_\cC}
\sum_{c\in \Irr(\cC)}d_c\cdot
\tikzmath{
\filldraw[knot, thick, cyan, rounded corners=5pt, fill=cyan!20] (.15,.15) rectangle (.85,.85);
\node[cyan] at (.25,.5) {$\scriptstyle c$};
\draw[thick] (.01,-.5) grid (1.5,1.5);
\draw[thick, blue] (0,-.5) -- (0,1.5);
}
\end{equation}
For plaquettes meeting the boundary, we must use the module action to resolve the diagram into the local Hilbert spaces.
The result \cite[Lem.~2.8]{2305.14068} (see also \cite[\S5]{0907.2204}, \cite[\S5]{MR3204497}) provides a formula for the matrix coefficients $C(\xi,\xi')$ for the $B_p$ operators in the bulk in terms of diagrams in the associated \emph{skein module} which shows these operators are self-adjoint projectors.
A similar formula holds along the boundary.

In more detail, 
the skein module for $\cC$ with $n$ boundary points, denoted $\cS_\cC(n)$, is the Hilbert space
$$
\cC(1_\cC\to X^{\otimes n})
=
\bigoplus_{x_1,\dots, x_n\in\Irr(\cC)}
\cC(1_\cC\to x_1\otimes \cdots \otimes x_n)
$$
where the direct sum is orthogonal, 
and each subspace $\cC(1_\cC\to x_1\otimes \cdots \otimes x_n)$ has inner product
$$
\left\langle
\tikzmath{
\draw (-.3,.3) --node[left]{$\scriptstyle x_1$} (-.3,.7);
\node at (.05,.5) {$\cdots$};
\draw (.3,.3) --node[right]{$\scriptstyle x_n$} (.3,.7);
\roundNbox{}{(0,0)}{.3}{.3}{.3}{$f$}
}
\middle|
\tikzmath{
\draw (-.3,.3) --node[left]{$\scriptstyle x_1$} (-.3,.7);
\node at (.05,.5) {$\cdots$};
\draw (.3,.3) --node[right]{$\scriptstyle x_n$} (.3,.7);
\roundNbox{}{(0,0)}{.3}{.3}{.3}{$g$}
}
\right\rangle
:=
\frac{1}{\sqrt{d_{x_1}\cdots d_{x_n}}} 
\tikzmath{
\draw (-.3,.3) --node[left]{$\scriptstyle x_1$} (-.3,.7);
\node at (.05,.5) {$\cdots$};
\draw (.3,.3) --node[right]{$\scriptstyle x_n$} (.3,.7);
\roundNbox{}{(0,1)}{.3}{.3}{.3}{$f^\dag$}
\roundNbox{}{(0,0)}{.3}{.3}{.3}{$g$}
}\,.
$$
Similarly, recalling $W := \bigoplus_{m\in\Irr(\cM)} m$, the skein module for $\cM$ with two $\cM$-boundary points and $n$ $\cC$-boundary points, denoted $\cS_\cM(2+n)$, is the Hilbert space
$$
\cM(W\to W\lhd X^{\otimes n})
=
\bigoplus_{\substack{
x_1,\dots, x_n\in\Irr(\cC)
\\
m,n\in\Irr(\cM)
}}
\cM(m\to n\lhd x_1\otimes \cdots \otimes x_n)
$$
where the direct sum is orthogonal, 
and each subspace $\cM(m\to n\lhd x_1\otimes \cdots \otimes x_n)$ has inner product
$$
\left\langle
\tikzmath{
\draw[thick, blue] (-.6,-.3) -- node[left]{$\scriptstyle m$} (-.6,-.7) ; 
\draw[thick, blue] (-.6,.3) -- (-.6,.7) node[above]{$\scriptstyle n$} ; 
\draw (-.3,.3) -- (-.3,.7) node[above]{$\scriptstyle x_1$} ;
\node at (.05,.5) {$\cdots$};
\draw (.3,.3) -- (.3,.7) node[above]{$\scriptstyle x_n$};
\roundNbox{}{(0,0)}{.3}{.6}{.3}{$f$}
}
\middle|
\tikzmath{
\draw[thick, blue] (-.6,-.3) -- node[left]{$\scriptstyle m$} (-.6,-.7) ; 
\draw[thick, blue] (-.6,.3) -- (-.6,.7) node[above]{$\scriptstyle n$} ; 
\draw (-.3,.3) -- (-.3,.7) node[above]{$\scriptstyle x_1$} ;
\node at (.05,.5) {$\cdots$};
\draw (.3,.3) -- (.3,.7) node[above]{$\scriptstyle x_n$};
\roundNbox{}{(0,0)}{.3}{.6}{.3}{$g$}
}
\right\rangle
:=
\frac{1}{\sqrt{\textcolor{blue}{d_md_n}d_{x_1}\cdots d_{x_n}}} 
\Tr^\cM_m\left(\,
\tikzmath{
\draw[thick, blue] (-.6,-.3) -- node[left]{$\scriptstyle m$} (-.6,-.7) ; 
\draw[thick, blue] (-.6,1.3) -- node[left]{$\scriptstyle m$} (-.6,1.7) ; 
\draw[thick, blue] (-.6,.3) -- node[left]{$\scriptstyle n$} (-.6,.7); 
\draw (-.3,.3) -- (-.3,.7);
\node at (.05,.5) {$\cdots$};
\draw (.3,.3) --node[right]{$\scriptstyle x_n$} (.3,.7);
\roundNbox{}{(0,1)}{.3}{.6}{.3}{$f^\dag$}
\roundNbox{}{(0,0)}{.3}{.6}{.3}{$g$}
}
\,\right).
$$
The square roots of dimensions are included in the denominator to ensure that gluing of morphisms is isometric as in \cite[Lem.~2.4]{2305.14068}.

We define
$p^A_\Lambda := \prod_{\ell\subset \Lambda}A_\ell$,
$p^B_\Lambda:=\prod_{p\subset \Lambda} B_p$, and $p_\Lambda:=p^A_\Lambda p^B_\Lambda$.
We have the following theorem for the operator $p^B_\Lambda$ on $p^A_\Lambda \bigotimes_{v\in \Lambda}\cH_v$.
For the rectangle of sites $\Lambda$ drawn below, we write $\partial_E\Lambda$ to denote the boundary edges emanating from sites in $\partial\Lambda$.
In the cartoon below, we color sites in \textcolor{red}{$\partial \Lambda$} in red and edges in \textcolor{blue}{$\partial_E\Lambda$} in blue.
\[
\tikzmath{
\draw[step=.5] (-.4,-.4) grid (2.4,2.4);
\foreach \x in {0,.5,1,1.5,2}{
\draw[thick, blue] (\x,0) -- (\x,-.4);
\draw[thick, blue] (\x,2) -- (\x,2.4);
\draw[thick, blue] (0,\x) -- (-.4,\x);
\draw[thick, blue] (2,\x) -- (2.4,\x);
}
\foreach \x in {0,.5,1,1.5,2}{
\filldraw[red] (\x,0) circle (.05cm);
\filldraw[red] (\x,2) circle (.05cm);
}
\foreach \x in {0,2}{
\filldraw[red] (\x,.5) circle (.05cm);
\filldraw[red] (\x,1) circle (.05cm);
\filldraw[red] (\x,1.5) circle (.05cm);
}
}
\]

\begin{lem}[{\cite{MR3204497}, \cite[Thm.~2.9]{2305.14068}}]
For a rectangle $\Lambda$ contained in the bulk, on the subspace $p^A_\Lambda \bigotimes_{v\in \Lambda}\cH_v$,
$$
p^B_\Lambda = D_{\cC}^{-\#p\subset \Lambda} \cdot \eval_\cC^\dag\circ \eval_\cC,
$$
where $\#p\subset \Lambda$ is the number of plaquettes in $\Lambda$
and $\eval_\cC$ is the evaluation map from the image of $p_A^\Lambda$ to the skein module $\cS_\cC(\#\partial_E \Lambda)$.
Hence $p_\Lambda \bigotimes_{v\in \Lambda} \cH_v$ is unitarily isomorphic to $\cS_\cC(\#\partial_E \Lambda)$ via the map $D_{\cC}^{-\#p\subset \Lambda/2}\eval_\cC$.
\end{lem}

This result was used in \cite{MR4945955} to prove that the net of projections $(p_\Lambda)$ for the 2D Levin-Wen model without boundary
satisfies the LTO axioms from~\cite{MR4945955} (see \S\ref{sec:LTO} above for the main definitions).
In order to extend this to the boundary Levin-Wen model, we have a similar result for rectangles which meet the topological boundary.

\begin{lem}
For a rectangle $\Lambda$ such that $\partial \Lambda\cap \partial\cL\neq \emptyset$, on $p^A_\Lambda \bigotimes_{v\in \Lambda}\cH_v$,
$$
p^B_\Lambda = D_{\cC}^{-\#p\subset \Lambda} \cdot \eval_\cM^\dag\circ \eval_\cM,
$$
where $\#p\subset \Lambda$ is the number of plaquettes in $\Lambda$
and $\eval_\cM$ is the evaluation map from the image of $p^A_\Lambda$ to the skein module $\cS_\cM(\#\partial_E \Lambda)$.
Hence $p_\Lambda \bigotimes_{v\in \Lambda} \cH_v$ is unitarily isomorphic to $\cS_\cM(\#\partial_E \Lambda)$ via the map $D_{\cC}^{-\#p\subset \Lambda/2}\eval_\cM$.
\end{lem}
Note that $\#\partial_E \Lambda$ contains two $\cM$-edges labelled $W$ and $(\#\partial_E \Lambda-2)$ $\cC$-edges labelled $X$ in the skein module. 
With this identification the proof of this lemma is entirely analogous to the proof of~\cite[Thm.~2.9]{2305.14068}.

\begin{defn}
We now cut our 2D region along a horizontal 1D slice as in \eqref{eq:BoundaryRectangleLW} below.
We denote the 1D slice by $\cK$.
For a 2D rectangle $\Lambda$ contained in the bottom half-space bounded by $\cK$, we let $I=I(\Lambda)$ be the intersection $\partial\Lambda\cap \cK$.
Observe that $I=\emptyset$ if $\Lambda\cap \cK=\emptyset$.
When $I\neq \emptyset$, we define
$$
\fM(I)
:=
\begin{cases}
\End_\cC(X^{\otimes I})
&
\text{if }
I\cap \partial\cL = \emptyset
\\
\End_\cM(W\lhd X^{\otimes I\setminus\partial\cL})
&
\text{if }
I\cap \partial\cL \neq \emptyset
\end{cases}
$$
where $X^{\otimes I}$ is the tensor product of the 
generator $X=\bigoplus_{c\in\Irr(\cC)} c \in\cC$ over the sites in $I$,
$X^{\otimes I\setminus\partial\cL}$ means taking the tensor product over the sites in $I\setminus \partial\cL$, 
and
$W=\bigoplus_{m\in\Irr(\cM)} m\in\cM$ as above.
In \eqref{eq:BoundaryRectangleLW} below, $I$ is
the interval where $\partial\Lambda$ intersects the red slice, so points in $I$ are oriented from \emph{left-to-right}.
For example:
\begin{equation}
\label{eq:BoundaryRectangleLW}
\tikzmath{
\draw[thick, step=.75, thick] (.01,0.25) grid (3.5,2.75);
\draw[thick, blue] (0,.25) -- (0,2.75);
\draw[thick, red] (0,1.5) node[left, xshift=.1cm]{$\scriptstyle \cK$} -- (3.5,1.5);
\draw[thick, cyan, rounded corners = 5pt] (1.125,.25) -- (1.125,1.7) --node[below,yshift=.1cm]{$\scriptstyle I$} (2.625,1.7) -- (2.625,.25);
\node[cyan] at (1.875,.5) {$\Lambda$};
\node[cyan] at (4.5,1.875) {$\rightsquigarrow$};
\node[cyan] at (7,1.875) {$\fM(I)=\End_\cC(X^{\otimes 2})$};
\draw[thick, orange, rounded corners = 5pt] (-.375,.25) -- (-.375,1.8) --node[below,yshift=.1cm]{$\scriptstyle J$} (1.125,1.8) -- (3.375,1.8) -- (3.375,.25);
\node[orange] at (.375,.5) {$\Delta$};
\node[orange] at (4.5,1.125) {$\rightsquigarrow$};
\node[orange] at (7.4,1.125) {$\fM(J)=\End_\cM(W \lhd X^{\otimes 4})$};
}
\end{equation}
We call the 1D net of algebras with boundary $\fM$ the \emph{fusion module categorical net} associated to $\cM_{\cC}$ with generators $X\in\cC$ and $W\in\cM$.
\end{defn}

When $I=\partial\Lambda\cap\cK\neq\emptyset$ but $\partial\Lambda\cap \partial\cL=\emptyset$,
just as \cite[Def.~4.6]{MR4945955}, for $\varphi \in \End_\cC(X^{\otimes I})$, we can define a gluing operator $\Gamma_\varphi$ acting on $\bigotimes_{v\in \Lambda} \cH_v$. 
When $I$ is as in \eqref{eq:BoundaryRectangleLW} and
$\varphi \in \cC(\bigotimes_I a_i \to \bigotimes_I b_i)$, $\Gamma_\varphi$ first applies the projector $p_\Lambda$ so that the edge labels of the vectors match and plaquettes are effectively filled.
This allows us to view a vector in $p_\Lambda\bigotimes_{v\in\Lambda} \cH_v$ 
as an element of the skein module $\cS_\cC(\#\partial_E\Lambda)$, giving an element in $\cC(1_\cC \to X^{\otimes \partial_E \Lambda})$.
We can then take our desired morphism $\varphi$ and glue it on, first projecting to make sure the labels $a_i$ in $\bigotimes_I a_i$ match so that composition is well-defined.
In order for this action to be a $*$-action, we also multiply by a ratio of quantum dimensions to the $1/4$th power; we refer the reader to \cite[Eq.~(2)]{2305.14068} for further details.
\begin{align*}
\tikzmath{
\draw[thick, step=.75] (0.25,1) grid (3.5,2);
\foreach \x in {1,2,3,4}{
\draw[thick] ($ (0,2) + \x*(.75,0) $) node[above]{$\scriptstyle c_{\x}$};
}
}
&\longmapsto
\prod_{i} \delta_{c_{i}=a_{i}}
\tikzmath{
\draw[thick, step=.75] (0.25,1) grid (3.5,2);
\foreach \x in {1,2,3,4}{
\draw[thick] ($ (0,2) + \x*(.75,0) $) node[above]{$\scriptstyle a_{\x}$};
}
}
\\&\longmapsto
\prod_{i} \delta_{c_{i}=a_{i}}
\left(\prod_j \frac{d_{b_j}}{d_{a_j}}\right)^{1/4}
\tikzmath{
\filldraw[thick, fill=gray!30, rounded corners=5pt]
(.35,2) rectangle (3.4,2.4);
\node at (1.75,2.2) {$\scriptstyle\varphi$};
\draw[thick, step=.75] (.25,1) grid (3.5,2);
\foreach \x in {1,2,3,4}{
\draw[thick] ($ \x*(.75,0) + (0,2.4) $) -- ($ \x*(.75,0) +(0,2.8) $) node[above]{$\scriptstyle b_{\x}$};
}
}
\end{align*}
This new morphism must be resolved into a vector in $\bigotimes_{v\in\Lambda} \cH_v$ in the usual way using semisimplicity.
Using these gluing operators, it was shown in \cite[Lem.~4.7 and Thm.~4.8]{MR4945955} that the usual Levin-Wen model for $\cC$ satisfies \ref{LTO:QECC}--\ref{LTO:Injective} for $s=1$, and $\Gamma$ provides an isomorphism from the fusion categorical net associated to $(\cC,X)$ to the boundary algebra:
$$
\fB(\Lambda\Subset \Delta)=\set{\Gamma_\varphi}{\varphi\in\End_\cC(X^{\otimes I})}\cong \fM(I)
\qquad\qquad
I=\partial\Delta\cap \partial\Lambda\cap \cK.
$$
The argument in \cite[Thm.~4.8]{MR4945955} uses the following essential lemma without proof; we provide a proof here for the reader's convenience.

\begin{lem}[Semisimple commutation]
\label{lem:SemisimpleCommutation}
Suppose $\cC$ is a semisimple category and $X\in \cC$ such that $\cC(c\to X)\neq 0$ for all $c\in\Irr(\cC)$.
Then for any $a\in\cC$, the commutant of the left action of $\End_\cC(X)$ on $\cC(a\to X)$ by postcomposition is the right action of $\End_\cC(a)$ by pre-composition.
\end{lem}
\begin{proof}
By semisimplicity, we can write
$$
\cC(a\to X) = \bigoplus_{c\in\Irr(\cC)} \cC(a\to c)\otimes \cC(c\to X).
$$
Observe that $\cC(c\to X)$ is an irreducible $\End_\cC(X)$-module for each $c\in\Irr(\cC)$, and $\cC(a\to c)$ then acts as a multiplicity space.
We then see that 
\begin{align*}
\End_{\End_\cC(X)}(\cC(a\to X))
&=
\bigoplus_{c\in\Irr(\cC)} \End(\cC(a\to c))
=
\bigoplus_{c\in\Irr(\cC)} \cC(a\to c)\otimes \underbrace{\cC(a\to c)^\vee}_{\cong \cC(c\to a)}
\\&=
\bigoplus_{c\in\Irr(\cC)} \cC(a\to c)\otimes \cC(c\to a)
= 
\End_\cC(a).
\qedhere
\end{align*}
\end{proof}

Similarly, 
when $\partial\Lambda\cap\cK\neq\emptyset$ and $\partial\Lambda\cap \partial\cL\neq\emptyset$,
for $\psi \in \End_\cM(W \lhd X^{\otimes J\setminus \partial\cL})$, we can define a gluing operator $\Gamma_\psi$ acting on $\bigotimes_{v\in \Lambda} \cH_v$. 
When $J$ is as in \eqref{eq:BoundaryRectangleLW} and
$\psi \in \cM(m\lhd \bigotimes_{J\setminus \partial\cL} a_j \to n\lhd \bigotimes_{J\setminus \partial\cL} b_j)$, $\Gamma_\psi$ first applies the projector $p_\Lambda$ which allows us to view a vector in $p_\Lambda\bigotimes_{v\in\Lambda} \cH_v$ 
as an element of the skein module $\cS_\cM(\#\partial_E\Lambda)$, giving a morphism in $\cM$.
We can then take our desired morphism $\psi$ and glue it on by post-composition, first projecting to make sure the labels $a_j$ in $\bigotimes_{J\setminus \partial\cL} a_j$ and $m\in\Irr(\cM)$ match so that composition is well-defined.
Again, there is a scalar we introduce in order to make the action a $*$-action.
\begin{align*}
\tikzmath{
\draw[thick, step=.75] (.76,1) grid (3.5,2);
\draw[thick, blue] (.75,1) -- (.75,2) node[above]{$\scriptstyle p$};
\foreach \x in {1,2,3}{
\draw[thick] ($ (.75,2) + \x*(.75,0) $) node[above]{$\scriptstyle c_{\x}$};
}
}
&\longmapsto
\delta_{p=m}
\prod_{i} \delta_{c_{i}=a_{i}}
\tikzmath{
\draw[thick, step=.75] (.76,1) grid (3.5,2);
\draw[thick, blue] (.75,1) -- (.75,2) node[above]{$\scriptstyle m$};
\foreach \x in {1,2,3}{
\draw[thick] ($ (.75,2) + \x*(.75,0) $) node[above]{$\scriptstyle a_{\x}$};
}
}
\\&\longmapsto
\delta_{p=m}
\left(\frac{d_{n}}{d_{m}}\right)^{1/4}
\prod_{i} \delta_{c_{i}=a_{i}}
\left(\prod_j \frac{d_{b_j}}{d_{a_j}}\right)^{1/4}
\tikzmath{
\filldraw[thick, fill=gray!30, rounded corners=5pt]
(.35,2) rectangle (3.4,2.4);
\node at (1.75,2.2) {$\scriptstyle\psi$};
\draw[thick, step=.75] (.76,1) grid (3.5,2);
\draw[thick, blue] (.75,1) -- (.75,2);
\draw[thick, blue] (.75,2.4) -- (.75,2.8) node[above]{$\scriptstyle n$};
\foreach \x in {1,2,3}{
\draw[thick] ($ \x*(.75,0) + (.75,2.4) $) -- ($ \x*(.75,0) +(.75,2.8) $) node[above]{$\scriptstyle b_{\x}$};
}
}
\end{align*}
As before, this new morphism from $\cM$ must be resolved into a vector in $\bigotimes_{v\in\Lambda} \cH_v$ in the usual way using semisimplicity.

Arguing as in \cite[Lem.~4.7 and Thm.~4.8]{MR4945955} using Lemma \ref{lem:SemisimpleCommutation}, but using a variant of the strip/ladder algebra for $\cM_{\cC}$ 
from \cite{MR2942952,MR3975865}
instead of the tube algebra, we get the following theorem.

\begin{thm}
\label{thm:BoundaryLTQO-LW}
The boundary Levin-Wen model for $\cM_\cC$ satisfies \ref{BoundaryLTO:QECC}--\ref{BoundaryLTO:Injective} for $s=1$
and the boundary algebra is isomorphic to the fusion module categorical net for $\cM_{\cC}$ with generators $X\in\cC$ and $W\in\cM$:
$$
\fB(\Lambda\Subset^\partial \Delta)=\set{\Gamma_\psi}{\psi \in \End_\cM(W\lhd X^{\otimes J\setminus \partial\cL})}\cong \fM(J)
\qquad\qquad
\forall\,J=\partial\Lambda\cap \partial\Delta\cap \cK.
$$
\end{thm}

\subsection{The boundary DHR theory for the boundary Levin-Wen model}
\label{sec:BoundaryDHRforBoundaryLW}

We now analyze the boundary DHR theory for our boundary Levin-Wen model using subfactor theory.
We denote the quasi-local algebra of the fusion module categorical net (also called the \emph{fusion module boundary algebra}) by
$$
\fM:=\varinjlim_n \End_\cM(W\lhd X^{\otimes n}).
$$

We now unpack Definition \ref{defn:BoundaryDHRBimodule} of the tensor category $\DHR^\partial(\fM)$ of boundary DHR bimodules for $\fM$.
Such a bimodule is an $\fM-\fM$ Hilbert $\rmC^*$-bimodule ${}_\fM Y_\fM$ such that there is an interval $I\subset \cK$ which meets the boundary ($I\cap \partial\cL\neq \emptyset$) such that $Y$ is \emph{localizable} in $I$, i.e., there is a finite $Y_\fM$-basis $\{\beta_i\}\subset Y$ such that $\beta_i a=a\beta_i$ for all 
$$
a\in 
\fM(I^c)
:= 
\varinjlim_k \End_\cC(X^{\otimes k})
\subset 
\varinjlim_k \End_\cM(W\lhd X^{\otimes |I-1|+k})
=
\fM
$$
That is, the basis elements $\beta_i$ commute with all the elements of $\fM$ which are in the subalgebra supported on $I^c$.

\begin{construction}
\label{construct:BoundaryDHR}
We construct for each $F\in \End(\cM_\cC)$ a DHR bimodule ${}_\fM Y^F_\fM$.
We then show that the bimodules $Y^F$ assemble into a fully faithful tensor functor 
\begin{equation}
\label{eq:EmbedEndIntoDHR}
Y: \End(\cM_\cC)\to \DHR^\partial(\fM).
\end{equation}
For $n\geq 0$, we define
$$
Y^F_n := \cM(W\lhd X^{\otimes n} \to FW\lhd  X^{\otimes n})
$$
which is naturally a bimodule for $\fM_n:=\End_\cM(W\lhd X^{\otimes n})$, where
the left action is post-composition after applying $F$ and the right action is pre-composition.
We equip $Y^F_n$ with a right $\fM_n$-valued inner product by 
$$
\langle \eta|\xi\rangle_{\fM_n}:= \eta^\dag\circ \xi \in \End((\fM_n)_{\fM_n})=\fM_n.
$$
Observe that the inclusions $-\otimes \id_X: Y^F_n\hookrightarrow Y^F_{n+1}$
are compatible with the inclusions $-\otimes \id_X : \fM_n\hookrightarrow \fM_{n+1}$, and that
$$
\langle \eta\otimes\id_X|\xi\otimes\id_X\rangle_{\fM_{n+1}}
=
(\eta\otimes \id_X)^\dag\circ (\xi\otimes \id_X)
=
\eta^\dag\circ \xi \otimes \id_X
=
\langle \eta|\xi\rangle_{\fM_n}\otimes \id_X.
$$
We thus get an inductive limit $\fM-\fM$ bimodule $Y^F:=\varinjlim Y^F_n$.

Now for each $\cC$-module natural transformation $\theta: F\Rightarrow G$ in $\End(\cM_\cC)$, we get a bimodule intertwiner as the inductive limit of the maps
$$
\cM(W\lhd X^{\otimes n} \to FW\lhd \otimes X^{\otimes n})
=
Y_n^F
\ni
f
\mapsto 
(\theta_W\otimes \id_{X^{\otimes n}})\circ f
\in 
Y_n^G
=
\cM(W\lhd X^{\otimes n} \to GW\lhd \otimes X^{\otimes n}).
$$
This construction is easily seen to be a unitary tensor functor.
Moreover, it can be shown that this functor is fully faithful using the unitary multifusion category embedding results of \cite{MR4916103}, together with the \emph{Q-system realization} technique from \cite{MR4419534}.
We review this argument in the context of $\rmC^*$-Hilbert bimodules in Appendix \ref{appendix:AFActionsOfUmFCs} below for convenience of the reader.

The reason why $Y^F$ lands in $\DHR^\partial(\fM)$ is entirely similar to the reasoning from \cite[\S6.2]{MR4945955} and \cite[Cor.~4.17]{MR4814692}.
We claim there is a projective basis for $Y^F$ which is localized at the single $\cM$-site in $\fM$.
To see this, it suffices to find a finite subset $\{b_i\}\subset \cM(W\to FW)$ such that $\sum_i b_ib_i^\dag =\id_{FW}$; including these $b_i$ into $Y^F$ then gives the desired projective basis.
Indeed, by semisimplicity,
\[
\End_\cM(FW)
\cong
\bigoplus_{m\in \Irr(\cM)} \cM(FW\to m)\otimes (m\to FW).
\]
Hence we may write $\id_{FW}$ as a finite sum 
\[
\id_{FW}
=
\sum_{m\in\Irr(\cM)}\sum_{i=1}^{n_m} f_i^m\circ g_i^m;
\qquad\qquad
g_i^m\in \cM(FW\to m)
\text{ and }
f_i^m\in \cM(m\to FW)
\]
By Connes' trick \cite[Proof of Prop.~3(b)]{MR561983} (see \cite[Chapter 3]{UQSLbook}), without loss of generality, we may take $g_i^m=(f_i^m)^\dag$.
Finally, for each $m\in \Irr(\cM)$, we fix an isometry $v_m:m\hookrightarrow W$ and we set $b_i^m:= f_i^m\circ v_m^\dag \in \cM(W\to FW)$.
The collection $\{b_i^m\}$ is our desired projective basis localized at the single boundary $\cM$-site.
\end{construction}

\begin{rem}
\label{rem:FullyFaithful}
Appendix \ref{appendix:AFActionsOfUmFCs} proves far more than the fully faithfulness of the unitary tensor functor
$\End(\cM_\cC)\hookrightarrow \Bim(\fM)$ given by $F\mapsto Y^F$.
We may think of $\End(\cM_\cC)$ as one corner of a unitary multifusion category which also includes $\cC,\cM$:
$$
\cE := 
\begin{pmatrix}
\End(\cM_\cC) & \cM
\\
\cM^{\rm op} & \cC
\end{pmatrix}.
$$
Similar to the above construction, we may also embed $\cC$ fully faithfully as $\rmC^*$-Hilbert bimodules over the inductive limit algebra
$$
\fN:=\varinjlim_k \End_\cC(X^{\otimes k}),
$$
and similarly, we can embed $\cM$ fully faithfully as $\fM-\fN$ bimodules.
All told, this produces a fully faithful unitary tensor functor
$$
\cE := 
\begin{pmatrix}
\End(\cM_\cC) & \cM
\\
\cM^{\rm op} & \cC
\end{pmatrix}
\hookrightarrow
\begin{pmatrix}
\Bim(\fM) & \Bim(\fM,\fN)
\\
\Bim(\fN,\fM) & \Bim(\fN)
\end{pmatrix}
=
\Bim(\fM\oplus\fN).
$$
Indeed, one first uses the results of \cite{MR4916103} to fully faithfully embed $\cC$ into $\Bim(\fN)$.
The Q-system realization technique then allows one to boost this embedding into a fully faithful embedding of the unitary multifusion category of $Q-Q$ bimodules ${}_Q\cC_Q$ into $\Bim(|Q|_{\fN})$, the Hilbert $\rmC^*$-bimodules over the realized Q-system.
Choosing $Q=1\oplus \underline{\End}_\cC(W)$ using the unitary internal end \cite{MR4750417,2410.05120}, we have ${}_Q\cC_Q\cong \cE$ and $|Q|_\fN\cong \fM\oplus\fN$, as claimed.

Moreover, we do not need that $\cC$ is fusion nor that $\cM$ is indecomposable; this construction works for arbitrary unitary \emph{multi}fusion $\cC$ and unitary modules $\cM_\cC$.
However, one must choose the unitary module trace on $\cM$ carefully to be compatible with the generator $W\in\cM$.

We refer the reader to Appendix \ref{appendix:AFActionsOfUmFCs} for a more detailed discussion.
\end{rem}

\begin{nota}
For an interval $I\subset \cK$ which meets $\partial \cL$,
we write $\fN_I:=\fM(I^c)$, so that we have a canonical inclusion 
$\iota_I:=1_{\fM(I)}\otimes - :\fN_I\hookrightarrow \fM$.
Observe that for all $I$, the algebras $\fN_I$ are obviously isomorphic to
$$
\fN:=\varinjlim_k \End_\cC(X^{\otimes k}).
$$
\end{nota}

The algebra $\fM$ is naturally an $\fN_I-\fN_I$ Hilbert $\rmC^*$-bimodule using the Q-system realization technique for the Q-system
$\underline{\End}_\cC(W\lhd X^{\otimes |I|-1})\in\cC$.
Hence ${}_{\fN_I}\fM_{\fN_I}$ lies in the fully faithful image of $\cC \subset \Bim(\fN_I)\cong\Bim(\fN)$.
Moreover, the centralizer
\begin{equation}
\label{eq:FusionModuleBoundaryHaagDuality}
Z_\fM(\fN_I) 
= 
\End_{\fN_{I}-\fM}(\fM)
\cong
\End_\cM(W\lhd X^{\otimes |I|-1})=\fM(I),
\end{equation}
i.e., the fusion module categorical net satisfies boundary algebraic Haag duality.
Again, more details can be found in Appendix \ref{appendix:AFActionsOfUmFCs}.

\begin{thm}
\label{thm:BoundaryLTOforLW}
The category $\DHR^\partial(\fM)$ of boundary DHR bimodules is the essential image of $\End(\cM_\cC)\subseteq \DHR^\partial(\fM)$.
\end{thm}
\begin{proof}
Suppose ${}_\fM Y_\fM\in \DHR^\partial(\fM)$.
Since $Y$ is localized in some boundary interval $I$, by Lemma~\ref{lem:BoundaryDHRBimoduleCharacterization}, 
the restriction ${}_{\fN_I}Y_\fM=  {}_{\fN_I}\fM\boxtimes_\fM Y_\fM$ is isomorphic to a summand of a finite direct sum of copies of ${}_{\fN_I}\fM_\fM$, which lies in the essential image of $\cE$ in $\Bim(\fM\oplus \fN_I)$.
The result now follows by Lemma~\ref{lem:2outof3},
as each of ${}_{\fN_I} Y_{\fM}$ and ${}_{\fN_I}\fM_\fM$ lie in $\cE\subset \Bim(\fM\oplus \fN_I)$.
\end{proof}

\section{Braided categorical nets}
\label{sec:BraidedCategoricalNets}

In~\cite{MR4945955}, an important role is played by the fusion categorical nets (or fusion spin chains) on $\bbZ$ defined in Example~\ref{ex:fusion_spin}.
For higher dimensional lattices, extra structure is needed to define inclusion by tensoring with identity morphisms compatible with the geometric locality of the underlying lattice.
We now construct a canonical net of algebras associated to a choice of object $X$ in a unitary \emph{braided} fusion category $\cB$.
We write $\beta_{a,b}:a\otimes b\to b\otimes a$ for the braiding.
Our goal is to assign an object in $\cB$ to a \emph{disk-like region} in the sense of \cite[Def.~2.4]{2605.10693}.
Given a point $v$ in our $\bbZ^2$-lattice, we write $[v]$ for the closed $\ell^\infty$-ball of radius $1/2$ about $v$ in $\bbR^2$. 
For a finite subset $F\subset \bbZ^2$, we write $[F]:=\bigcup_{v\in F} [v]\subset \bbR^2$.
We say $F\subset \bbZ^2$ is \emph{disk-like} if the boundary $\partial [F]$ in $\bbR^2$ is homeomorphic to $S^1$.
Obviously all finite rectangles $\Lambda$ are disk-like.

Now given a disk-like region $R\subset \bbZ^2$, we want to make sense of the object $X^{\otimes R}\in \cB$.
We will then define a net of algebras on the poset of disk-like regions in $\bbZ^2$ by
$$
\tikzmath{
\foreach \y in {0,.5,...,2.5}{
\foreach \x in {0,.5,...,2.5}{
\filldraw (\x,\y) circle (.05cm);
}
}
\draw[thick, blue] (.25,.25) -- (.75,.25) -- (.75,.75) -- (1.75,.75) -- (1.75,.25) -- (2.25,.25) -- (2.25,2.25) -- (1.75,2.25) -- (1.75,1.75) -- (.75,1.75) -- (.75,2.25) -- (.25,2.25) -- (.25,.25);
\node[blue] at (1.25,1.25) {$R$};
}
\qquad
\longmapsto
\qquad
\fA(R)
:= \End_\cB(X^{\otimes R}).
$$

Our construction follows the idea sketched in~\cite[\S IV.B]{MR4808260}.
We do so in more generality, where we are allowed to assign any object to any site in $R$.
That is, if $b_{ij}\in\cB$ is the object assigned to site $(i,j)\in R$, we  define $(b_{ij})^{\otimes R}$.

When $\cB$ is symmetric, there is an obvious choice for $(b_{ij})^{\otimes R}$, namely the \emph{unordered tensor product} over the sites of $R$.
We choose an arbitrary ordering of the sites in $R$, say $(i_n,j_n)$, and we take the tensor product $\bigotimes_{n=1}^{N} b_{i_nj_n}\in\cB$
(here, we have assumed $R$ has $N$ sites).
Given another ordering $(i'_n,j'_n)$ of the sites in $R$, there is a unique element $\sigma$ of the symmetric group $S_{N}$ which takes the first ordering to the second,
and we get a canonical unitary 
$u_\sigma:\bigotimes_{n=1}^{N} b_{i_nj_n}\to \bigotimes_{n=1}^{N} b_{i_n'j_n'}$.
Moreover, given any third ordering $(i''_n,j''_n)$,
there are canonical permutations $\sigma'$ taking $(i'_n,j'_n)$ to $(i''_n,j''_n)$
and $\sigma''$ taking
$(i_n,j_n)$ to $(i''_n,j''_n)$ such that $\sigma' \sigma = \sigma''$.
The corresponding canonical unitaries 
$u_{\sigma'}:\bigotimes_{n=1}^{N} b_{i_n'j_n'}\to \bigotimes_{n=1}^{N} b_{i_n''j_n''}$
and
$u_{\sigma''}:\bigotimes_{n=1}^{N^2} b_{i_nj_n}\to \bigotimes_{n=1}^{N^2} b_{i_n''j_n''}$
satisfy
$u_{\sigma'} u_{\sigma} = u_{\sigma''}$.
We thus have a contractible groupoid, and in this sense, the unordered tensor product $(b_{ij})^{\otimes R}$ is canonical.

When $\cB$ is merely braided, we must be more careful.
We again proceed by choosing an arbitrary ordering $(i_n,j_n)$ of the sites in $R$ along an interval $I$ in a $\bbZ$ lattice, e.g., a lexicographic ordering.
\begin{equation}
\label{eq:LexicographicOrdering}
\tikzmath{
\clip (-.5,-.5) rectangle (12.5,3);
\foreach \x in {.5,1,...,8}{
\filldraw (4+\x,2) circle (.05cm);
}
\foreach \y in {.5,1,1.5,2}{
\foreach \x in {.5,1,1.5,2}{
\filldraw (\x,\y) circle (.05cm);
\draw[red,thick,->] (\x,\y) to[out=-30,in=-90] 
($ 4*(\y,0)+(\x,0) + (2,2) $);
}
}
\draw[thick, blue] (.25,.25) rectangle (2.25,2.25);
\node[blue] at (1.25,1.25) {$R$};
}
\end{equation}
However, we now want this interval $I$ to be rescaled and translated to be contained \emph{inside $[R]$ itself}, unlike the diagram above.
Our choice of ordering is replaced by 
\emph{choice of homeomorphism} of $[R]$ fixing $\partial [R]$ pointwise which moves sites $(i,j)\in R$ to sites in $I\subset [R]$.
(We do not depict this in \eqref{eq:LexicographicOrdering} as the diagram would be too cramped.)
As above, we get an object $\bigotimes_{n=1}^{N} b_{i_nj_n}\in\cB$.

Now for any two choices of homeomorphisms $h_1,h_2$ of $[R]$ which map $R \to I$, we get a canonical homeomorphism of the $N$-punctured disk by taking the punctures in $I$ to themselves by 
$$
I
\xrightarrow{h_1^{-1}} R \xrightarrow{h_2} I.
$$
Since the mapping class group of the $N$-punctured disk is the braid group $B_N$, this homeomorphism of $[R]$ which fixes $\partial[R]$ and maps $I$ to itself gives a canonical element $\beta_{12}\in B_{N}$, which gives a canonical unitary $u_{\beta_{12}}: \bigotimes_{n=1}^{N} b_{i_nj_n}\to \bigotimes_{n=1}^{N} b_{i_n'j_n'}$.
As before, for a third ordering, we get two more canonical elements of $B_{N}$ and two more unitaries in $\cB$ which compose as they should.
We again get a contractible groupoid, and in this sense, the \emph{braided tensor product} $(b_{ij})^{\otimes R}$ is canonical.

\begin{rem}
That the above object $(b_{i,j})^{\otimes R}$ is well defined is proven carefully in 
\cite[Prop~7.6]{MR1989873} in the language of operads.
They work in the more restrictive setting that $\cB$ is \emph{balanced} (although the authors of \cite{MR1989873} use the term ribbon instead.)
A \emph{balancing} structure is a unitary natural isomorphism $\theta: \id_\cB\Rightarrow\id_\cB$ called the \emph{twist} which satisfies the balancing axiom
$$
(\theta_a\otimes \theta_b)
\circ
\beta_{b,a}\circ \beta_{a,b}
=
\theta_{a\otimes b}.
$$
For example, if we equip $\cB$ with any unitary dual functor $\vee: \cB \to \cB^{\textrm{mop}}$ in the sense of \cite{MR2091457,MR2767048,MR2861112,MR4133163},%
\footnote{Here $\cB^{\textrm{mop}}$ is the \emph{monoidal opposite} category obtained by reversing the arrows of $\cB$ and also reversing the order of the tensor product.}
the induced unitary pivotal structure is given by 
$$
\varphi_b:= 
(\coev_b^\dagger \otimes \id_{b^{\vee\vee}}) \circ (\id_b \otimes \coev_{b^\vee}) 
=
(\id_{b^{\vee\vee}} \otimes \ev_b) \circ (\ev_b^\dagger \otimes \id_b)
$$
which induces a canonical twist given by
$$
\theta_b:=
\tikzmath{
\draw (0,-.3) -- node[left]{$\scriptstyle b$} 
 (0,-.7);
\draw (.25,1) arc(90:-90:.2cm) to[out=180, in=-90] (0,1.2);
\draw[knot] (0,.3) to[out=90,in=180](.25,1);
\node at (-.3,.5) {$\scriptstyle b^{\vee\vee}$};
\node at (.7,.8) {$\scriptstyle b^{\vee}$};
\node at (-.2,1) {$\scriptstyle b$};
\roundNbox{}{(0,0)}{.3}{0}{0}{$\varphi_b$}
}=
(\id_b\otimes \ev_{b^\vee})
\circ
(\beta_{b^{\vee\vee},b}\otimes \id_{b^\vee})
\circ
(\id_{b^{\vee\vee}}\otimes \coev_b)
\circ\varphi_b
.
$$
This canonical twist induces a \emph{ribbon structure} (that is, $\theta_{b^\vee} = \theta_b^\vee$)
if and only if $\varphi$ is spherical \cite[Appendix~A.2]{MR3578212}.

In particular, in our setting, we may choose to equip $\cB$ with its unique unitary spherical structure \cite{MR1444286}; this is the only choice when $\cB$ is fusion \cite{MR4133163}.
Thus the results of \cite{MR1989873} can be used `out-of-the-box' to show $(b_{ij})^{\otimes R}$ is well-defined, and comes equipped with an action of the ribbon braid group.
However,
the balancing is not needed for the constructions described in this section.
The object  $(b_{ij})^{\otimes R}$ is still well-defined in the unitary braided fusion category $\cB$, as we are not keeping track of duality in defining $(b_{ij})^{\otimes R}$.
\end{rem}

In order to define a net of algebras on the poset of disk-like regions, we briefly sketch the operadic structure of this braided tensor product of objects.
Suppose we have an inclusion of disk-like regions $R\subset S$, together with labelings $(b_{ij})$ and $(c_{ij})$ of the objects in $R,S$ respectively which agree on the sites of $R$.
There is a canonical inclusion
\[
\End_\cB((b_{ij})^{\otimes R})
\hookrightarrow
\End_\cB((c_{ij})^{\otimes S})
\]
which `tensors by identity morphisms' on all sites in $S\setminus R$.
The main idea here is that since our homeomorphism of $[R]$ was chosen to fix $\partial [R]$, we may view the entire region $[R]$ as a \emph{single point} which is assigned the single object $(b_{ij})^{\otimes R}$.\footnote{Since two disk-like regions may share an interval along their boundaries, we should technically shrink $[R]$ by an $\varepsilon$-collar so that for disjoint disk-like regions $R,S$, we would also have $[R]\cap[S]=\emptyset$.}
We thus view $S$ as this single point together with $S\setminus R$, and we consider only those homeomorphisms of $[S]$ which move $[R]$ as a `single point' in the distinguished interval in $[S]$, i.e., just perform translation and dilation of $[R]$ and do not perform any $2\pi$ rotation of $\partial[R]$.\footnote{Since we are not considering balanced fusion categories, a $2\pi$ rotation of $\partial[R]$ would give a pure braid on the object $(b_{ij})^{\otimes R}$ which would not be visible if $R$ only contained one puncture.}

Using the above operadic structure, it is straightforward to check that these canonical inclusions are compatible, i.e., inclusions $R\subset S\subset T$ with compatible labelings for objects yields inclusions
\[
\End_\cB((b_{ij})^{\otimes R})
\hookrightarrow
\End_\cB((c_{ij})^{\otimes S})
\hookrightarrow
\End_\cB((d_{ij})^{\otimes T})
\]
which compose correctly, and if $R,S$ are disjoint disk-like regions contained in $T$, then the images of 
$\End_\cB((b_{ij})^{\otimes R})$ and $\End_\cB((c_{ij})^{\otimes S})$ commute inside 
$\End_\cB((d_{ij})^{\otimes T})$.
That is, we have a well-defined 2D net of algebras.

\begin{defn}[Braided fusion categorical net]
\label{defn:BraidedFusionCategoricalNet}
Let $\cB$ be a unitary braided fusion category with a choice of object $X\in\cB$ (which is often chosen to be a strong tensor generator).
Consider the lattice $\bbZ^2$.
For a disk-like region $R\subset \bbZ^2$, we define
$$
\fA(R):=\End_\cB(X^{\otimes R})
$$ 
where 
$X^{\otimes R}$ denotes the braided tensor product of copies of $X$ over the points of $R$.
\end{defn}
\subsection{Superselection sectors of the canonical state}
\label{sec:SSSofBraidedCategoricalNet}

As in the previous section, let $\fA$ denote the braided categorical net for a choice of object $X\in\cB$, where we assume that $1_\cB$ is a subobject of $X$.
For a disk-like region $R$, on the local algebra $\fA(R)$, we define the state
$$
\psi_R(f)
:=
\tikzmath{
\foreach \y in {0,1,2}{
\foreach \x in {0,1,2}{
\draw[red, thick] (\x,\y) -- ($ (\x,\y) + (.2,-.2) $);
\filldraw[red] ($ (\x,\y) + (.2,-.2) $) circle (.05cm);
}}
\filldraw[thick, rounded corners=5pt, fill=gray!50, opacity=.8] (-.2,-.2) rectangle (2.2,2.2);
\foreach \y in {0,1,2}{
\foreach \x in {0,1,2}{
\draw[red, thick] (\x,\y) -- ($ (\x,\y) + (-.2,.2) $);
\filldraw[red] ($ (\x,\y) + (-.2,.2) $) circle (.05cm);
}}
\node at (1,.5) {$f$};
}
:=
(\textcolor{red}{i_X^{\otimes R}}
)^\dag
\circ f \circ
\textcolor{red}{i_X^{\otimes R}}
\qquad\qquad
\forall\,f\in\fA(R)=\End_{\cB}(
\textcolor{red}{X^{\otimes R}}
)
$$
where the red univalent vertex is the partial isometry $i_X:1_{\cB} \hookrightarrow X$
or its adjoint.
On the inductive limit algebra $\fA=\varinjlim \fA(R)$, we define the state $\psi:=\varinjlim \psi_R$
(cf.~\cite[Cor.~2.19]{MR4945955}).
As $\psi_R(1)=\psi_R(i_X^{\otimes R}\circ (i_X^{\otimes R})^\dag)$,
as on \cite[p.~43]{MR4945955}, 
we may identify $f\Omega_\psi\in L^2(\fA,\psi)$ for $f\in \fA(R)$ by
\[
f\cdot (i_X^{\otimes R}\circ (i_X^{\otimes R})^\dag)
=
\tikzmath{
\foreach \y in {0,1,2}{
\foreach \x in {0,1,2}{
\draw[red, thick] (\x,\y) -- ($ (\x,\y) + (.2,-.2) $);
\filldraw[red] ($ (\x,\y) + (.2,-.2) $) circle (.05cm);
\filldraw[red] ($ (\x,\y) + (.4,-.4) $) circle (.05cm);
\draw[red, thick] ($ (\x,\y) + (.4,-.4) $) -- ($ (\x,\y) + (.6,-.6) $);
}}
\filldraw[thick, rounded corners=5pt, fill=gray!50, opacity=.8] (-.2,-.2) rectangle (2.2,2.2);
\foreach \y in {0,1,2}{
\foreach \x in {0,1,2}{
\draw[red, thick] (\x,\y) -- ($ (\x,\y) + (-.3,.3) $);
}}
\node at (1,.5) {$f$};
}\,.
\]

We now consider the net $\fA$ equipped with the state $\psi$, which results in a net of von Neumann algebras on the poset of cones\footnote{We use the word cone here in a slightly generalised sense where we allow the opening angle of the cone to be larger than $\pi$ (but smaller than $2 \pi$, i.e. $\mathbb{R}^2$ itself is not a cone).
See e.g.~\cite[App. A]{MR4362722}, or~\cite{MR4927814} for a discussion of what geometric properties are needed.} in $\bbR^2$ by taking $A_\Lambda:=\fA(\Lambda)''$ in the GNS representation of $\psi$.
We now prove that the net of von Neumann algebras $A$ satisfies Haag duality, and we compute the underlying $\rmW^*$-category of its superselection sectors.

Let $\langle X\rangle$ denote the full braided fusion subcategory of $\cB$ generated by $X$ and $\Hilb\langle X\rangle$ its unitary ind completion \cite[\S2.6]{MR3687214} (see also \cite{MR3509018}).
Recall that the objects of $\Hilb\langle X\rangle$ can be written as
$$
\bigoplus_{x\in \Irr\langle X\rangle} H_x\otimes x
$$
where $\Irr\langle X\rangle$ denotes a set of representatives for the simple objects of $\langle X\rangle$, and each $H_x$ is a separable (possibly infinite dimensional) multiplicity Hilbert space.

The following construction is the braided categorical net analog of \cite[Const.~6.13]{MR4945955} for fusion categorical nets.

\begin{construction}
Fix an object $x\in \Irr\langle X\rangle$.
For each rectangle $R$, define
$$
\cK_x(R):=\cB(x\to X^{\otimes R})
=
\set{
\tikzmath{
\draw[cyan, thick, snake] (1,1) -- (3,-.5) node[above]{$\scriptstyle x$};
\filldraw[thick, rounded corners=5pt, fill=gray!50, opacity=.8] (-.2,-.2) rectangle (2.2,2.2);
\foreach \y in {0,1,2}{
\foreach \x in {0,1,2}{
\draw[red, thick] (\x,\y) -- ($ (\x,\y) + (-.3,.3) $);
}}
\node at (1,.5) {$f$};
}
}{f: x\to X^{\otimes R}}
$$
with inner product
$$
\langle f|g\rangle_{\cK_x(R)}:= \tr_\cB(f^\dag\circ g).
$$
If $R \subset S$, we include $\cK_x(R)\hookrightarrow \cK_x(S)$ by tensoring with the inclusion isometry $i_X:1_\cB \hookrightarrow X$ for every site of $S \setminus R$.
For example, for the inclusion $R \subset R^{+1}$, we have the following.
$$
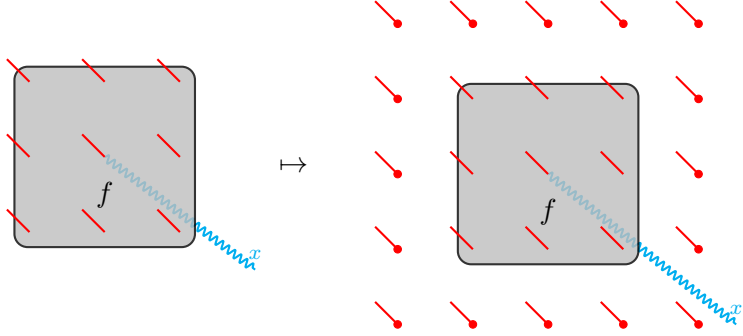
\tikzmath{
\draw[cyan, thick, snake] (1,1) -- (3,-.5) node[above]{$\scriptstyle x$};
\filldraw[thick, rounded corners=5pt, fill=gray!50, opacity=.8] (-.2,-.2) rectangle (2.2,2.2);
\foreach \y in {0,1,2}{
\foreach \x in {0,1,2}{
\draw[red, thick] (\x,\y) -- ($ (\x,\y) + (-.3,.3) $);
}}
\node at (1,.5) {$f$};
}
\mapsto \qquad
\tikzmath{
\draw[cyan, thick, snake] (1,1) -- (3.5,-1) node[above]{$\scriptstyle x$};
\filldraw[thick, rounded corners=5pt, fill=gray!50, opacity=.8] (-.2,-.2) rectangle (2.2,2.2);
\foreach \y in {-1,0,1,2,3}{
\foreach \x in {-1,0,1,2,3}{
\draw[red, thick] (\x,\y) -- ($ (\x,\y) + (-.3,.3) $);
}}
\foreach \x in {-1,0,1,2,3}{
\filldraw[red] (\x,-1) circle (.05cm);
\filldraw[red] (\x,3) circle (.05cm);
}
\foreach \y in {0,1,2}{
\filldraw[red] (-1,\y) circle (.05cm);
\filldraw[red] (3,\y) circle (.05cm);
}
\node at (1,.5) {$f$};
}
$$
Setting $\cK_x=\varinjlim \cK_x(R)$, we get a representation $\fA\to B(\cK_x)$.

Observe that $\cK_{1_\cB}\cong L^2(\fA,\psi)$ via the (inductive limit of the) map
$$
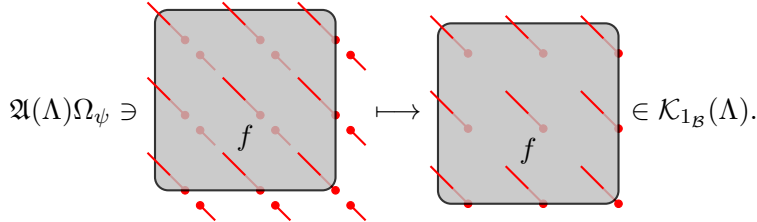
\fA(R)\Omega_\psi
\ni
\tikzmath{
\foreach \y in {0,1,2}{
\foreach \x in {0,1,2}{
\draw[red, thick] (\x,\y) -- ($ (\x,\y) + (.2,-.2) $);
\filldraw[red] ($ (\x,\y) + (.2,-.2) $) circle (.05cm);
\filldraw[red] ($ (\x,\y) + (.4,-.4) $) circle (.05cm);
\draw[red, thick] ($ (\x,\y) + (.4,-.4) $) -- ($ (\x,\y) + (.6,-.6) $);
}}
\filldraw[thick, rounded corners=5pt, fill=gray!50, opacity=.8] (-.2,-.2) rectangle (2.2,2.2);
\foreach \y in {0,1,2}{
\foreach \x in {0,1,2}{
\draw[red, thick] (\x,\y) -- ($ (\x,\y) + (-.3,.3) $);
}}
\node at (1,.5) {$f$};
}
\longmapsto
\tikzmath{
\foreach \y in {0,1,2}{
\foreach \x in {0,1,2}{
\draw[red, thick] (\x,\y) -- ($ (\x,\y) + (.2,-.2) $);
\filldraw[red] ($ (\x,\y) + (.2,-.2) $) circle (.05cm);
}}
\filldraw[thick, rounded corners=5pt, fill=gray!50, opacity=.8] (-.2,-.2) rectangle (2.2,2.2);
\foreach \y in {0,1,2}{
\foreach \x in {0,1,2}{
\draw[red, thick] (\x,\y) -- ($ (\x,\y) + (-.3,.3) $);
}}
\node at (1,.5) {$f$};
}
\in\cK_{1_\cB}(R).
$$

For an arbitrary $z\in \Hilb\langle X\rangle$, we write $z=\bigoplus_{x\in\Irr\langle X\rangle} H_x \otimes x$ and set $\cK_z:=\bigoplus_{x\in\Irr\langle X\rangle} H_x\otimes \cK_x$.
\end{construction}

In fact, one can actually reduce the study of many structural properties of the 2D braided categorical net to the 1D setting, which have already been analyzed in detail in \cite{MR4814692} and \cite{MR4945955}.
The main trick is that a cone $\Lambda$ splits the 2D plane into two halves, and the braiding can be used to move the $\bbZ^2$ lattice points on either side of the split into a 1D $\bbZ$ lattice.
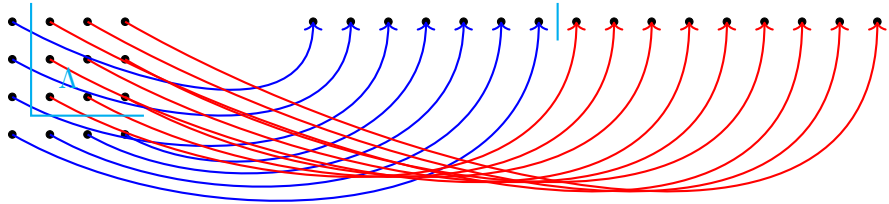
\begin{equation}
\label{eq:DimensionalReduction}
\tikzmath{
\clip (-.5,-.5) rectangle (12.5,3);
\foreach \x in {.5,1,...,8}{
\filldraw (4+\x,2) circle (.05cm);
}
\foreach \y in {.5,1,1.5,2}{
\foreach \x in {.5,1,1.5,2}{
\filldraw (\x,\y) circle (.05cm);
}}
\foreach \x in {.5,1,1.5,2}{
\draw[blue,thick,->] (\x,.5) to[out=-30,in=-90] 
($ (8,2) - (\x,0) $);
}
\foreach \y in {1,1.5,2}{
\draw[blue,thick,->] (.5,\y) to[out=-30,in=-90] 
($ (6.5,2) - (\y,0) $);
\foreach \x in {1,1.5,2}{
\draw[red,thick,->] (\x,\y) to[out=-30,in=-90] 
($ 3*(\y,0)+(\x,0) + (4,2) $);
}
}
\draw[thick, cyan] (.75,2.25) -- (.75,.75) -- (2.25,.75);
\draw[thick, cyan] (7.75, 2.25) -- (7.75,1.75);
\node[cyan] at (1.25,1.25) {$\Lambda$};
}
\end{equation}
We can thus infer properties about our braided categorical net $\fA$ acting on $\cK_x$ from the fusion categorical setting.

For the results 
Proposition~\ref{prop:NormalActionOfConeAlgebra}, 
Lemma \ref{lem:IrreducbleConeAlgebraRepresentation}, 
and
Theorems~\ref{thm:BraidedCategoricalNetHaagDuality}
and
\ref{thm:SSSforBraidedCategoricalNet}
below, our statements are given generally so that they may simultaneously be interpreted for 1D fusion spin systems and 2D braided fusion spin systems, and it suffices to prove the 1D case.
For the 1D case, we always take our cone $\Lambda$ to be a 1-sided infinite interval in $\bbZ$, say $[1,\infty)$, and we set $\Lambda_n=[1,n]$ so that $\Lambda = \bigcup \Lambda_n$.
(For 2D braided fusion spin systems, $\Lambda_n$ should grow by more than 1 site at a time, corresponding to larger initial disk-like segments of the cone $\Lambda$.)
Without loss of generality, it then suffices to consider 
$$
\cK_x(\Lambda_n) = \cB(x\to X^{\otimes n}),
$$
and $\Lambda_{n+1}$ is obtained by adding sites to the right, so that the map $\cK_x(\Lambda_n)\hookrightarrow \cK_x(\Lambda_{n+1})$ is $-\otimes i_X$.

\begin{prop}
\label{prop:NormalActionOfConeAlgebra}
Fix a cone $\Lambda$ in $\bbR^2$ (or a half-infinite interval for the 1D case).
If $x\in \langle X\rangle$, then the action of $\fA(\Lambda)$ on $B(\cK_x)$ extends to a normal action of $A_\Lambda$.
\end{prop}
\begin{proof}
By \cite[Prop.~4.15 and Rem.~4.16(2)]{MR3948170}, it suffices to exhibit a dense subspace $\cD_x\subset\cK_x$ whose vector states can be realized on $\bigcup \fA(\Lambda_n)$ by vectors in $\cK_{1_\cB}$.
Indeed, fixing a disk-like region $R$, we claim $\cD_x=\bigcup_k\cK_x(R^{+k})$ works, where $R^{+k}$ consists of all points whose distance (in the $\ell^\infty$ norm) is at most $k$ from $R$.
Since $x\in\langle X\rangle$, there is an isometry $v: x^\vee\to X^{\otimes n}$ for some $n$.
Thus if $f\in \cD_x(R^{+k})$, we may add $n$ sites in $R^{+k+n}\setminus \Lambda$ which, under the above discussed dimensional reduction, appear to the the \emph{left} of those in $R^{+k}$, as all new sites in $\Lambda$ are added to the right.
We then define $g:= (v\otimes f)\circ \ev^\dag_x\in \cK_{1_\cB}$ which realizes the vector state $\omega_f$.
For an explicit example of how this works, consider an $f\in \cK_x(\Lambda_k)$ instead of $\cK_x(R^{+k})$.
Given an $a\in \fA(\Lambda_n)$, we have
$$
\omega_f(a)
=
\tr_\cB(f^\dag af)
=
\tikzmath{
\draw (-.3,.3) -- (-.3,.7);
\node at (.03,.5) {$\cdots$};
\draw (.3,.3) -- (.3,.7);
\draw (-.3,-.3) -- (-.3,-.7);
\node at (.03,-.5) {$\cdots$};
\draw (.3,-.3) -- (.3,-.7);
\draw[cyan] (0,1.3) node[right,yshift=.2cm]{$\scriptstyle x$} arc(0:180:.5cm and .3cm) --node[left]{$\scriptstyle x^\vee$} (-1,-1.3) arc(-180:0:.5cm and .3cm) node[right,yshift=-.2cm]{$\scriptstyle x$};
\roundNbox{}{(0,1)}{.3}{.15}{.15}{$f^\dag$}
\roundNbox{}{(0,0)}{.3}{.15}{.15}{$a$}
\roundNbox{}{(0,-1)}{.3}{.15}{.15}{$f$}
}
=
\tikzmath{
\draw (-.7,-.7) -- (-.7,.7);
\node at (-.97,0) {$\cdots$};
\draw (-1.3,-.7) -- (-1.3,.7);
\draw (-.3,.3) -- (-.3,.7);
\node at (.03,.5) {$\cdots$};
\draw (.3,.3) -- (.3,.7);
\draw (-.3,-.3) -- (-.3,-.7);
\node at (.03,-.5) {$\cdots$};
\draw (.3,-.3) -- (.3,-.7);
\draw[cyan] (-1,1.3) node[left,yshift=.2cm]{$\scriptstyle x^\vee$} arc(180:0:.5cm and .3cm) node[right,yshift=.2cm]{$\scriptstyle x$};
\draw[cyan] (0,-1.3) node[right,yshift=-.2cm]{$\scriptstyle x$} arc(0:-180:.5cm and .3cm) node[left,yshift=-.2cm]{$\scriptstyle x^\vee$};
\roundNbox{}{(-1,1)}{.3}{.15}{.15}{$v^\dag$}
\roundNbox{}{(0,1)}{.3}{.15}{.15}{$f^\dag$}
\roundNbox{}{(0,0)}{.3}{.15}{.15}{$a$}
\roundNbox{}{(0,-1)}{.3}{.15}{.15}{$f$}
\roundNbox{}{(-1,-1)}{.3}{.15}{.15}{$v$}
\roundNbox{dashed, red}{(-.5,-1.15)}{.55}{.5}{.5}{}
\node[red] at (.8,-1) {$g$}; 
}
=
\omega_g(a).
$$
Observe that the above equality is compatible with adding strands to the right, i.e., taking inductive limits in $\fA(\Lambda)$.
The result follows.
\end{proof}

\begin{lem}
\label{lem:IrreducbleConeAlgebraRepresentation}
For a cone $\Lambda$  (or a half-infinite interval for the 1D case), let $\cK_x(\Lambda):=\varinjlim \cK_x(\Lambda_n)$ for $\Lambda_n$ bounded contractible regions which increase to $\Lambda$.
The action of $\fA(\Lambda)$ on $\cK_x(\Lambda)$ is irreducible.
\end{lem}
\begin{proof}
Set $\cL_1:=\cK_x(\Lambda_1)$, and inductively define $\cL_{n+1}:=\cK_x(\Lambda_{n+1})\ominus \cK_x(\Lambda_n)$.
Now $\cK_x(\Lambda)=\bigoplus \cL_n$, and we may choose an ONB $B_n$ for each $\cL_n$, and observe that the ONBs for $m\neq n$ are orthogonal.
Moreover, for all $f,g\in B_n$, observe that
$g^\dag\circ f = \delta_{g=f} d_x^{-1} \cdot \id_x$.

We now observe that we can form a system of matrix units for the ONB $\bigcup B_n$ of $\cK_x(\Lambda)$ which lie in $\bigcup \fA(\Lambda_n)\subset \fA(\Lambda)$.
That is, for $f\in B_m$ and $g\in B_n$ with $m\leq n$, we define 
$$
e_{fg}
:=
\begin{cases}
d_x\cdot (f\otimes i_X^{\otimes n-m})g^\dag
&
\text{if }m<n
\\
d_x\cdot fg^\dag
&
\text{if }m=n
\\
d_x\cdot f(g\otimes i_X^{\otimes m-n})^\dag
&
\text{if }m>n,
\end{cases}
$$
which lies in $\fA(\Lambda_{\max\{m,n\}})$.
The result follows.
\end{proof}

\begin{thm}
\label{thm:BraidedCategoricalNetHaagDuality}
The net of von Neumann algebras $(A_\Lambda)$ satisfies Haag duality.
Moreover, each cone algebra $A_\Lambda$ is a finite direct sum of type $\rm I_\infty$ factors, with one summand for each simple object in the full fusion subcategory of $\cB$ generated by $X$.
\end{thm}
\begin{proof}
Consider $A_\Lambda$ acting in the ground state representation
$L^2(\fA,\psi)\cong \cK_{1_\cB}$.
Let $\Delta_n$ be an increasing family of rectangles where $\bigcup \Delta_n$ covers the entire lattice,
and set $\Lambda_n:=\Delta_n\cap \Lambda$.
By semisimplicity, writing $X^\vee$ for the dual of $X$,
\begin{align*}
\cK_{1_{\cB}}(\Delta_n) 
&=
\cB(1_\cB\to X^{\otimes \Delta_n})
\\&\cong
\cB((X^\vee)^{\otimes\Delta_n\setminus\Lambda} \to X^{\otimes \Lambda_n})
\\&\cong
\bigoplus_{x\in \Irr\langle X\rangle} \cB( (X^\vee)^{\otimes\Delta\setminus\Lambda} \to x)\otimes \cB(x\to X^{\otimes \Lambda_n})
\\&=
\bigoplus_{x\in \Irr\langle X\rangle} \cB((X^\vee)^{\otimes\Delta_n\setminus\Lambda} \to x)\otimes \cK_x(\Lambda_n)
\\&\cong
\bigoplus_{x\in \Irr\langle X\rangle} \cB(x^\vee\to X^{\otimes\Delta_n\setminus\Lambda})\otimes \cK_x(\Lambda_n)
\\&\cong
\bigoplus_{x\in \Irr\langle X\rangle} \cK_{x^\vee}(\Delta_n\setminus\Lambda)\otimes \cK_x(\Lambda_n)
\end{align*}
(There is a subtlety with the inner product on each summand above where one normalizes by $d_X$ as in the proof of Lemma \ref{lem:IrreducbleConeAlgebraRepresentation}, but this normalization does not affect the rest of this argument.)
One may visualize this isomorphism as splitting each map as follows.
$$
\tikzmath{
\draw (-.3,.3) --node[left]{$X^{\otimes \Delta_n\setminus\Lambda}$} (-.3,.7);
\draw (.3,.3) --node[right]{$X^{\otimes \Lambda_n}$} (.3,.7);
\roundNbox{}{(0,0)}{.3}{.3}{.3}{$f$}
}
=
\sum_{x\in\Irr\langle X\rangle}\sum_{i=1}^{n_x}
\tikzmath{
\draw (0,.3) --node[left]{$X^{\otimes \Delta_n\setminus\Lambda}$} (0,.7);
\draw (1,.3) --node[right]{$X^{\otimes \Lambda_n}$} (1,.7);
\draw[thick, cyan] (0,-.3) node[left,yshift=-.2cm]{$\scriptstyle x^\vee$} arc(-180:0:.5cm) node[right,yshift=-.2cm]{$\scriptstyle x$};
\roundNbox{}{(0,0)}{.3}{0}{0}{$f_x^i$}
\roundNbox{}{(1,0)}{.3}{0}{0}{$g_x^i$}
}
$$
The $\fA(\Lambda_n)$-action is by post-composition on the Hilbert space $\cB(x\to X^{\otimes \Lambda_n})=\cK_x(\Lambda_n)$,
and each space $\cK_{x^\vee}(\Delta_n\setminus\Lambda)$ acts as a multiplicity space for the $\fA(\Lambda_n)$ action.
Observe that the left $\fA(\Lambda_n)$-action may never change the internal label $x\in\Irr\langle X\rangle$, and we thus see that $\cK_{1_\cB}$ decomposes as an $\fA(\Lambda)$-representation as 
$$
\bigoplus_{x\in\Irr\langle X\rangle} \cK_{x^\vee}(\Lambda^c) \otimes \cK_x(\Lambda).
$$
We know that the action of $\fA(\Lambda)$ on each $\cK_x(\Lambda)$ is irreducible, so $A_\Lambda\cong \bigoplus_{x\in\Irr\langle X\rangle} B(\cK_x(\Lambda))$ is a finite direct sum of type $\rmI_\infty$ factors as claimed.
The commutant is clearly the action by $\fA(\Lambda^c)$ on the multiplicity spaces.
\end{proof}

Thus our net of von Neumann algebras $(A_\Lambda)$ is isotone, properly infinite, and satisfies Haag duality.
By \cite[Thm.~A]{MR4927814}, we have a well-defined braided $\rmW^*$-tensor category $\mathsf{SSS}(A)$ of superselection sectors.
The category $\mathsf{SSS}(A)$ has as objects families of injective normal representations $\{ \pi_\Lambda \}_\Lambda$ of $A_\Lambda$ which are compatible with the net structure and are local, in the sense that $[\pi_\Lambda(A_\Lambda), \pi_\Delta(A_\Delta)] = 0$ for disjoint cones/half-infinite intervals $\Lambda, \Delta$.
The morphisms in the cateogry are the intertwiners between such representations.
By \cite[\S6.2]{MR4927814}, $\mathsf{SSS}(A)$ is equivalent to the full $\rmW^*$-subcategory of Hilbert space representations $K\in\Rep(\fA)$ of the quasi-local algebra $\fA$ which satisfy the following criterion:
\begin{itemize}
\item (superselection) 
For any cone/half-infinite interval $\Lambda$, the representations $K$ and $\cK_{1}\cong L^2(\fA,\psi)$ are unitarily isomorphic as $A_\Lambda$-modules.
\end{itemize}    

We now characterize the underlying $\rmW^*$-category of $\mathsf{SSS}(A)$.
Clearly the map $z\mapsto \cK_z$ is a unitary functor $\langle X\rangle^{\rm op} \to \mathsf{SSS}(A)$, as precomposition with $f\in \cB(x\to y)$ commutes with post-composition by $\End_\cB(X^{\otimes\Lambda})$.

\begin{thm}
\label{thm:SSSforBraidedCategoricalNet}
The functor $\Hilb\langle X\rangle^{\rm op}\ni z\mapsto \cK_z \in \mathsf{SSS}(A)$ is a unitary equivalence of $\rmW^*$-categories.
\end{thm}
\begin{proof}
We show $\cK$ is fully faithful and essentially surjective.
We will do this first for the case of the 1D fusion spin chain associated to the underlying UFC $\cC$ of $\langle X\rangle$, and we will then argue that this treatment also applies to the case of our 2D braided categorical net associated to $(\cB,X)$.

The 1D analog of the functor $\cK:\Hilb(\cC)^{\rm op}\to \mathsf{SSS}(A)$ from \cite[Const.~6.13]{MR4945955} yields superselection sectors of the 1D net of algebras coming from the fusion spin chain associated to $(\cC,X)$ by Proposition \ref{prop:NormalActionOfConeAlgebra}.
The argument that this functor is fully faithful follows exactly as in \cite[Thm.~6.14]{MR4945955} (ignoring the module-action there, as we only care about the category itself).

We now give a direct proof of essential surjectivity.
First, we consider the left and right halves $\fA_L,\fA_R$ of the quasi-local algebra $\fA$ of our fusion spin chain.
Arguing as in Appendix \ref{appendix:AFActionsOfUmFCs} and Remark \ref{rem:NoNeedForvNA}, we get fully faithful tensor functors $\cC\to \Bim(\fA_L)$ and $\cC^{\rm mp}\to \Bim(\fA_R)$ (here $\Bim$ means Hilbert $\rmC^*$-bimodules as before), giving a fully faithful tensor functor
\begin{equation}
\label{eq:C*Action}    
\cC\boxtimes \cC^{\rm mp}\hookrightarrow \Bim(\fA_L\otimes \fA_R).
\end{equation}

Given a superselection sector $K\in \mathsf{SSS}(A)$, we may consider the restriction of the $\fA$-action to $\fA_L\otimes \fA_R$, giving an object of $\Rep(\fA_L\otimes \fA_R)$.
By \cite[Thm.~4.13]{MR4814692}, $\fA$ is the \emph{Q-system realization} \cite{MR4419534} of the symmetric enveloping algebra Q-system 
$S\in \cC\boxtimes \cC^{\rm mp}$ \cite[Ex.~5.33]{MR3687214} under the fully faithful tensor functor
$\cC\boxtimes \cC^{\rm mp}\to \Bim(\fA_L\otimes \fA_R)$.
Now one of the main results of \cite{MR4419534} is that the Q-system realization functor $|\,\cdot\,|: \QSys(\CStar)\to \CStar$ is an equivalence of $\rmC^*$-2-categories.
Thus we automatically get an equivalence of 
$\rmC^*$-categories
\[
\left\{
\parbox{5cm}{
$S$-modules in $\Rep(\fA_L\otimes \fA_R)=\CStar(\fA_L\otimes \fA_R\to \bbC)$
}
\right\}
=
\QSys(\CStar)(S\to 1_\bbC)
\xrightarrow[\cong]{|\,\cdot\,|}
\CStar(|S|\to \bbC)
=
\Rep(\fA).
\]
This means that we may identify the full $\rmW^*$-subcategory $\mathsf{SSS}(A)\subset \Rep(\fA)$ as a full subcategory of the $\rmW^*$-category $\Mod_{\Rep(\fA_L\otimes \fA_R)}(S)$ of $S$-modules in $\Rep(\fA_L\otimes \fA_R)$.

Now using the superselection criterion, one sees that when considering 
\[
\mathsf{SSS}(A)\subset \Mod_{\Rep(\fA_L\otimes \fA_R)}(S),
\]
the $\fA_L$ and $\fA_R$ actions extend to $A_L$ and $A_R$ respectively, which are finite direct sums of type $\rmI$ factors whose centers have dimension $|\Irr(\cC)|$.
Hence we may identify $\Rep(A_L\otimes A_R)\cong \Hilb(\cC\boxtimes \cC^{\rm mp})$.
By Theorem \ref{thm:BraidedCategoricalNetHaagDuality}, we see that the $A_L\vee A_R$-action on $\cK_{1_\cC}\cong L^2(\fA, \psi)$ extends to $A_L\otimes A_R$, and by a localization argument, this is also true on $K$.
Again arguing as in Appendix \ref{appendix:AFActionsOfUmFCs}, the functor \eqref{eq:C*Action} extends to a functor
\[
\cC\boxtimes \cC^{\rm mp}\to \Bim(A_L\otimes A_R),
\]
endowing $\Rep(A_L\otimes A_R)$ with the structure of a $\cC\boxtimes \cC^{\rm mp}$-module category, which we may identify with the usual $\cC\boxtimes \cC^{\rm mp}$-action on $\Hilb(\cC\boxtimes \cC^{\rm mp})$.

Now putting it all together, our $K\in \mathsf{SSS}(A)$ both lies in $\Hilb(\cC\boxtimes \cC^{\rm mp})$ and is equipped with an $S$-action.
But by \cite[Remark 8.2]{MR1966525}, $S\in \cC\boxtimes \cC^{\rm mp}$ is dual to the canonical Lagrangian algebra $I(1)\in Z(\cC)$ where $I$ is adjoint to the forgetful functor $Z(\cC)\to \cC$ under the Morita equivalence $\cC$ as a $Z(\cC)-\cC\boxtimes \cC^{\rm mp}$ bimodule category.
We thus have an equivalence of $\rmW^*$-categories
\[
\Mod_{\Hilb(\cC\boxtimes \cC^{\rm mp})}(S) \cong \Hilb(\cC)^{\rm op}.
\]
This means that $K$ lies in the image of $\Hilb(\cC)^{\rm op}$ as claimed.

To extend this argument to our 2D braided categorical net, we note that any dimensional reduction maps local operators to local operators, and $\mathsf{SSS}(A)$ is just a full subcategory of $\Rep(\fA)$.
Thus every choice of dimensional reduction as in \eqref{eq:DimensionalReduction} gives a fully faithful functor
\[
\mathsf{SSS}^{2D}(A) \hookrightarrow \mathsf{SSS}^{1D}(A).
\]
Indeed, the superselection criterion for the 1D fusion spin chain follows from the 2D superselection criterion, since every half-infinite line segment of $\bbZ$ is in the image of some cone under dimensional reduction.
We thus have a composite of fully faithful functors
\[
\Hilb(\cC)^{\rm op}
\hookrightarrow
\mathsf{SSS}^{2D}(A) 
\hookrightarrow 
\mathsf{SSS}^{1D}(A)
\cong
\Hilb(\cC)^{\rm op}
\]
which composes to the identity, where the last equivalence follows from the 1D argument above.
The result follows.
\end{proof}

The identification of the fusion and braiding of $\mathsf{SSS}(A)$ with that of $\Hilb\langle X\rangle$ has recently appeared in \cite{2609.20725}.

\subsection{Bounded spread isomorphism to edge-restricted Levin-Wen operators}
\label{sec:BoundedSpreadIso}

\begin{nota}
For a region $\Lambda$, $\Lambda^{+s}$ is the smallest region containing all lattice points of distance at most $s$ away from $\Lambda$.
When $\Lambda$ is a rectangle, $\Lambda^{+s}$ is another rectangle.
When $\Lambda$ is a cone, $\Lambda^{+s}$ means the cone obtained from $\Lambda$ by moving the boundaries of the cone perpendicularly by $s$ to obtain a larger cone.
\end{nota}

\begin{defn}
Suppose $\fA$ and $\fB$ are two nets of finite dimensional $\rm C^*$-algebras on the same lattice.
A \emph{bounded spread isomorphism} $\alpha:\fA\to \fB$ is a $*$-algebra isomorphism between the quasi-local algebras such that there is a universal constant $s>0$ called the \emph{spread} such that $\alpha(\fA(\Lambda)) \subset \fB(\Lambda^{+s})$ 
and
$\alpha^{-1}(\fB(\Lambda)) \subset \fA(\Lambda^{+s})$ 
for all bounded contractible $\Lambda$.
\end{defn}

One utility of bounded spread isomorphism between nets of algebras is that bounded spread Haag duality for one net implies the same for the other \cite[Prop.~5.10]{MR4927814} in a given state.
Thus the two nets of algebras have isomorphic braided $\rmW^*$-tensor categories of superselection sectors \cite[Thm.~C]{MR4927814}.

\begin{ex}[Intertwined nets of algebras]
Suppose $\fA,\fB$ are two nets of finite dimensional $\rmC^*$-algebras on the same lattice.
We say that $\fA$ and $\fB$ are \emph{intertwined with spread $t>0$} if $\fA(\Lambda)\subseteq \fB(\Lambda^{+t})$ and $\fB(\Lambda)\subseteq \fA(\Lambda^{+t})$ for all bounded contractible $\Lambda$. 
Choosing increasing finite rectangles whose union is all of $\cL$, we see that we have an \emph{equality} of quasi-local algebras $\fA=\fB$, and the identity map is manifestly bounded spread.
\end{ex}

\begin{subex}[{\cite[\S~IV.E]{MR4808260}}]
\label{subex:EdgeRestricted}
Consider the \emph{edge restricted local Hilbert spaces} of the Levin-Wen model for a UFC $\cC$ given by $p^A_\Lambda\bigotimes_{v\in\Lambda} \cH_v$ where $p^A_\Lambda:=\prod_{\ell\in \Lambda} A_\ell$.
It was shown in \cite[Prop.~2.6]{2305.14068} that
$$
p^A_\Lambda\bigotimes_{v\in\Lambda} \cH_v
\cong
\cC(X^{\partial_E\Lambda} \to F(A)^{\# p\subset\Lambda})
\cong
Z(\cC)(\Tr(X^{\partial_E \Lambda}) \to A^{\#p\subset \Lambda}),
$$
where $X^{\partial_E \Lambda}$ is the tensor product of our generator $X=\bigoplus_{c\in\Irr(\cC)}c$ over the number of edges in $\partial_E \Lambda$, $A\in Z(\cC)$ is the canonical Lagrangian algebra given by $A=\operatorname{Tr}(1_\cC)$ where $\Tr: \cC\to Z(\cC)$ is the unitary adjoint \cite[\S2.1]{MR4750417} of the forgetful functor $F: Z(\cC)\to \cC$, and $\#p\subset \Lambda$ is the number of plaquettes in $\Lambda$.

We now consider two nets of algebras.
The first is given by 
$$
\fA(\Lambda):= \End_{Z(\cC)}(A^{\#p\subset \Lambda}),
$$
which is exactly the commutant of 
$\End_{Z(\cC)}(\Tr(X^{\otimes \partial_E \Lambda}) )$
on the above Hilbert space by Lemma~\ref{lem:SemisimpleCommutation}.
Indeed, $X$ contains all simples in $\cC$, so $\Tr(X^{\otimes \partial_E \Lambda})$ contains all simples in $Z(\cC)$.
It is worth mentioning that $\End_{Z(\cC)}(\Tr(X^{\otimes \partial_E \Lambda}) )$ is Morita equivalent to Ocneanu's tube algebra, which is given by
$$
\bigoplus_{a,b,c\in\Irr(\cC)} \cC(ac\to cb)
\cong
\bigoplus_{a,b,c\in\Irr(\cC)} \cC(a\to cbc^\vee)
\cong
\cC(X\to F(\Tr(X))
\cong
\End_{Z(\cC)}(\Tr(X)).
$$
Indeed, we define the algebra
$$
\Tube_\cC(\partial_E\Lambda)
=
\left\{
\tikzmath{
\draw[thick, rounded corners=5pt] (0,0) rectangle (3,3);
\draw[thick, blue, rounded corners=5pt] (.5,.5) rectangle (2.5,2.5);
\draw[thick, rounded corners=5pt] (1,1) rectangle (2,2);
\foreach \x in {1.2,1.4,1.6,1.8}{
\draw[thick, blue] (\x,0) -- (\x,1);
\draw[thick, blue] (\x,2) -- (\x,3);
\draw[thick, blue] (0,\x) -- (1,\x);
\draw[thick, blue] (2,\x) -- (3,\x);
}
\node at (1.5,1.5) {$\Lambda$};
}
\right\}
$$
which acts by gluing around the boundary and resolving in the usual way using semisimplicity.
One shows that
$$
\Tube_\cC(\partial_E\Lambda) \cong \End_{Z(\cC)}(\Tr(X^{\partial_E \Lambda})).
$$

We can visualize the action of $\fA(\Lambda)$ on vectors in $p^A_\Lambda\bigotimes_{v\in\Lambda}\cH_v$ as follows.
First, by \cite{MR3578212}, the canonical Lagrangian $A=\Tr(1_\cC)\in Z(\cC)$ should be viewed as an empty tube corresponding to a `curled up' empty sheet labelled by $\cC$.
One thus thinks of the edge-restricted Levin-Wen Hilbert space as having tubes emanating upwards from every plaquette.
$$
\tikzmath{
\draw (-.5,-.5) grid (3.5,3.5);
\foreach \x in {.3,1.3,2.3}{
\foreach \y in {1.1,2.1,3.1}{
\filldraw[white] ($ (\x,\y) + (0,-.35) $) rectangle ($ (\x,\y) + (.4,0) $);
\draw[thick] (\x,\y) -- ($ (\x,\y) + (0,-.35) $) .. controls ++(270:.2cm) and ++(70:.1cm) .. ($ (\x,\y) + (-.1,-.6) $);
\draw[thick] ($ (\x,\y) + (.4,0) $) -- ($ (\x,\y) + (.4,-.35) $) .. controls ++(270:.2cm) and ++(110:.1cm) .. ($ (\x,\y) + (.5,-.6) $) ;
\filldraw[thick, fill=cyan!30] ($ (\x,\y) + (.2,0) $) ellipse (.2 and .1);
}}
}
\qquad\longrightarrow\qquad
\tikzmath{
\draw (-.5,-.5) grid (3.5,3.5);
\foreach \x in {.3,1.3,2.3}{
\foreach \y in {1.1,2.1,3.1}{
\filldraw[white] ($ (\x,\y) + (0,-.35) $) rectangle ($ (\x,\y) + (.4,0) $);
\draw[thick] (\x,\y) -- ($ (\x,\y) + (0,-.35) $) .. controls ++(270:.2cm) and ++(70:.1cm) .. ($ (\x,\y) + (-.1,-.6) $);
\draw[thick] ($ (\x,\y) + (.4,0) $) -- ($ (\x,\y) + (.4,-.35) $) .. controls ++(270:.2cm) and ++(110:.1cm) .. ($ (\x,\y) + (.5,-.6) $) ;
\filldraw[thick, fill=red!30] ($ (\x,\y) + (.2,0) $) ellipse (.2 and .1);
}}
\filldraw[thick, rounded corners=5pt, fill=gray!50, opacity=.8] (.1,.8) rectangle (2.9,3.4);
\foreach \x in {.3,1.3,2.3}{
\foreach \y in {1.1,2.1,3.1}{
\filldraw[white, opacity=.8] ($ (\x,\y) $) arc(-180:0:.2 and .1) -- ($ (\x,\y) + (.4,.4) $) -- ($ (\x,\y) + (0,.4) $);
\draw[thick] (\x,\y) -- ($ (\x,\y) + (0,.4) $);
\draw[thick] ($ (\x,\y) + (.4,0) $) -- ($ (\x,\y) + (.4,.4) $);
\filldraw[thick, fill=cyan!30] ($ (\x,\y) + (.2,.4) $) ellipse (.2 and .1);
}}
}
$$
The action of $\fA(\Lambda)$ is then by gluing a local operator in the braided categorical net for $Z(\cC)$ with strand label $A$ (which is \emph{not} a strong tensor generator!) onto these empty tubes.
The resulting diagram is again in $\cC(X^{\partial_E \Lambda}\to F(A^{\#p\subset\Lambda}))$.

Second, the \emph{boundary preserving local operators} $\fB(\Lambda)$ are the endomorphisms of the edge restricted local Hilbert spaces $p^A_\Lambda\bigotimes_{v\in\Lambda} \cH_v$ which fix the labels on the external boundary edges of $\Lambda$.
Since operators in $\fA(\Lambda)$ clearly fix the external boundary, we have an obvious inclusion 
$$
\fA(\Lambda)\subset \fB(\Lambda).
$$
Note, however, that operators in $\fB(\Lambda)$ need not preserve the internal edge labels, but must map $p^A_\Lambda\bigotimes_{v\in\Lambda} \cH_v$ back to itself.\footnote{The algebra $\fB(\Lambda)$ should be viewed as the local operators which preserve the `gauge-constrained' subspace $p^A_\Lambda\bigotimes_{v\in\Lambda} \cH_v$  \cite[\S IV.E]{MR4808260}.
However, in that article, it was elided that the algebras $\fB(\Lambda)$ must map the `gauge-invariant' edge-restricted subspace $p^A_\Lambda\bigotimes_{v\in\Lambda} \cH_v$ back to itself.}

Observe now that operators in $\fB(\Lambda)$ included into $\fB(\Lambda^{+1})$ commute with the external action of $\Tube_\cC(\partial_E\Lambda^{+1})$!
Since 
$$
\fA(\Lambda^{+1})=\End_{\Tube_\cC(\partial_E\Lambda^{+1})}(\cC(X^{\otimes\partial_E \Lambda^{+1}}\to F(A^{\#p\subset\Lambda^{+1}})),
$$
we see that $\fB(\Lambda)\subset \fA(\Lambda^{+1})$.
Thus the nets of algebra $\fA,\fB$ are intertwined.
\end{subex}

In order to use the above intertwining to make a statement about an isomorphism of categories of superselection sectors, we must analyze the canonical states on $\fA$ and $\fB$.
Recall that $\psi_\fA$ on $\fA$ from \S\ref{sec:SSSofBraidedCategoricalNet} above is defined by conjugating by unit maps $1_{Z(\cC)} \to A$ for each site in $\Lambda$ and taking inductive limits.
The state $\psi_\fB$ is given on $\fB(\Lambda)$ by compressing by $p_\Delta$ for $\Lambda \Subset_1 \Delta$ as in \S\ref{sec:CanonicalStateAndBoundaryAlgebra}:
$$
\psi_\fB(f)p_\Delta = p_\Delta fp_\Delta
\qquad\qquad\qquad
\forall\, f\in \fB(\Lambda).
$$
Observe that for the image of $\fA(\Lambda)\subset \fB(\Lambda)$, compressing by $p_\Delta$ implements the conjugation by the unit maps by \cite[Rem.~2.10]{2305.14068}.
In more detail, in the graphical calculus of strings on tubes for the categorified trace from \cite{MR3578212}, we have
$$
\frac{1}{D_\cC}
\sum_{c\in\Irr(\cC)}
d_c
\cdot
\tikzmath{
\draw[thick] (0,0) -- (0,1);
\draw[thick] (1,0) -- (1,1);
\draw[thick] (0,0) arc(-180:0:.5 and .2);
\draw[thick, dotted] (0,0) arc(180:0:.5 and .2);
\draw[thick, cyan] (0,.5) arc(-180:0:.5 and .2);
\draw[thick, cyan, dotted] (0,.5) arc(180:0:.5 and .2);
\node[cyan] at (.5,.45) {$\scriptstyle c$};
\filldraw[thick, fill=white] (.5,1) ellipse (.5 and .2);
}
=
\tikzmath{
\draw[thick] (0,0) arc (180:0:.5 and .4);
\draw[thick] (0,1) arc (-180:0:.5 and .4);
\draw[thick] (0,0) arc(-180:0:.5 and .2);
\draw[thick, dotted] (0,0) arc(180:0:.5 and .2);
\filldraw[thick, fill=white] (.5,1) ellipse (.5 and .2);
}
=
i\circ i^\dag.
$$
Thus compressing an operator $f\in \fA(\Lambda)$ by $p_\Delta$ gives
$$
p_\Delta f p_\Delta = \psi_\fA(f) p_\Delta,
$$
which implies $\psi_\fB(f)=\psi_\fA(f)$.
We thus see that $\psi_\fA=\psi_\fB$ on the inductive limit algebras $\fA,\fB$.

As $\fA,\fB$ are intertwined and $\fA$  satisfies Haag duality in the state $\psi$ by Theorem \ref{thm:BraidedCategoricalNetHaagDuality}, $\fB$ satisfies bounded spread Haag duality in the state $\psi$ by \cite[Prop.~5.10]{MR4927814}.
Thus $\fB$ has a well-defined braided $\rmW^*$-tensor category of superselection sectors by \cite{MR4362722} or \cite[Thm.~B]{MR4927814}.
By \cite[Thm.~C]{MR4927814}, the braided $\rmW^*$-tensor categories of superselection sectors for $\fB$ and $\fA$ are isomorphic.
By Theorem \ref{thm:SSSforBraidedCategoricalNet}, the underlying $\rmW^*$-category of the superselection sectors for $\fA$ is equivalent to the full subcategory of $\Hilb(Z(\cC))$ generated by $A$, i.e., $\Hilb\langle A\rangle$.
We thus have the following theorem.

\begin{thm}
\label{thm:OnlyFluxes}
The superselection sector category for the net of boundary preserving local operators in the Levin-Wen model is equivalent to $\Hilb\langle A\rangle\subset \Hilb(Z(\cC))$, the full subcategory generated by the canonical Lagrangian algebra $A=\Tr(1)$.
\end{thm}

\begin{rem}
Note that the above theorem only establishes an isomorphism of $\rmW^*$-categories, and not braided $\rmW^*$-tensor categories, as we have not yet identified composition of superselection sectors for the braided categorical net with the tensor product and braiding in $\Hilb(Z(\cC))$.
\end{rem}

\begin{rem}
The reason we consider the edge-restricted Levin-Wen operators is that this model is bounded-spread equivalent to the braided fusion categorical nets we defined earlier.
The superselection structure of the usual (unrestricted) Levin-Wen model, including the braiding, has recently been worked out by Bols and Kj{\ae}r~\cite{2511.21521,2603.01936}.
The physical boundary algebra for the full model was studied in our previous work~\cite{MR4945955}, where it was found that its DHR bimodules are given by $Z(\cC)$, matching the bulk topological superselection theory.
\end{rem}

\subsection{Adapting the Levin-Wen model to see all excitations in the edge-restricted sector}

Unfortunately, only the fluxes (summands of the canonical Lagrangian algebra) appear as superselection sectors in Theorem \ref{thm:OnlyFluxes} above.
The model from \cite[Rem.~3.5]{2305.14068} due to the first author is a generalization of the Levin-Wen model where we can see charges and dyons as superselection sectors of the net of boundary preserving local algebras, and not just fluxes.
The local Hilbert space is given by
$$
\tikzmath{
\draw (-.5,0) node[left]{$\scriptstyle a$} -- (.5,0) node[right]{$\scriptstyle d$};
\draw (0,-.5) node[below]{$\scriptstyle b$} -- (0,.5) node[above]{$\scriptstyle c$};
\draw[thick, red] (0,0) -- (-.3,.3) node[above]{$\scriptstyle z$};
}
\qquad
\longleftrightarrow
\qquad
\cH_v 
:= 
\bigoplus_{\substack{
a,b,c,d\in \Irr(\cC)
\\
z\in \Irr(Z(\cC))
}}
\cC(ab \to F(z)cd)
=
\cC(X^{\otimes 2} \to F(Y)\otimes X^{\otimes 2})
$$
where $F: Z(\cC) \to \cC$ is the forgetful functor, the direct sum is orthogonal, i.e., the summands for distinct objects in $Z(\cC)$ are orthogonal, $X=\bigoplus_{c\in \Irr(\cC)} c$, and
$Y=\bigoplus_{z\in \Irr(Z(\cC))} z$.
The space $\cH_v$ is equipped with the enriched skein module inner product
$$
\left\langle
\tikzmath{
\draw (-.5,0) node[left]{$\scriptstyle a$} -- (.5,0) node[right]{$\scriptstyle d$};
\draw (0,-.5) node[below]{$\scriptstyle b$} -- (0,.5) node[above]{$\scriptstyle c$};
\draw[thick, red] (0,0) -- (-.3,.3) node[above]{$\scriptstyle z$};
\node at (.2,-.2) {$\scriptstyle \xi$};
}
\middle|
\tikzmath{
\draw (-.5,0) node[left]{$\scriptstyle a'$} -- (.5,0) node[right]{$\scriptstyle d'$};
\draw (0,-.5) node[below]{$\scriptstyle b'$} -- (0,.5) node[above]{$\scriptstyle c'$};
\draw[thick, red] (0,0) -- (-.3,.3) node[above]{$\scriptstyle z'$};
\node at (.2,-.2) {$\scriptstyle \xi'$};
}
\right\rangle
=
\delta_{a=a'}
\delta_{b=b'}
\delta_{c=c'}
\delta_{d=d'}
\delta_{z=z'}
\frac{1}{\sqrt{\textcolor{red}{d_z}d_ad_bd_cd_d}}
\cdot
\tr_\cC(\xi^\dag\circ \xi').
$$

As in the usual Levin-Wen model, we have edge terms $A_\ell$ which enforce that string labels match.
We also have plaquette terms which must incorporate the half-braiding with vertical edges as in \cite[\S~II.B]{MR4640433}.
$$
\frac{1}{D_\cC}\sum_r
d_r
\tikzmath{
\draw[step=1.0,black,thin] (0.5,0.5) grid (2.5,2.5);
\filldraw[thick, cyan, rounded corners=5pt, fill=cyan!20] (1.15,1.15) rectangle (1.85,1.85);
\draw[thick, knot, red] (2,1) -- ($ (2,1) + (-.4,.4)$);
\draw[thick, red] (1,1) -- ($ (1,1) + (-.4,.4)$);
\draw[thick, red] (1,2) -- ($ (1,2) + (-.4,.4)$);
\draw[thick, red] (2,2) -- ($ (2,2) + (-.4,.4)$);
\node at (2.3,.8) {$\scriptstyle \xi_{2,1}$};
\node at (.7,1.8) {$\scriptstyle \xi_{1,2}$};
\node at (1.3,.8) {$\scriptstyle \xi_{1,1}$};
\node at (2.3,1.8) {$\scriptstyle \xi_{2,2}$};
\node[cyan] at (1.3,1.5) {$\scriptstyle r$};
}
=
\frac{1}{D_\cC}\sum_r
d_r
\tikzmath{
\draw[step=1.0,black,thin] (0.5,0.5) grid (2.5,2.5);
\draw[thick, cyan] (1.3,1) -- (1,1.3);
\draw[thick, cyan] (1.7,1) -- (2,1.3);
\draw[thick, cyan] (1.3,2) -- (1,1.7);
\draw[thick, cyan] (1.7,2) -- (2,1.7);
\filldraw[fill=green] (1.3,1) circle (.05cm);
\filldraw[fill=green] (1.7,1) circle (.05cm);
\filldraw[fill=orange] (2,1.3) circle (.05cm);
\filldraw[fill=orange] (2,1.7) circle (.05cm);
\filldraw[fill=yellow] (1.3,2) circle (.05cm);
\filldraw[fill=yellow] (1.7,2) circle (.05cm);
\filldraw[fill=purple] (1,1.3) circle (.05cm);
\filldraw[fill=purple] (1,1.7) circle (.05cm);
\fill[white] (1.85,1.15) circle (.05cm);
\draw[thick, red] (2,1) -- ($ (2,1) + (-.4,.4)$);
\draw[thick, red] (1,1) -- ($ (1,1) + (-.4,.4)$);
\draw[thick, red] (1,2) -- ($ (1,2) + (-.4,.4)$);
\draw[thick, red] (2,2) -- ($ (2,2) + (-.4,.4)$);
\node at (2.3,.8) {$\scriptstyle \xi_{2,1}$};
\node at (.7,1.8) {$\scriptstyle \xi_{1,2}$};
\node at (1.3,.8) {$\scriptstyle \xi_{1,1}$};
\node at (2.3,1.8) {$\scriptstyle \xi_{2,2}$};
}
=
\sum_\eta
C(\xi,\eta)
\tikzmath{
\draw[step=1.0,black,thin] (0.5,0.5) grid (2.5,2.5);
\draw[thick, red] (2,1) -- ($ (2,1) + (-.4,.4)$);
\draw[thick, red] (1,1) -- ($ (1,1) + (-.4,.4)$);
\draw[thick, red] (1,2) -- ($ (1,2) + (-.4,.4)$);
\draw[thick, red] (2,2) -- ($ (2,2) + (-.4,.4)$);
\node at (2.3,.8) {$\scriptstyle \eta_{2,1}$};
\node at (1.3,1.8) {$\scriptstyle \eta_{1,2}$};
\node at (1.3,.8) {$\scriptstyle \eta_{1,1}$};
\node at (2.3,1.8) {$\scriptstyle \eta_{2,2}$};
}
$$
There is also a \emph{unit term} local projector $C_v$ for each vertex $v$ which only allows $z=1_{Z(\cC)}$ at $v$.
Thus the local Hamiltonian
$$
H =
-\sum A_\ell
-\sum B_p
-\sum C_v
$$
again recovers the $\cC$-skein module as its ground state space.

We now proceed as in Sub-Example \ref{subex:EdgeRestricted} with our augmented Levin-Wen model.
We pass to the edge restricted local Hilbert spaces $p^A_\Lambda\bigotimes_{v\in\Lambda} \cH_v$ where $p^A_\Lambda:=\prod_{\ell\subset \Lambda} A_\ell$.
We again define $\fB(\Lambda)$ as the endomorphisms of the edge restricted Hilbert local Hilbert space $p^A_\Lambda\bigotimes_{v\in\Lambda} \cH_v$ which fix the labels on the external $\cC$-edge boundary $\partial_E\Lambda$.
By an argument similar to \cite[Prop.~2.6]{2305.14068},
$$
p^A_\Lambda\bigotimes_{v\in\Lambda} \cH_v
=
\cC(X^{\partial_E \Lambda} \to F(A)^{\#p\subset\Lambda}\otimes F(Y)^{\#v\in \Lambda})
\cong
Z(\cC)(\Tr(X^{\partial_E \Lambda}) \to A^{\#p\subset\Lambda}\otimes Y^{\#v\in \Lambda})
$$
where again $F:Z(\cC)\to \cC$ is the forgetful functor, $\Tr:\cC\to Z(\cC)$ is its unitary adjoint, and
$A=\Tr(1)$ is the canonical Lagrangian algebra.
Using a similar intertwining argument, the boundary preserving local operators $\fB(\Lambda)$ contains
$$
\fA(\Lambda)
:=
\End_{\Tube_\cC(\partial_E \Lambda)}(Z(\cC)(\Tr(X^{\partial_E \Lambda}) \to A^{\#p\subset\Lambda}\otimes Y^{\#v\in \Lambda}))
=
\End_{Z(\cC)}(A^{\#p\in\Lambda}\otimes Y^{\#v\in \Lambda}),
$$
and $\fB(\Lambda)\subset \fA(\Lambda^{+1})$ as before.
One can view the net of algebras $\fA$ as a braided fusion categorical net built from 
two generators  $\textcolor{cyan}{A=\Tr(1)}$ and $\textcolor{red}{Y=\bigoplus_{z\in \Irr(Z(\cC))} z}$
on a 2D edge lattice.
$$
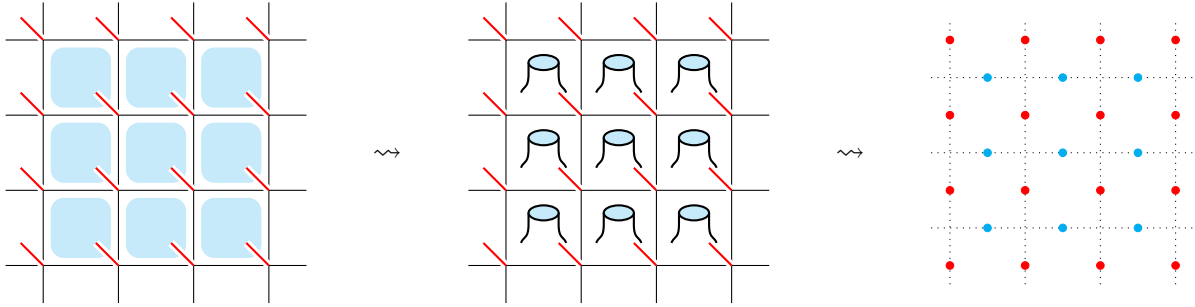
\tikzmath{
\draw[step=1.0,black,thin] (0.5,0.5) grid (4.5,4.5);
\foreach \i in {1,2,3} {
\foreach \j in {1,2,3} {
        \fill[cyan!20, rounded corners=5pt] ($ (\i,\j) + (.1,.1)$) rectangle ($ (\i,\j) + (.9,.9)$) ;
}}
\foreach \i in {1,2,3,4} {
\foreach \j in {1,2,3,4} {
        \draw[thick, knot, red] (\i,\j) -- ($ (\i,\j) + (-.3,.3)$);
}}
}
\qquad
\rightsquigarrow
\qquad
\tikzmath{
\draw (-.5,-.5) grid (3.5,3.5);
\foreach \x in {.3,1.3,2.3}{
\foreach \y in {.7,1.7,2.7}{
\filldraw[white] ($ (\x,\y) + (0,-.15) $) rectangle ($ (\x,\y) + (.4,0) $);
\draw[thick] (\x,\y) -- ($ (\x,\y) + (0,-.15) $) .. controls ++(270:.2cm) and ++(70:.1cm) .. ($ (\x,\y) + (-.1,-.4) $);
\draw[thick] ($ (\x,\y) + (.4,0) $) -- ($ (\x,\y) + (.4,-.15) $) .. controls ++(270:.2cm) and ++(110:.1cm) .. ($ (\x,\y) + (.5,-.4) $) ;
\filldraw[thick, fill=cyan!20] ($ (\x,\y) + (.2,0) $) ellipse (.2 and .1);
}}
\foreach \i in {0,1,2,3} {
\foreach \j in {0,1,2,3} {
        \draw[thick, knot, red] (\i,\j) -- ($ (\i,\j) + (-.3,.3)$);
}}
}
\qquad
\rightsquigarrow
\qquad
\tikzmath{
\draw[dotted, yshift=-.5cm] (-.25,.25) grid (3.25,3.75);
\foreach \y in {0,1,2,3}{
\foreach \x in {0,1,2,3}{
\filldraw[red] (\x,\y) circle (.05cm);
}}
\foreach \y in {0,1,2}{
\foreach \x in {0,1,2}{
\filldraw[cyan] (\x+.5,\y+.5) circle (.05cm);
}}
}
$$
The state on $\fA$ corresponding to the canonical ground state $\psi$ on $\fB$ is given by the inductive limit of the local states
$$
\phi_\Lambda(f)
:=
\tikzmath{
\foreach \y in {0,1,2}{
\foreach \x in {0,1,2}{
\draw[red, thick] (\x,\y) -- ($ (\x,\y) + (.2,-.2) $);
\filldraw[red] ($ (\x,\y) + (.2,-.2) $) circle (.05cm);
}}
\foreach \y in {.5,1.5}{
\foreach \x in {.5,1.5}{
\draw[cyan, thick] (\x,\y) -- ($ (\x,\y) + (.2,-.2) $);
\filldraw[cyan] ($ (\x,\y) + (.2,-.2) $) circle (.05cm);
}}
\filldraw[thick, rounded corners=5pt, fill=gray!50, opacity=.8] (-.2,-.2) rectangle (2.2,2.2);
\foreach \y in {0,1,2}{
\foreach \x in {0,1,2}{
\draw[red, thick] (\x,\y) -- ($ (\x,\y) + (-.2,.2) $);
\filldraw[red] ($ (\x,\y) + (-.2,.2) $) circle (.05cm);
}}
\foreach \y in {.5,1.5}{
\foreach \x in {.5,1.5}{
\draw[cyan, thick] (\x,\y) -- ($ (\x,\y) + (-.2,.2) $);
\filldraw[cyan] ($ (\x,\y) + (-.2,.2) $) circle (.05cm);
}}
\node at (1,.5) {$f$};
}
:=
(
\textcolor{cyan}{i_A^{\#p\subset\Lambda}}
\otimes
\textcolor{red}{i_Y^{\#v\in\Lambda}}
)^\dag
\circ f \circ
(
\textcolor{cyan}{i_A^{\#p\subset\Lambda}}
\otimes
\textcolor{red}{i_Y^{\#v\in\Lambda}}
)
\quad
\forall\,f\in\End_{Z(\cC)}(
\textcolor{cyan}{A^{\#p\subset\Lambda}}
\otimes 
\textcolor{red}{Y^{\#v\in\Lambda}}
).
$$
Here, the 
cyan univalent vertex is the partial isometry $i_A:1_{Z(\cC)}\hookrightarrow A$ 
or its adjoint, and the 
red univalent vertex is the partial isometry $i_Y:1_{Z(\cC)} \hookrightarrow Y$
or its adjoint.
Indeed, $\psi(f)p_\Delta= p_\Delta fp_\Delta$ for all $\Lambda \ll \Delta$,
and 
compression of $f$ by
$p_\Delta=p^B_\Delta p^C_\Delta$ is easily seen to conjugate $f$ by both 
$\textcolor{cyan}{i_A^{\#p\in\Lambda}}$
and
$\textcolor{red}{i_Y^{\#v\in\Lambda}}$ 
where appropriate,
corresponding to the compression of $f$
of both $p^C_\Delta$ and $p^B_\Delta$ respectively.

As the nets $\fA,\fB$ are intertwined, they yield isomorphic $\rmW^*$-tensor categories of superselection sectors after completing the cone von Neumann algebras in the GNS representation of the canonical ground state $\psi$ by \cite[Thm.~C]{MR4927814}.
Since $\Hilb\langle A\otimes Y\rangle =\Hilb (Z(\cC))$, we have the following theorem.

\begin{thm}
The superselection sector category for the net of boundary preserving local operators in the augmented Levin-Wen model is equivalent to $\Hilb(Z(\cC))$.
\end{thm}

We conclude that we can now see all excitations and not just those under $A\in Z(\cC)$.

\subsection{Braided-enriched fusion categorical net on a lattice with boundary}
\label{sec:BraidedEnrichedNetWithBoundary}

In \S\ref{sec:BoundaryLW} above, we saw how to construct a fusion module spin chain $\fM$ from a module category $\cM_\cC$ for a fusion category $\cC$ on a 1D $\bbZ$ lattice with boundary.
On a 2D $\bbZ^2$ lattice with a 1D $\bbZ$ boundary, we can use a \emph{braided-enriched} fusion category to construct a braided fusion module categorical net.

Braided-enriched fusion categories were introduced in \cite{MR3961709} and shown to be equivalent to \emph{module fusion categories} \cite{MR3578212,MR3866911} over a braided fusion category.
In more detail, given a unitary braided fusion category $\cA$, a unitary \emph{module fusion category} for $\cA$ or a \emph{unitary $\cA$-enriched fusion category} is a unitary fusion category $\cX$ equipped with a unitary braided central functor $\Phi^Z: \cA\to Z(\cX)$.
This central functor is exactly the data needed to hook up the 3D Walker-Wang model for $\cA$ \cite{1104.2632} to the 2D Levin-Wen model for $\cX$ \cite{PhysRevB.71.045110} as discussed in \cite{MR4640433,2305.14068}.
We will discuss this lattice model in more detail in \S\ref{sec:WW} below.

Just as one can assign a well-defined object $(a_{ij})^{\otimes \Lambda}\in \cA$ to a tuple of objects $a_{ij}\in \cA$ for $(i,j)\in\Lambda$,
we can assign an object $(x_{ij})^{\otimes\Lambda}\in\cX$ to a tuple of objects $x_{ij}\in\cX$ indexed by a subset $(i,j)\in\Lambda$ as long as $\Lambda$ meets the boundary $\partial \cL$ of our $\bbZ^2$ lattice, and $x_{ij}\in Z(\cX)$ for all sites $(i,j)$ which do not lie in the boundary $\partial \cL$, i.e., $(i,j)\in\Lambda \setminus \partial \cL$.
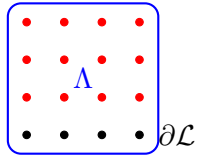
$$
\tikzmath{
\foreach \y in {.5,1,...,1.5}{
\foreach \x in {0,.5,...,1.5}{
\filldraw[red] (\x,\y) circle (.05cm);
}
}
\foreach \x in {0,.5,...,1.5}{
\filldraw (\x,0) circle (.05cm);
}
\draw[thick, blue, rounded corners=5pt] (-.25,-.25) rectangle (1.75,1.75);
\node[blue] at (.75,.75) {$\Lambda$};
\node at (2,0) {$\partial \cL$};
}
$$
That is, the lattice sites colored black belonging to $\Lambda\cap \partial \cL$ above are assigned objects $x_{ij}\in\cX$ and the lattice sites colored red belonging to $\Lambda\setminus \partial\cL$ above are assigned objects $x_{ij}\in Z(\cX)$.
As we assumed $\cX$ is $\cA$-enriched so that we have a braided central functor $\Phi^Z:\cA\to Z(\cX)$, we can label sites colored black belonging to $\Lambda\cap \partial \cL$ by objects $x_{ij}\in\cX$ and lattice sites colored red belonging to $\Lambda\setminus \partial\cL$ by objects $a_{ij}\in\cA$.
As the objects $\Phi^Z(a_{ij})\in Z(\cX)$, we thus get a well-defined object
$$
(x_{ij})^{\otimes \Lambda\cap \partial \cL}
\lhd 
(a_{ij})^{\otimes \Lambda\setminus \partial \cL}
=
(x_{ij})^{\otimes \Lambda\cap \partial \cL}
\otimes 
(\Phi^Z(a_{ij}))^{\otimes \Lambda\setminus \partial \cL}.
$$
Here, we have implicitly chosen a lexicographic ordering of the sites in $\Lambda$ as in \eqref{eq:LexicographicOrdering}, as the $x_{ij}\in\cX$ are only linearly ordered along $\partial \cL$, so that the $\Phi^Z(a_{ij})$ appear to the right, viewing $\cX$ as a right $\cA$-module (tensor) category.

\begin{defn}[Braided-enriched fusion categorical net]
Let $\cA$ be a braided unitary fusion category and let $\cX$ be an $\cA$-enriched unitary fusion category, i.e., $\cX$ is a unitary fusion category equipped with a braided unitary tensor functor $\Phi^Z: \cA\to Z(\cX)$.
We write $\Phi:= \Forget\circ \Phi^Z: \cA \to \cX$.
Let $A\in\cA$ and $X\in\cX$ be a \emph{strongly generating pair}, i.e., there are $n,k\in\bbN$ such that every simple object of $\cX$ appears as a direct summand of
$$
X^{\otimes n}\lhd A^{\otimes k}
=
X^{\otimes n} \otimes \Phi(A)^{\otimes k}.
$$
For example, we may take $X=\bigoplus_{x\in\Irr(\cX)} x$ and $A=\bigoplus_{a\in\Irr(\cA)} a$.
For rectangles $\Lambda$ which do not meet $\partial \cL$, we assign the algebra $\fX(\Lambda)$ to be the value of the braided categorical net that $(\cA,A)$ would assign to $\Lambda$, i.e., $\fX(\Lambda):=\End_\cA(A^{\otimes\Lambda})$.
For rectangles $\Delta$ which meet $\partial \cL$, we assign 
$\fX(\Delta)=\End_\cX(X^{\otimes \Delta \cap \partial \cL}\lhd A^{\otimes \Delta\setminus \partial \cL})$.
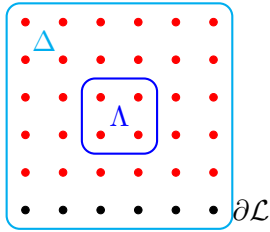
$$
\tikzmath{
\foreach \y in {.5,1,...,2.5}{
\foreach \x in {0,.5,...,2.5}{
\filldraw[red] (\x,\y) circle (.05cm);
}
}
\foreach \x in {0,.5,...,2.5}{
\filldraw (\x,0) circle (.05cm);
}
\draw[thick, cyan, rounded corners=5pt] (-.25,-.25) rectangle (2.75,2.75);
\draw[thick, blue, rounded corners=5pt] (.75,.75) rectangle (1.75,1.75);
\node[blue] at (1.25,1.25) {$\Lambda$};
\node[cyan] at (.25,2.25) {$\Delta$};
\node at (3,0) {$\partial \cL$};
}
\qquad
\qquad
\begin{aligned}
\Lambda
&\longmapsto
\fX(\Lambda)
= \End_\cA(A^{\otimes \Lambda})
\\
\Delta
&\longmapsto
\fX(\Delta)
= \End_\cX(X^{\otimes \Delta \cap \partial \cL}\lhd A^{\otimes \Delta\setminus \partial \cL})    
\end{aligned}
$$
One uses the braided central functor $\Phi^Z$ to obtain the inclusion $\fX(\Lambda) \subset \fX(\Delta)$ for such regions.
One checks that $\fX$ is a net of algebras on the $\bbZ^2$ lattice with $\bbZ$ boundary.
\end{defn}

To compute the tensor category of DHR bimodules, we use the folding trick.
Given a boundary rectangle $\Delta$ meeting $\partial \cL$,
we fold our lattice $\cL$ with boundary into a thickened 1D lattice with boundary, which we call a `1D slab'.
This process squeezes the bulk lattice sites of $\Delta$ farther into the center of this thickened 1D slab, but choosing a careful bijection $\bbZ\times \bbN\to \{0\}\cup\bbN^3$,\footnote{Here, careful means that the preimage of every initial segment of the 1D slab is disk-like in $\bbZ\times \bbN$. 
} 
we can view $\Delta$ as sitting inside a boundary interval of the 1D slab.
$$
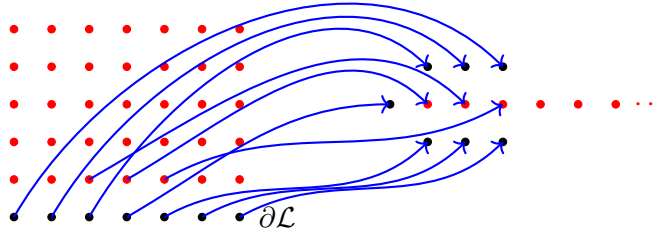
\tikzmath{
\foreach \y in {.5,1,...,2.5}{
\foreach \x in {-.5,0,.5,...,2.5}{
\filldraw[red] (\x,\y) circle (.05cm);
}
}
\foreach \x in {-.5,0,.5,...,2.5}{
\filldraw (\x,0) circle (.05cm);
}
\node at (3,0) {$\partial \cL$};
\filldraw (4.5,1.5) circle (.05cm);
\foreach \x in {5,5.5,6}{
\filldraw (\x,1) circle (.05cm);
\filldraw[red] (\x,1.5) circle (.05cm);
\filldraw[red] ($ (\x,1.5) + (1.5,0) $) circle (.05cm);
\filldraw (\x,2) circle (.05cm);
}
\draw[thick, blue, ->] (1,0) to[out=30,in=180] (4.5,1.5);
\draw[thick, blue, ->] (1.5,0) to[out=30,in=-135] (5,1);
\draw[thick, blue, ->] (2,0) to[out=30,in=-135] (5.5,1);
\draw[thick, blue, ->] (2.5,0) to[out=30,in=-135] (6,1);
\draw[thick, blue, ->] (.5,0) to[out=60,in=135] (5,2);
\draw[thick, blue, ->] (0,0) to[out=60,in=135] (5.5,2);
\draw[thick, blue, ->] (-.5,0) to[out=60,in=135] (6,2);
\draw[thick, blue, ->] (1,.5) to[out=30,in=135] (5,1.5);
\draw[thick, blue, ->] (.5,.5) to[out=30,in=135] (5.5,1.5);
\draw[thick, blue, ->] (1.5,.5) to[out=30,in=-150] (6,1.5);
\node[red] at (8,1.5) {$\cdots$};
}
$$
This process allows us to perform a dimensional reduction to consider our 2D net of algebras $\fX$ with boundary as a 1D net of algebras $\fM$ with boundary.
We use the notation $\fM$ for the 1D boundary net as the situation is completely analogous to the boundary Levin-Wen setup from \S\ref{sec:BoundaryLW}.
Although this reduction is highly non-local, it restricts to an isomorphism on the subalgebras of local operators on either side.

We can now identify $\fM$ with a fusion module categorical net as in \eqref{eq:BoundaryRectangleLW}.
The single boundary site on the left in the 1D slab corresponds to $\cM:=\cX$ as a right module for the unitary multifusion category 
$$
\cC:=\cX^{\rm mp}\boxtimes_\cA \cA \boxtimes_\cA \cX
\cong
\cX^{\rm mp}\boxtimes_\cA \cX.
$$
The single site $\{0\}$ at the left is assigned the object
\[
W:=\bigoplus_{z\in \Irr(\cX)} z
\in \cX=\cM
\]
Using the operadic approach, we may further identify the diagonal points $(n,n,n)\in \bbN^3$ as single sites in our 1D slab which are assigned the single object 
\[
X
:=
\bigoplus_{\substack{x,y\in \Irr(\cX)\\a\in \Irr(\cA)}} x^{\rm mp}\boxtimes a\boxtimes y
\in
\cX^{\rm mp}\boxtimes_\cA \cA \boxtimes_\cA \cX
=\cC
\]
Here, we use the \emph{ladder category model} for the Deligne product \cite{mitchell-thesis,MR3975865}.
In more detail $\cX^{\rm mp}\boxtimes_\cA\cX$ is the unitary multifusion category obtained by taking the unitary Cauchy completion \cite[\S3.1.1]{MR4598730} of $\cL=\mathsf{Lad}(\cX^{\rm mp},\cA,\cX)$; the latter has as objects formal symbols $x^{\rm mp}\boxtimes y$ for $x,y\in \cX$ and 
$$
\Hom(
x_1^{\rm mp}\boxtimes y_1
\to
x_2^{\rm mp}\boxtimes y_2
)
:=
\bigoplus_{a\in\Irr(\cA)}
\cX^{\rm mp}(x_1^{\rm mp} \to x_2^{\rm mp}\lhd a)
\otimes
\cX(a\rhd y_1 \to y_2)
$$
where $x^{\rm mp}\lhd a := x^{\rm mp}\otimes_{\cX^{\rm mp}} \Phi(a)^{\rm mp}=(\Phi(a)\otimes x)^{\rm mp}$ 
and $a\rhd x:= \Phi(a)\otimes x$ is the $\cA$-action on $\cX$ via $\Phi:=\Forget\circ \Phi^Z: \cA\to \cX$.

To compose morphisms, we stack ladders vertically and decompose the rungs into simples using the semisimplicity relation for $\cA$.
To depict this relation we use the paired node convention from \cite[\S2.5]{MR3663592} which involves summing over an ONB for the trivalent skein module and its adjoint, which is independent of the choice of basis.
Using the skein module inner product, the fusion/semsimplicity relation is given by
\begin{equation}
\label{eq:FusionRelationInSkeinModule}
\sqrt{d_ad_b}\cdot
\tikzmath{
\draw[thick, red] (-.3,-.6) node[below]{$\scriptstyle a$} -- (-.3,.6) node[above]{$\scriptstyle a$};
\draw[thick, red] (.3,-.6) node[below]{$\scriptstyle b$} -- (.3,.6) node[above]{$\scriptstyle b$};
}
=
\sum_{c\in\Irr(\cA)}
\sqrt{d_c}\cdot
\tikzmath{
\draw[thick, red] (-.3,-.6) node[below]{$\scriptstyle a$} -- (0,-.3) -- (.3,-.6) node[below]{$\scriptstyle b$};
\draw[thick, red] (-.3,.6) node[above]{$\scriptstyle a$} -- (0,.3) -- (.3,.6) node[above]{$\scriptstyle b$};
\draw[thick, red] (0,-.3) --node[left]{$\scriptstyle c$} (0,.3);
\filldraw[fill=blue] (0,-.3) circle (.05cm);
\filldraw[fill=blue] (0,.3) circle (.05cm);
}\,.
\end{equation}
With this notation, composition is defined thusly:
$$
\tikzmath{
\draw[dashed, cyan, thick] (0,-1) -- (0,1);
\draw (-.5,-1) --node[left]{$\scriptstyle x_2^{\rm mp}$} (-.5,-.7) -- (-.5,.7) --node[left]{$\scriptstyle x_3^{\rm mp}$} (-.5,1);
\draw (.5,-1) --node[right]{$\scriptstyle y_2$} (.5,-.7) -- (.5,.7) --node[right]{$\scriptstyle y_3$} (.5,1);
\draw[thick, red] (-.5,-.5) --node[left]{$\scriptstyle b$} (.5,.5);
\roundNbox{fill=white}{(-.5,-.4)}{.3}{0}{0}{$f_2$}
\roundNbox{fill=white}{(.5,.4)}{.3}{0}{0}{$g_2$}
}
\circ
\tikzmath{
\draw[dashed, cyan, thick] (0,-1) -- (0,1);
\draw (-.5,-1) --node[left]{$\scriptstyle x_1^{\rm mp}$} (-.5,-.7) -- (-.5,.7) --node[left]{$\scriptstyle x_2^{\rm mp}$} (-.5,1);
\draw (.5,-1) --node[right]{$\scriptstyle y_1$} (.5,-.7) -- (.5,.7) --node[right]{$\scriptstyle y_2$} (.5,1);
\draw[thick, red] (-.5,-.5) --node[left]{$\scriptstyle a$} (.5,.5);
\roundNbox{fill=white}{(-.5,-.4)}{.3}{0}{0}{$f_1$}
\roundNbox{fill=white}{(.5,.4)}{.3}{0}{0}{$g_1$}
}
:=
\tikzmath{
\draw[dashed, cyan, thick] (0,-1) -- (0,2.7);
\draw (-.5,-1) --node[left]{$\scriptstyle x_1^{\rm mp}$} (-.5,-.7) -- (-.5,.6) --node[left]{$\scriptstyle x_2^{\rm mp}$} (-.5,1) -- (-.5,2.4) --node[left]{$\scriptstyle x_3^{\rm mp}$} (-.5,2.7);
\draw (.5,-1) --node[right]{$\scriptstyle y_1$} (.5,-.7) -- (.5,.7) --node[right]{$\scriptstyle y_2$} (.5,1) -- (.5,2.4) --node[right]{$\scriptstyle y_3$} (.5,2.7);
\draw[thick, red] (-.5,-.5) --node[left]{$\scriptstyle a$} (.5,.5);
\draw[thick, red] (-.5,1.3) --node[left]{$\scriptstyle b$} (.5,2.1);
\roundNbox{fill=white}{(-.5,-.4)}{.3}{0}{0}{$f_1$}
\roundNbox{fill=white}{(.5,.4)}{.3}{0}{0}{$g_1$}
\roundNbox{fill=white}{(-.5,1.3)}{.3}{0}{0}{$f_2$}
\roundNbox{fill=white}{(.5,2.1)}{.3}{0}{0}{$g_2$}
}
\underset{\text{\eqref{eq:FusionRelationInSkeinModule}}}{=}
\sum_{c\in\Irr(\cA)}
\frac{\sqrt{d_c}}{\sqrt{d_ad_b}}
\tikzmath{
\draw[dashed, cyan, thick] (0,-2.2) -- (0,2.2);
\draw (-1,-2.2) --node[left]{$\scriptstyle x_1^{\rm mp}$} (-1,-1.9) -- (-1,-1.3) --node[left]{$\scriptstyle x_2^{\rm mp}$} (-1,-.9) -- (-1,-.3) -- (-1,1.9) --node[left]{$\scriptstyle x_3^{\rm mp}$} (-1,2.1);
\draw (1,2.2) --node[right]{$\scriptstyle y_3$} (1,1.9) -- (1,1.3) --node[right]{$\scriptstyle y_2$} (1,.9) -- (1,-.3) -- (1,-1.9) --node[right]{$\scriptstyle y_1$} (1,-2.1);
\draw[thick, red] (-1,-.6) --node[above]{$\scriptstyle b$} (-.3,-.3) --node[right]{$\scriptstyle c$} (.3,.3) --node[below]{$\scriptstyle a$} (1,.6);
\draw[thick, red] (-1,-1.6) --node[right]{$\scriptstyle a$} (-.3,-.3);
\draw[thick, red] (1,1.6) --node[left]{$\scriptstyle b$} (.3,.3);
\roundNbox{fill=white}{(-1,-1.6)}{.3}{0}{0}{$f_1$}
\roundNbox{fill=white}{(1,.6)}{.3}{0}{0}{$g_1$}
\roundNbox{fill=white}{(-1,-.6)}{.3}{0}{0}{$f_2$}
\roundNbox{fill=white}{(1,1.6)}{.3}{0}{0}{$g_2$}
\filldraw[fill=blue] (-.3,-.3) circle (.05cm);
\filldraw[fill=blue] (.3,.3) circle (.05cm);
}
$$
Our ladders include a dashed line in the middle to distinguish between both sides of the ladder in the tensor product decomposition.

The tensor product on objects is given by
$$
(x^{\rm mp}\boxtimes y)\otimes_\cL (w^{\rm mp}\boxtimes z)
:=
(x^{\rm mp}\otimes_{\cX^{\rm mp}} w^{\rm mp})
\boxtimes 
(y\otimes_\cX z)
=
(w\otimes_\cX x)^{\rm mp}
\boxtimes
(y\otimes_\cX z),
$$
and on morphisms, it is nesting ladders using that the $\cA$-action factors through $\cA\to Z(\cX)$.\footnote{It appears that our `nested' monoidal product differs from the one defined in \cite[\S3.2.1]{mitchell-thesis}, but this is a manifestation of the fact that one copy of $\cX$ has the opposite monoidal product.}
\begin{align*}
\tikzmath{
\draw[dashed, cyan, thick] (0,-1) -- (0,1);
\draw (-.5,-1) --node[left]{$\scriptstyle x_1^{\rm mp}$} (-.5,-.7) -- (-.5,.7) --node[left]{$\scriptstyle x_2^{\rm mp}$} (-.5,1);
\draw (.5,-1) --node[right]{$\scriptstyle y_1$} (.5,-.7) -- (.5,.7) --node[right]{$\scriptstyle y_2$} (.5,1);
\draw[thick, red] (-.5,-.5) --node[left]{$\scriptstyle a$} (.5,.5);
\roundNbox{fill=white}{(-.5,-.4)}{.3}{0}{0}{$f_1$}
\roundNbox{fill=white}{(.5,.4)}{.3}{0}{0}{$g_1$}
}
\otimes
\tikzmath{
\draw[dashed, cyan, thick] (0,-1) -- (0,1);
\draw (-.5,-1) --node[left]{$\scriptstyle w_1^{\rm mp}$} (-.5,-.7) -- (-.5,.7) --node[left]{$\scriptstyle w_2^{\rm mp}$} (-.5,1);
\draw (.5,-1) --node[right]{$\scriptstyle z_1$} (.5,-.7) -- (.5,.7) --node[right]{$\scriptstyle z_2$} (.5,1);
\draw[thick, red] (-.5,-.5) --node[left]{$\scriptstyle b$} (.5,.5);
\roundNbox{fill=white}{(-.5,-.4)}{.3}{0}{0}{$f_2$}
\roundNbox{fill=white}{(.5,.4)}{.3}{0}{0}{$g_2$}
}
&:=
\tikzmath{
\draw[thick, red] (-1.3,-.5) --node[left, xshift=-.1cm]{$\scriptstyle b$} (1.3,.5);
\draw (-1.3,-1.4) --node[left]{$\scriptstyle w_1^{\rm mp}$} (-1.3,-1.1) -- (-1.3,1.1) --node[left]{$\scriptstyle w_2^{\rm mp}$} (-1.3,1.4);
\draw[knot] (-.5,-1.4) --node[left]{$\scriptstyle x_1^{\rm mp}$} (-.5,-1.1) -- (-.5,1.1) --node[left]{$\scriptstyle x_2^{\rm mp}$} (-.5,1.4);
\draw[knot] (.5,-1.4) --node[right]{$\scriptstyle y_1$} (.5,-1.1) -- (.5,1.1) --node[right]{$\scriptstyle y_2$} (.5,1.4);
\draw (1.3,-1.4) --node[right]{$\scriptstyle z_1$} (1.3,-1.1) -- (1.3,1.17) --node[right]{$\scriptstyle z_2$} (1.3,1.4);
\draw[thick, red, knot] (-.5,-.8) --node[below]{$\scriptstyle a$} (.5,.8);
\draw[dashed, cyan, thick] (0,-1.4) -- (0,1.4);
\roundNbox{fill=white}{(-.5,-.8)}{.3}{0}{0}{$f_1$}
\roundNbox{fill=white}{(.5,.8)}{.3}{0}{0}{$g_1$}
\roundNbox{fill=white}{(-1.3,-.6)}{.3}{0}{0}{$f_2$}
\roundNbox{fill=white}{(1.3,.6)}{.3}{0}{0}{$g_2$}
}
\\&
\underset{\text{\eqref{eq:FusionRelationInSkeinModule}}}{=}
\sum_{c\in\Irr(\cA)}
\frac{\sqrt{d_c}}{\sqrt{d_ad_b}}
\tikzmath{
\draw[thick, red] (.4,.2) --node[above]{$\scriptstyle c$} (-.4,-.2);
\draw[thick, red] (-1.8,-.5) --node[above]{$\scriptstyle b$} (-.4,-.2);
\draw[thick, red] (1.8,.5) --node[below]{$\scriptstyle b$} (.4,.2);
\draw (-1.8,-1.4) --node[left]{$\scriptstyle w_1^{\rm mp}$} (-1.8,-1.1) -- (-1.8,1.1) --node[left]{$\scriptstyle w_2^{\rm mp}$} (-1.8,1.4);
\draw[knot] (-1,-1.4) --node[left]{$\scriptstyle x_1^{\rm mp}$} (-1,-1.1) -- (-1,1.1) --node[left]{$\scriptstyle x_2^{\rm mp}$} (-1,1.4);
\draw[knot] (1,-1.4) --node[right]{$\scriptstyle y_1$} (1,-1.1) -- (1,1.17) --node[right]{$\scriptstyle y_2$} (1,1.4);
\draw (1.8,-1.4) --node[right]{$\scriptstyle z_1$} (1.8,-1.1) -- (1.8,1.1) --node[right]{$\scriptstyle z_2$} (1.8,1.4);
\draw[thick, red, knot] (-1,-.8) --node[right]{$\scriptstyle a$} (-.4,-.2);
\draw[thick, red, knot] (1,.8) --node[left]{$\scriptstyle a$} (.4,.2);
\draw[dashed, cyan, thick] (0,-1.4) -- (0,1.4);
\roundNbox{fill=white}{(-1,-.8)}{.3}{0}{0}{$f_1$}
\roundNbox{fill=white}{(1,.8)}{.3}{0}{0}{$g_1$}
\roundNbox{fill=white}{(-1.8,-.6)}{.3}{0}{0}{$f_2$}
\roundNbox{fill=white}{(1.8,.6)}{.3}{0}{0}{$g_2$}
\filldraw[fill=blue] (-.4,-.2) circle (.05cm);
\filldraw[fill=blue] (.4,.2) circle (.05cm);
}\,.
\end{align*}
Again we use the pair of shaded nodes convention to sum over an ONB and its adjoint, but this notation suppresses a braiding to swap the $a$ and $b$ strands on the right hand side after taking the adjoint.

Now $\cX$ has a right $\cX^{\rm mp}\boxtimes_\cA\cX$-module structure given by $z \lhd (x^{\rm mp}\boxtimes y) := x\otimes z \otimes y$ and on morphisms by
$$
\tikzmath{
\draw (0,-1) --node[left]{$\scriptstyle z_1$} (0,-.7) -- (0,.7) --node[left]{$\scriptstyle z_2$} (0,1);
\roundNbox{fill=white}{(0,0)}{.3}{0}{0}{$h$}
}
\lhd 
\tikzmath{
\draw[dashed, cyan, thick] (0,-1) -- (0,1);
\draw (-.5,-1) --node[left]{$\scriptstyle x_1^{\rm mp}$} (-.5,-.7) -- (-.5,.7) --node[left]{$\scriptstyle x_2^{\rm mp}$} (-.5,1);
\draw (.5,-1) --node[right]{$\scriptstyle y_1$} (.5,-.7) -- (.5,.7) --node[right]{$\scriptstyle y_2$} (.5,1);
\draw[thick, red] (-.5,-.5) --node[left]{$\scriptstyle a$} (.5,.5);
\roundNbox{fill=white}{(-.5,-.4)}{.3}{0}{0}{$f$}
\roundNbox{fill=white}{(.5,.4)}{.3}{0}{0}{$g$}
}
:=
\tikzmath{
\draw (-.8,-1.4) --node[left]{$\scriptstyle x_1$} (-.8,-1.1) -- (-.8,1.1) --node[left]{$\scriptstyle x_2$} (-.8,1.4);
\draw (.8,-1.4) --node[right]{$\scriptstyle y_1$} (.8,-1.1) -- (.8,1.1) --node[right]{$\scriptstyle y_2$} (.8,1.4);
\draw[thick, red] (-.8,0) --node[left]{$\scriptstyle a$} (.8,.8);
\draw[knot] (0,-1.4) --node[right]{$\scriptstyle z_1$} (0,-1.1) -- (0,1.1) --node[right]{$\scriptstyle z_2$} (0,1.4);
\roundNbox{fill=white}{(-.8,0)}{.3}{0}{0}{$f$}
\roundNbox{fill=white}{(.8,.8)}{.3}{0}{0}{$g$}
\roundNbox{fill=white}{(0,-.8)}{.3}{0}{0}{$h$}
}
=
\tikzmath{
\draw (-.8,-1.4) --node[left]{$\scriptstyle x_1$} (-.8,-1.1) -- (-.8,1.1) --node[left]{$\scriptstyle x_2$} (-.8,1.4);
\draw (.8,-1.4) --node[right]{$\scriptstyle y_1$} (.8,-1.1) -- (.8,1.1) --node[right]{$\scriptstyle y_2$} (.8,1.4);
\draw[thick, red] (-.8,-.8) --node[left]{$\scriptstyle a$} (.8,0);
\draw[knot] (0,-1.4) --node[right]{$\scriptstyle z_1$} (0,-1.1) -- (0,1.1) --node[right]{$\scriptstyle z_2$} (0,1.4);
\roundNbox{fill=white}{(-.8,-.8)}{.3}{0}{0}{$f$}
\roundNbox{fill=white}{(.8,0)}{.3}{0}{0}{$g$}
\roundNbox{fill=white}{(0,.8)}{.3}{0}{0}{$h$}
}
\,.
$$

Above, we discussed the ladder category model of $\cX^{\rm mp}\boxtimes_\cA \cX$ acting on $\cM=\cX$; for our model, $\cC=\cX^{\rm mp}\boxtimes_\cA \cA \boxtimes_\cA \cX$ 
which means our ladders really have 3 vertical strands and two horizontal rungs, but this is not a major detail since $\cC\cong \cX^{\rm mp}\boxtimes_\cA \cX$.

We summarize the above construction in the following proposition.

\begin{prop}
\label{prop:FoldingAlgebraIsomorphism}
The folding trick gives a (non-bounded spread!) isomorphism between the quasi-local algebra of the braided enriched fusion categorical net $\fX$ for $\Phi^Z: \cA\to Z(\cX)$ on a $\bbZ\times \bbN$ lattice and the quasi-local algebra of the (multi)fusion module categorical net $\fM$ for $\cM_\cC$ where $\cM=\cX$ and $\cC= \cX^{\rm mp}\boxtimes_\cA \cA \boxtimes_\cA \cX$.
Moreover, this isomorphism restricts to an isomorphism on the subalgebras of local operators.
\end{prop}

Now since the folding isomorphism maps local operators to local operators, we can use them to line up the boundary DHR bimodules on either side.
First, observe that even though $\cC$ is now multifusion, Construction \ref{construct:BoundaryDHR} applies; it is for this reason that Appendix~\ref{appendix:AFActionsOfUmFCs} is written in this generality (see Remark~\ref{rem:FullyFaithful}).
Hence for each $F\in \End(\cM_\cC)$, we get a boundary DHR bimodule ${}_\fM Y^F_\fM \in \DHR^\partial(\fM)$.
By Theorem \ref{thm:BoundaryLTOforLW}, we conclude that
$\DHR^\partial(\fM)\cong \End(\cM_\cC)$ as unitary tensor categories.
Now as 
$$
\cC
=
\cX^{\rm mp}\boxtimes_\cA \cA \boxtimes_\cA \cX
\cong
\cX^{\rm mp}\boxtimes_\cA \cX
$$
as multifusion categories, by 
\cite[Ex.~II.9]{MR4640433}
together with the folding trick in the 4-category of unitary braided fusion categories \cite{MR4228258},\footnote{Braided fusion categories are fully dualizable in the 4-category of braided tensor categories by \cite{MR4228258}, which allows us to apply the folding trick.}
we have 
$$
\End(\cM_\cC) 
\cong
\End(\cX_{\cX^{\rm mp}\boxtimes_\cA \cX})
\cong
\End^\cA({}_\cX\cX_{\cX})
\cong 
Z^\cA(\cX)
=
\cA'\subset Z(\cX),
$$
the \emph{enriched center} of $\cX$ as an $\cA$-enriched fusion category \cite{MR3725882}, a.k.a., the \emph{M\"uger centralizer} of $\cA\subset Z(\cX)$ \cite{MR1990929}.
Here, $\End^\cA({}_\cX\cX_\cX)$ denotes the the $\cA$-enriched endomorphisms in the 3-category of $\cA$-enriched unitary fusion categories.

\begin{rem}
From a subfactor perspective, it would be reasonable to call the unitary internal endomorphism \cite{MR4750417} algebra
$$
S:=
\underline{\End}^\dag_{\cX^{\rm mp}\boxtimes_\cA\cX}(1_\cX)
$$
the
\emph{enriched symmetric enveloping algebra object}.
Similar to \cite[Lem.~6.3]{MR4581741}, one can think of $S$ as the unitary splitting of the orthogonal projection
$$
\frac{1}{D_\cA}
\sum_{\substack{
x,y\in\Irr(\cX)
\\
a\in\Irr(\cA)
}}
\sqrt{d_a}\cdot
\tikzmath{
\draw[dashed, cyan, thick] (0,-.6) -- (0,.6);
\draw (-.3,-.6) --node[left]{$\scriptstyle x^{\rm mp}$} (-.3,-.3) -- (-.3,.3) --node[left]{$\scriptstyle y^{\rm mp}$} (-.3,.6);
\draw (.3,-.6) --node[right]{$\scriptstyle x$} (.3,-.3) -- (.3,.3) --node[right]{$\scriptstyle y$} (.3,.6);
\draw[thick, red] (-.3,-.3) --node[left, xshift=.05cm]{$\scriptstyle a$} (.3,.3);
\filldraw[fill=blue] (-.3,-.3) circle (.05cm);
\filldraw[fill=blue] (.3,.3) circle (.05cm);
}
\in 
\End_{\cX^{\rm mp}\boxtimes_\cA \cX}\left( \bigoplus_{x\in\Irr(\cX) }x^{\rm mp}\boxtimes x\right),
$$
where we again use a pair of shaded vertices to denote summing over an ONB and its conjugate for an appropriate inner product.
That the formula defines an orthogonal projection is entirely similar to \cite[Eq.~(30)]{MR4581741}.

Now given an $\cA$-enriched UFC $\cX$ and an action of $\cX$ as bimodules over the hyperfinite $\rm II_1$ factor $R$, the realization construction \cite{MR3948170,MR4419534} yields a \emph{sub-multifactor} (an inclusion of multifactors which have finite dimensional centers)
$R^{\rm op}\otimes R \subset |S|$
which could reasonably called the \emph{enriched symmetric enveloping algebra}.
\end{rem}

We now turn our attention to $\DHR^\partial(\fX)$.
For the time being, we focus only on its unitary tensor structure, as we have made no mention of any \emph{braiding} on either $\DHR^\partial(\fX)$ nor $Z^\cA(\cX)$.
By Proposition \ref{prop:FoldingAlgebraIsomorphism}, we have a folding isomorphism of quasilocal algebars $\fX\cong \fM$, so we have an equivalence of categories $\Bim(\fX)\cong \Bim(\fM)$.
Now being a DHR-bimodule is a \emph{property} of a Hilbert $\rmC^*$-bimodule, namely that it admits localized finite projective bases.

\begin{prop}
\label{prop:TensorEquivalence}
The folding isomorphism gives an equivalence between the unitary tensor category of boundary DHR bimodules $\DHR^\partial(\fX)$ and the unitary tensor category of boundary DHR bimodules $\DHR^\partial(\fM)$.
In particular,  $\DHR^\partial(\fX)$ is equivalent to the enriched center $Z^\cA(\cX)$ as a unitary fusion category.
\end{prop}
\begin{proof}
We claim that the equivalence $E_f:\Bim(\fX)\to \Bim(\fM)$ arising from the folding isomorphism $f:\fX\to \fM$ manifestly preserves the DHR property on either side.
To see this, 
recall that $Y\in \DHR^\partial(\fX)$ has a projective basis localized in any sufficiently large boundary rectangle.
Choosing such a boundary rectangle about the folding point, we get a projective basis for $E_f(Y)$ which is localized in $I$ for any sufficiently long initial interval $I\subset \bbN$ for $\fM$.
Hence $E_f$ restricts to a fully faithful tensor functor $\DHR^\partial(\fX)\to \DHR^\partial(\fM)$.

Now since we know that $\DHR^\partial(\fM)\cong Z^\cA(\cX)$ as unitary fusion categories, it suffices to construct a fully faithful unitary tensor functor $Y:Z^\cA(\cX)\to \DHR^\partial(\fX)$ such that the composite $E_f\circ G$ is naturally isomorphic to the unitary tensor functor 
$Z^\cA(\cX)\cong \End_\cC(\cM)\to \DHR^\partial(\fM)$
from Construction \ref{construct:BoundaryDHR}.
Indeed, this proves that $E_f$ is essentially surjective, which will conclude the proof.
This is exactly the content of Construction \ref{const:BoundaryDHRFOrEnrichedNet} below.
\end{proof}

\begin{construction}
\label{const:BoundaryDHRFOrEnrichedNet}
We construct a boundary DHR bimodule $Y^z$ for $z\in Z^\cA(\cX)$ parallel to the 1D Construction \ref{construct:BoundaryDHR} which exhausts all boundary DHR bimodules.
Pick an initial boundary point $\lambda_0\in \partial \cL$, and inductively construct a sequence of rectangles by
$\Lambda_0:=\{\lambda_0\}$ and $\Lambda_{n+1}=\Lambda_n^{+1}$, i.e., $\Lambda_{n}^{+1}$ is obtained from $\Lambda_n$ by adding all points in $\cL$ with $\ell^\infty$-distance 1.
For $n\in \bbN$, we define
$$
Y^z_n := 
\Hom_\cX( 
X^{\otimes \Lambda_n\cap \partial \cL}
\lhd 
A^{\otimes \Lambda_n\setminus \partial \cL}
\to 
X^{\otimes \Lambda_n\cap \partial \cL}
\lhd 
A^{\otimes \Lambda_n\setminus \partial \cL}
\otimes z
).
$$
The inclusion $Y_n^z\hookrightarrow Y^z_{n+1}$ is achieved by tensoring with identity morphisms, $\id_X$ for sites in $\partial \cL$ and $\id_A$ for sites in $\cL\setminus \partial\cL$.
One uses the braiding in $Z(\cX)$ to move the $z$-string past the $\id_A$-strings, and since $z\in Z^\cA(\cX)$, the inclusion is independent of the choice of how one performs this braiding.
For example, the inclusion $Y^z_3\hookrightarrow Y^z_4$ is given by
$$
Y^z_3
\ni
\tikzmath{
\foreach \x in {0,1,2}{
\draw[thick] (\x,0) -- ($ (\x,0) + (-.4,-.4) $);
\foreach \y in {1,2}{
\draw[red, thick] (\x,\y) -- ($ (\x,\y) + (-.4,-.4) $);
}}
\filldraw[thick, rounded corners=5pt, fill=gray!50, opacity=.8] (-.2,-.2) rectangle (2.2,2.2);
\foreach \x in {0,1,2}{
\draw[thick] (\x,0) -- ($ (\x,0) + (.4,.4) $);
\foreach \y in {1,2}{
\draw[red, thick] (\x,\y) -- ($ (\x,\y) + (.4,.4) $);
}}
\node at (1,.5) {$f$};
\draw[cyan, knot, thick, snake] (2.2,0) -- (3.2,0) node[above]{$\scriptstyle z$};
}
\qquad \longmapsto \qquad
\tikzmath{
\foreach \x in {-1,0,1,2,3}{
\draw[thick] (\x,0) -- ($ (\x,0) + (-.4,-.4) $);
\foreach \y in {1,2,3}{
\draw[red, thick] (\x,\y) -- ($ (\x,\y) + (-.4,-.4) $);
}}
\filldraw[thick, rounded corners=5pt, fill=gray!50, opacity=.8] (-.2,-.2) rectangle (2.2,2.2);
\foreach \x in {-1,0,1,2,3}{
\draw[thick] (\x,0) -- ($ (\x,0) + (.4,.4) $);
\foreach \y in {1,2,3}{
\draw[red, thick] (\x,\y) -- ($ (\x,\y) + (.4,.4) $);
}}
\node at (1,.5) {$f$};
\draw[cyan, knot, thick, snake] (2.2,0) -- (3.5,0) node[above]{$\scriptstyle z$};
}
\in Y^z_{4}.
$$
Here, we use the convention from \cite[Rem.~2.12]{2305.14068}; even though strands corresponding to $A\in\cA$ are on one side of $\cX$, the $z$-strand in the centralizer $Z^\cA(\cX)=\cA'\subset Z(\cX)$ is drawn on the other side of $\cX$ to emphasize that $\cX$ is an $\cA-Z^\cA(\cX)$ bimodule tensor category.
The $Y^z_n$ are equipped with the obvious $\fX(\Lambda_n)$-valued inner products which are compatible with the inclusion maps.
We thus obtain a well-defined inductive limit bimodule $Y^z:=\varinjlim Y^z_n$ which clearly lies in $\DHR^\partial(\fX)$.
The reasoning is completely analogous to that for the fusion module spin chain in Construction \ref{construct:BoundaryDHR}:
since $X$ contains a copy of every simple in $\cX$, we can find a projective basis for $Y^z$ which is localized at any single site in $\partial \cL$.
The bimodules $Y^z$ assemble into a functor $Y: Z^\cA(\cX)\to \DHR^\partial(\fX)$ which is fully faithful by the argument from Appendix \ref{appendix:AFActionsOfUmFCs}. 
Indeed, this is shown at the level of the 1D fusion module categorical nets for the 1D slab.
By construction, we see that the following triangle commutes (up to natural isomorphism), concluding the proof of Proposition \ref{prop:TensorEquivalence}.
\[
\begin{tikzcd}
Z^\cA(\cX)\cong \End(\cM_\cC)
\arrow[rr,hook,"\text{Const.~\ref{construct:BoundaryDHR}}"]
\arrow[dr,hook,"Y"]
&&
\DHR^\partial(\fM)
\\
&
\DHR^\partial(\fX)
\arrow[ur,"E_f"]
\end{tikzcd}
\]
\end{construction}

Since the 1D fusion module categorical net satisfies boundary algebraic Haag duality by \eqref{eq:FusionModuleBoundaryHaagDuality} (see also Appendix \ref{appendix:AFActionsOfUmFCs}), the 2D enriched fusion category categorical net satisfies boundary algebraic Haag duality.
This allows us to follow the prescription from \cite{MR4814692} to define a unitary braiding on the boundary DHR bimodules $\DHR^\partial(\fX)$.

\begin{defn}
Suppose $Y,Z\in \DHR^\partial(\fX)$.
Choose projective bases $\{b_{i}\}, \{c_{j}\}$ for $Y,Z$ respectively which are localized in sufficiently large rectangles $\Lambda,\Delta$ which are sufficiently far apart in $\cL$, both of which meet $\partial \cL$.
We define 
$$
u^{\Lambda,\Delta}_{Y,Z}: Y\boxtimes_\fX Z\rightarrow Z\boxtimes_\fX Y
\qquad\qquad\text{by}\qquad\qquad
\sum_{i,j}b_{i}\boxtimes c_{j}a_{ij}
\mapsto
\sum_{i,j}c_{j}\boxtimes b_{i} a_{ij}.
$$
As in \cite[\S3.3]{MR4814692},
one verifies this formula is a well-defined unitary $\fX-\fX$ bimodule map
which is independent of the choices of projective bases, as well as the choices of $\Lambda,\Delta$, provided they are sufficiently large and separated and meet our 1D boundary $\partial \cL$ in the same linear ordering (or we might accidentally reverse the braiding).
We provide some initial details to convince the reader that this works as claimed.
In all the calculations below, the elements with underbraces live in the indicated algebras by Lemma~\ref{lem:BoundaryIPLocalizedInCentralizer}.
First, for all $a_{ij}\in \fX$, 
\begin{align*}
\left\langle 
u_{Y,Z}^{\Lambda,\Delta}\sum b_i \boxtimes c_j a_{ij}
\bigg|
u_{Y,Z}^{\Lambda,\Delta}\sum b_k \boxtimes c_\ell a_{k\ell}
\right\rangle
&=
\sum
a_{ij}^*
\langle 
c_j \boxtimes b_i 
|
c_\ell \boxtimes b_k 
\rangle_{\fX}
a_{k\ell}
\\&=
\sum
a_{ij}^*
\langle 
b_i 
|
\underbrace{\langle c_j | c_\ell \rangle_{\fX}}_{\in\fX(\Delta)} b_k 
\rangle_{\fX}
a_{k\ell}
\\&=
\sum
a_{ij}^*
\langle 
b_i 
|
b_k \langle c_j | c_\ell \rangle_{\fX} 
\rangle
a_{k\ell}
\\&=
\sum
a_{ij}^*
\underbrace{\langle b_i|b_k  
\rangle_\fX}_{\in \fX(\Lambda)}
\langle c_j | c_\ell \rangle_{\fX}
a_{k\ell}
\\&=
\sum
a_{ij}^*
\langle c_j | c_\ell \rangle_{\fX}
\langle b_i |b_k  \rangle_\fX
a_{k\ell}
\\&=
\sum
a_{ij}^*
\langle c_j | c_\ell 
\langle b_i |b_k  \rangle_\fX
\rangle_{\fX}
a_{k\ell}
\\&=
\sum
a_{ij}^*
\langle c_j |  
\langle b_i |b_k  \rangle_\fX
c_\ell
\rangle_{\fX}
a_{k\ell}
\\&=
\sum
a_{ij}^*
\langle b_i\boxtimes c_j |  
b_k \boxtimes c_\ell
\rangle_\fX
a_{k\ell}
\\&=
\left\langle 
\sum b_i \boxtimes c_j a_{ij}
\bigg|
\sum b_k \boxtimes c_\ell a_{k\ell}
\right\rangle.
\end{align*}
If $\{b_k'\}$, $\{c_\ell'\}$ are another choice of bases still localized in $\Lambda,\Delta$ respectively, we have
\begin{align*}
u^{\Lambda,\Delta}_{Y,Z} (b_k'\boxtimes c_\ell')
&=
u^{\Lambda,\Delta}_{Y,Z}\sum b_i\underbrace{\langle b_i|b_k'\rangle_\fX}_{\in\fX(\Lambda)}\boxtimes c_j
\underbrace{\langle c_j|c_\ell'\rangle_\fX}_{\in\fX(\Delta)}
\\&=
u^{\Lambda,\Delta}_{Y,Z}\sum b_i\boxtimes \langle b_i|b_k'\rangle_\fX c_j\langle c_j|c_\ell'\rangle_\fX
\\&=
u^{\Lambda,\Delta}_{Y,Z}\sum b_i\boxtimes c_j \langle b_i|b_k'\rangle_\fX \langle c_j|c_\ell'\rangle_\fX
\\&=
\sum c_j\boxtimes b_i \langle b_i|b_k'\rangle_\fX \langle c_j|c_\ell'\rangle_\fX
\\&=
\sum c_j\boxtimes b_k' \langle c_j|c_\ell'\rangle_\fX
\\&=
\sum c_j\boxtimes  \langle c_j|c_\ell'\rangle_\fX b_k'
\\&=
\sum c_j  \langle c_j|c_\ell'\rangle_\fX \boxtimes b_k'
\\&=
c_\ell' \boxtimes b_k'.
\end{align*}

By a similar argument to \cite[Lem.~3.10]{MR4814692}, $u^{\Lambda,\Delta}_{X,Y}$ is independent of the choice of $\Lambda, \Delta$, provided they meet $\partial \cL$ in the same linear ordering.
The argument for pairs $\Lambda_1,\Delta_1$ and $\Lambda_2,\Delta_2$ such that there are disjoint $\Lambda,\Delta$ which are sufficiently separated containing $\Lambda_1,\Lambda_2$ and $\Delta_1,\Delta_2$ respectively is an immediate consequence of the above calculation.
In the case when the linear ordering of the $\Lambda_i,\Delta_j$ prevents this nice scenario, say $\Lambda_1 < \Delta_1 \leq \Lambda_2 < \Delta_2$ (in terms of their `starting positions' meeting $\partial\cL$), we simply apply the above argument twice showing
$
u_{X,Y}^{\Lambda_1,\Delta_1}
=
u_{X,Y}^{\Lambda_1,\Delta_2}
=
u_{X,Y}^{\Lambda_2,\Delta_2}
$.
We may thus simply write $u_{X,Y}$.

By definition, $u_{X,Y}$ is manifestly a right $\fX$-module map.
To see left $\fX$-linearity, it suffices to check the left action of an arbitrary $a\in\fX(\Lambda)$ for an arbitrary $\Lambda$.
Picking $\{b_i\}$ localized in $\Lambda$ and $\{c_j\}$ localized in $\Delta$ (where the choice of the $b_i$ comes after the choice of $\Lambda$ which comes after the choice of $a\in \fX(\Lambda)$ which is acting),
we calculate
\begin{align*}
u^{\Lambda,\Delta}_{Y,Z}(ab_k\boxtimes c_j)
&=
u^{\Lambda,\Delta}_{Y,Z}\sum b_i\underbrace{\langle b_i| ab_k\rangle_\fX}_{\in\fX(\Lambda)}\boxtimes c_j
\\&=
u^{\Lambda,\Delta}_{Y,Z}\sum b_i\boxtimes \langle b_i| ab_k\rangle_\fX c_j
\\&=
u^{\Lambda,\Delta}_{Y,Z}\sum b_i\boxtimes c_j\langle b_i| ab_k\rangle_\fX
\\&=
\sum c_j\boxtimes b_i\langle b_i| ab_k\rangle_\fX
\\&=
c_j\boxtimes ab_k
\\&=
c_j a\boxtimes b_k
\\&=
ac_j \boxtimes b_k
\\&=
au^{\Lambda,\Delta}_{Y,Z}(b_k\boxtimes c_j).
\end{align*}
Hence the map $u^{\Lambda,\Delta}_{Y,Z}$ is a well-defined $\fX-\fX$ bimodule intertwiner.
One can similarly mimic the remainder of the proofs from \cite{MR4814692} to see we have a well-defined unitary braiding.
\end{defn}

\begin{thm}[]
\label{thm:BoundaryDHRForBraidedEnrichedNet}
The category of boundary DHR bimodules $\DHR^\partial(\fX)$ is equivalent as a unitary braided tensor category to the enriched center $Z^\cA(\cX)$.
\end{thm}
\begin{proof}
It remains to show that the functor $Y: Z^\cA(\cX)\to \DHR^\partial(\fX)$ is braided.
This follows similarly to the proof of \cite[Lem.~4.18]{MR4814692} based on \cite[Prop.~6.15]{MR4753059}, adding the $\cA$-enrichment.
For an easy way to see this, we provide a 2D graphical version of the the proof from \cite[Lem.~4.18]{MR4814692} for the case $\cA=\Hilb$, from which our result follows by adding $A$-strands going `back into the page,' making the diagrams 3D.

Similar to the argument at the end of Construction \ref{construct:BoundaryDHR}, since $X$ contains a copy of every $x\in\Irr(\cX)$, we can find a projective basis for $Y^z$ localized at any site $n\in \bbZ$ by choosing a partition of unity for $X\otimes Fz$ as follows.
\begin{equation}
\label{eq:PartitionOfUnityGivesProjectiveBasis}
1_{X\otimes Fz} = 
\sum_{j}
\tikzmath{
\draw (0,.3) --node[left]{$\scriptstyle X$} (0,.7);
\draw[cyan, knot, thick, snake] (.2,1.3) -- (.2,1.7) node[above]{$\scriptstyle Fz$};
\draw (-.2,1.3) -- (-.2,1.7) node[above]{$\scriptstyle X$};
\draw (-.2,-.3) -- (-.2,-.7) node[below]{$\scriptstyle X$};
\draw[cyan, knot, thick, snake] (.2,-.3) -- (.2,-.7) node[below]{$\scriptstyle Fz$};
\roundNbox{}{(0,1)}{.3}{.1}{.1}{$b_j$}
\roundNbox{}{(0,0)}{.3}{.1}{.1}{$b_j^\dag$}
}
\qquad\qquad
\rightsquigarrow
\qquad\qquad
\tikzmath{
\draw (-1,-.7) node[below]{$\scriptstyle X$} -- (-1,.7) node[above]{$\scriptstyle X$};
\draw (-.5,-.7) node[below]{$\scriptstyle X$} -- (-.5,.7) node[above]{$\scriptstyle X$};
\node at (.8,.5) {$\cdots$};
\draw (1,-.7) node[below]{$\scriptstyle X$} -- (1,.7) node[above]{$\scriptstyle X$};
\draw (.5,-.7) node[below]{$\scriptstyle X$} -- (.5,.7) node[above]{$\scriptstyle X$};
\node at (-.7,0) {$\cdots$};
\draw[thick, cyan, knot, snake] (.3,0) -- (1.3,0) node[right]{$\scriptstyle Fz$} ;
\draw (0,.3) -- (0,.7) node[above]{$\scriptstyle X$};
\draw (0,-.3) -- (0,-.7) node[below]{$\scriptstyle X$};
\roundNbox{fill=white}{(0,0)}{.3}{0}{0}{$b_j$}
}
\end{equation}
Now choose another $w\in Z(\cX)$ with partition of unity $1_{X\otimes w}=\sum a_i a_i^\dag$ similar to \eqref{eq:PartitionOfUnityGivesProjectiveBasis}, 
and we obtain a projective basis  $\{a_i\}$ for $Y^w$ localized at any other site we choose, say any $k<n$.
Observe that $\{a_i\boxtimes b_j\}$ is a projective basis for $Y^w\boxtimes_{\fX} Y^z$, but it is now localized in the interval $[k,n]$.
We can now compute that our functor $Y: Z(\cX)\to \DHR^\partial(\fX)$ is braided.
Indeed, postcomposing with the braiding $w\otimes z\to z\otimes w$, we have
$$
Y(\beta_{w,z})(a_i\boxtimes b_j)
=
\tikzmath{
\draw (-1.5,-.7) node[below]{$\scriptstyle X$} -- (-1.5,1.7) node[above]{$\scriptstyle X$};
\draw (-1,-.7) node[below]{$\scriptstyle X$} -- (-1,1.7) node[above]{$\scriptstyle X$};
\draw (1,-.7) node[below]{$\scriptstyle X$} -- (1,1.7) node[above]{$\scriptstyle X$};
\draw (.5,-.7) node[below]{$\scriptstyle X$} -- (.5,1.7) node[above]{$\scriptstyle X$};
\node at (.8,.5) {$\cdots$};
\node at (-1.2,.5) {$\cdots$};
\draw (0,.3) -- (0,1.7) node[above]{$\scriptstyle X$};
\draw (0,-.3) -- (0,-.7) node[below]{$\scriptstyle X$};
\draw (-.5,1.3) -- (-.5,1.7) node[above]{$\scriptstyle X$};
\draw (-.5,.7) -- (-.5,-.7) node[below]{$\scriptstyle X$};
\draw[thick, cyan, knot, snake] (.3,0) -- (1,0) to[out=0,in=180] (1.7,1) -- (2,1) node[right]{$\scriptstyle Fz$} ;
\draw[thick, orange, knot, snake] (-.2,1) -- (1,1) to[out=0,in=180] (1.7,0) -- (2,0) node[right]{$\scriptstyle Fw$} ;
\roundNbox{fill=white}{(0,0)}{.3}{0}{0}{$b_j$}
\roundNbox{fill=white}{(-.5,1)}{.3}{0}{0}{$a_i$}
}
=
\tikzmath{
\draw (-1.5,-.7) node[below]{$\scriptstyle X$} -- (-1.5,1.7) node[above]{$\scriptstyle X$};
\draw (-1,-.7) node[below]{$\scriptstyle X$} -- (-1,1.7) node[above]{$\scriptstyle X$};
\draw (1,-.7) node[below]{$\scriptstyle X$} -- (1,1.7) node[above]{$\scriptstyle X$};
\draw (.5,-.7) node[below]{$\scriptstyle X$} -- (.5,1.7) node[above]{$\scriptstyle X$};
\node at (.8,.5) {$\cdots$};
\node at (-1.2,.5) {$\cdots$};
\draw (0,1.3) -- (0,1.7) node[above]{$\scriptstyle X$};
\draw (0,.7) -- (0,-.7) node[below]{$\scriptstyle X$};
\draw (-.5,.3) -- (-.5,1.7) node[above]{$\scriptstyle X$};
\draw (-.5,-.3) -- (-.5,-.7) node[below]{$\scriptstyle X$};
\draw[thick, cyan, knot, snake] (.3,1) -- (1.3,1) node[right]{$\scriptstyle Fz$} ;
\draw[thick, orange, knot, snake] (-.2,0) -- (1.3,0) node[right]{$\scriptstyle Fw$} ;
\roundNbox{fill=white}{(0,1)}{.3}{0}{0}{$b_j$}
\roundNbox{fill=white}{(-.5,0)}{.3}{0}{0}{$a_i$}
}
\hspace*{-.2cm}=
b_j\boxtimes a_i
=
u^{k,n}_{Y^w,Y^z}(a_i\boxtimes b_j).
$$
Since both $Y(\beta_{w,z})$ and $u^{k,n}_{Y^w,Y^z}$ are $\fX-\fX$ bimodular, the result follows.

When $\cA$ is an arbitrary braided UFC and $\cX$ is an arbitrary $\cA$-enriched UFC, we still have projective bases for $Y^w,Y^z$ localized at a single boundary site as in \eqref{eq:PartitionOfUnityGivesProjectiveBasis}.
Again using the convention from \cite[Rem.~2.12]{2305.14068} in which the strands for $w,z\in Z^\cA(\cX)$ are drawn on the other side of $\cX$, we see that the red $A$-strands going `back into the page' play no role in the above argument.
$$
\tikzmath{
\draw (-1.5,-.7) node[below]{$\scriptstyle X$} -- (-1.5,1.7) node[above]{$\scriptstyle X$};
\draw[thick, red] (-1.6,-.8) -- (-1.6,1.6);
\draw[thick, red] (-1.7,-.9) -- (-1.7,1.5);
\draw (-1,-.7) node[below]{$\scriptstyle X$} -- (-1,1.7) node[above]{$\scriptstyle X$};
\draw[thick, red] (-1.1,-.8) -- (-1.1,1.6);
\draw[thick, red] (-1.2,-.9) -- (-1.2,1.5);
\draw (1,-.7) node[below]{$\scriptstyle X$} -- (1,1.7) node[above]{$\scriptstyle X$};
\draw[thick, red] (.9,-.8) -- (.9,1.6);
\draw[thick, red] (.8,-.9) -- (.8,1.5);
\draw (.5,-.7) node[below]{$\scriptstyle X$} -- (.5,1.7) node[above]{$\scriptstyle X$};
\draw[thick, red] (.4,-.8) -- (.4,1.6);
\draw[thick, red] (.3,-.9) -- (.3,1.5);
\filldraw[white] (-1.4,.5) circle (.05cm);
\filldraw (-1.4,.5) circle (.025cm);
\filldraw[white] (-1.25,.5) circle (.055cm);
\filldraw (-1.25,.5) circle (.025cm);
\filldraw[white] (-1.1,.5) circle (.05cm);
\filldraw (-1.1,.5) circle (.025cm);
\filldraw[white] (.9,.5) circle (.055cm);
\filldraw (.9,.5) circle (.025cm);
\filldraw[white] (.75,.5) circle (.055cm);
\filldraw (.75,.5) circle (.025cm);
\filldraw[white] (.6,.5) circle (.055cm);
\filldraw (.6,.5) circle (.025cm);
\draw[thick, red] (-.1,-.8) -- (-.1,1.6);
\draw[thick, red] (-.2,-.9) -- (-.2,1.5);
\draw[thick, red] (-.6,-.8) -- (-.6,1.6);
\draw[thick, red] (-.7,-.9) -- (-.7,1.5);
\roundNbox{knot, fill=white}{(0,0)}{.3}{0}{0}{$b_j$}
\roundNbox{knot, fill=white}{(-.5,1)}{.3}{0}{0}{$a_i$}
\draw (0,.3) -- (0,1.7) node[above]{$\scriptstyle X$};
\draw (0,-.3) -- (0,-.7) node[below]{$\scriptstyle X$};
\draw (-.5,1.3) -- (-.5,1.7) node[above]{$\scriptstyle X$};
\draw (-.5,.7) -- (-.5,-.7) node[below]{$\scriptstyle X$};
\draw[thick, cyan, knot, snake] (.3,0) -- (1,0) to[out=0,in=180] (1.7,1) -- (2,1) node[right]{$\scriptstyle Fz$} ;
\draw[thick, orange, knot, snake] (-.2,1) -- (1,1) to[out=0,in=180] (1.7,0) -- (2,0) node[right]{$\scriptstyle Fw$} ;
}
=
\tikzmath{
\draw (-1.5,-.7) node[below]{$\scriptstyle X$} -- (-1.5,1.7) node[above]{$\scriptstyle X$};
\draw[thick, red] (-1.6,-.8) -- (-1.6,1.6);
\draw[thick, red] (-1.7,-.9) -- (-1.7,1.5);
\draw (-1,-.7) node[below]{$\scriptstyle X$} -- (-1,1.7) node[above]{$\scriptstyle X$};
\draw[thick, red] (-1.1,-.8) -- (-1.1,1.6);
\draw[thick, red] (-1.2,-.9) -- (-1.2,1.5);
\draw (1,-.7) node[below]{$\scriptstyle X$} -- (1,1.7) node[above]{$\scriptstyle X$};
\draw[thick, red] (.9,-.8) -- (.9,1.6);
\draw[thick, red] (.8,-.9) -- (.8,1.5);
\draw (.5,-.7) node[below]{$\scriptstyle X$} -- (.5,1.7) node[above]{$\scriptstyle X$};
\draw[thick, red] (.4,-.8) -- (.4,1.6);
\draw[thick, red] (.3,-.9) -- (.3,1.5);
\filldraw[white] (-1.4,.5) circle (.05cm);
\filldraw (-1.4,.5) circle (.025cm);
\filldraw[white] (-1.25,.5) circle (.055cm);
\filldraw (-1.25,.5) circle (.025cm);
\filldraw[white] (-1.1,.5) circle (.05cm);
\filldraw (-1.1,.5) circle (.025cm);
\filldraw[white] (.9,.5) circle (.055cm);
\filldraw (.9,.5) circle (.025cm);
\filldraw[white] (.75,.5) circle (.055cm);
\filldraw (.75,.5) circle (.025cm);
\filldraw[white] (.6,.5) circle (.055cm);
\filldraw (.6,.5) circle (.025cm);
\draw[thick, red] (-.1,-.8) -- (-.1,1.6);
\draw[thick, red] (-.2,-.9) -- (-.2,1.5);
\draw[thick, red] (-.6,-.8) -- (-.6,1.6);
\draw[thick, red] (-.7,-.9) -- (-.7,1.5);
\roundNbox{knot,fill=white}{(0,1)}{.3}{0}{0}{$b_j$}
\roundNbox{knot,fill=white}{(-.5,0)}{.3}{0}{0}{$a_i$}
\draw (0,1.3) -- (0,1.7) node[above]{$\scriptstyle X$};
\draw (0,.7) -- (0,-.7) node[below]{$\scriptstyle X$};
\draw (-.5,.3) -- (-.5,1.7) node[above]{$\scriptstyle X$};
\draw (-.5,-.3) -- (-.5,-.7) node[below]{$\scriptstyle X$};
\draw[thick, cyan, knot, snake] (.3,1) -- (1.3,1) node[right]{$\scriptstyle Fz$} ;
\draw[thick, orange, knot, snake] (-.2,0) -- (1.3,0) node[right]{$\scriptstyle Fw$} ;
}
$$
The result follows.
\end{proof}

\section{Boundary algebras of Walker-Wang models}
\label{sec:WW}

In this section, we prove that the unitary tensor category version \cite{MR4640433,2305.14068} of the Walker-Wang string net model \cite{1104.2632} for a unitary braided fusion category (UBFC) $\cA$
(as opposed to the 6j-symbol model)
has a net of projections 
$$
p_\Lambda:= \prod_{p\subset \Lambda} B_p \prod_{\ell\subset\Lambda} A_\ell
$$
satisfying the LTO axioms from~\cite{MR4945955} (see \S\ref{sec:LTO} above for the main definitions).
We then identify the boundary algebra as a 2D braided categorical net, and we use the DHR bimodules for this net to calculate the point-like excitations.
Finally, we consider a topological boundary coming from an $\cA$-enriched UFC $\cX$ and study the boundary LTO and its boundary DHR bimodules.

\subsection{The Walker-Wang model}
Let $\cA$ denote a UBFC.
For simplicity we will only consider the model on a cubic lattice in three dimensions.
Schematically, the Hilbert space can be visualized as on the left hand side below, where the red edges carry labels from $\Irr(\cA)$.
Given an $\cA$-enriched UFC $\cX$ equipped with $\Phi^Z: \cA\to Z(\cX)$ and $\Phi=\operatorname{Forget}\circ \Phi^Z:\cA\to \cX$, we also get a 2D topological $\cX$-boundary for the 3D Walker-Wang model for $\cA$, pictured on the right hand side below.
$$
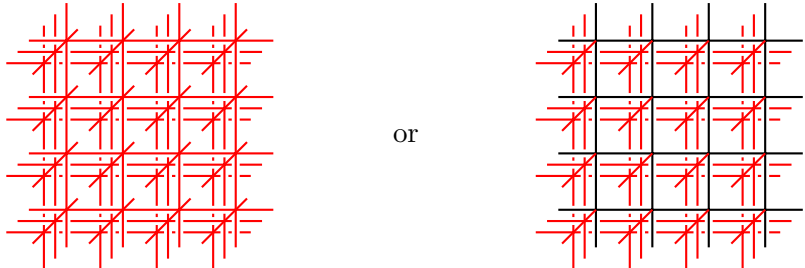
\tikzmath{
\draw[thick, red, step=.75, xshift=-.3cm, yshift=-.3cm] (0.25,0.25) grid (3.5,3.5);
\draw[thick, red, step=.75, knot, xshift=-.15cm, yshift=-.15cm] (0.25,0.25) grid (3.5,3.5);
\draw[thick, red, step=.75, knot] (0.25,0.25) grid (3.5,3.5);
\foreach \x in {.75,1.5,2.25,3}{
\foreach \y in {.75,1.5,2.25,3}{
\draw[thick, red] ($ (\x,\y) + (.15,.15) $) -- ($ (\x,\y) + (-.45,-.45) $);
}}
}
\qquad\qquad\text{or}\qquad\qquad
\tikzmath{
\draw[thick, red, step=.75, xshift=-.3cm, yshift=-.3cm] (0.25,0.25) grid (3.5,3.5);
\draw[thick, red, step=.75, knot, xshift=-.15cm, yshift=-.15cm] (0.25,0.25) grid (3.5,3.5);
\draw[thick, step=.75, knot, thick] (0.25,0.25) grid (3.5,3.5);
\foreach \x in {.75,1.5,2.25,3}{
\foreach \y in {.75,1.5,2.25,3}{
\draw[thick, red] (\x,\y) -- ($ (\x,\y) + (-.45,-.45) $);
}}
}
$$
Here, we read from bottom left to top right.
The total Hilbert space is the tensor product of local Hilbert spaces:
\begin{equation}
\label{eq:VertexSpaceOrientation}
\tikzmath{
\draw[thick, red] (-.5,-.5) node[left]{$\scriptstyle b_1$} -- (.5,0.5) node[right]{$\scriptstyle b_2$};
\draw[thick, red] (-.5,0) node[left]{$\scriptstyle a_1$} -- (.5,0) node[right]{$\scriptstyle a_2$};
\draw[thick, red] (0,-.5) node[below]{$\scriptstyle c_1$} -- (0,.5) node[above]{$\scriptstyle c_2$};
\filldraw[red] (0,0) circle (.05cm);
\draw[blue!50, very thin] (-.5,.5) -- (.5,-.5);
{\draw[blue!50, very thin, -stealth] ($ (-.3,.3) - (.1,.1)$) to ($ (-.3,.3) + (.1,.1)$);}
{\draw[blue!50, very thin, -stealth] ($ (.3,-.3) - (.1,.1)$) to ($ (.3,-.3) + (.1,.1)$);}
}
\qquad
\longleftrightarrow
\qquad
\cH_v 
:= 
\bigoplus_{
a_i,b_i,c_i\in \Irr(\cA)
}
\cA(a_1b_1c_1\to c_2b_2a_2)
\end{equation}
where the direct sum is orthogonal.
The space $\cH_v$ is equipped with the `skein-module' inner product
$$
\left\langle
\tikzmath{
\draw[thick, red] (-.5,-.5) node[left]{$\scriptstyle b_1$} -- (.5,.5) node[right]{$\scriptstyle b_2$};
\draw[thick, red] (-.5,0) node[left]{$\scriptstyle a_1$} -- (.5,0) node[right]{$\scriptstyle a_2$};
\draw[thick, red] (0,-.5) node[below]{$\scriptstyle c_1$} -- (0,.5) node[above]{$\scriptstyle c_2$};
\node[red] at (.2,-.2) {$\scriptstyle \xi$};
}
\middle|
\tikzmath{
\draw[thick, red] (-.5,-.5) node[left]{$\scriptstyle b_1'$} -- (.5,.5) node[right]{$\scriptstyle b_2'$};
\draw[thick, red] (-.5,0) node[left]{$\scriptstyle a_1'$} -- (.5,0) node[right]{$\scriptstyle a_2'$};
\draw[thick, red] (0,-.5) node[below]{$\scriptstyle c_1'$} -- (0,.5) node[above]{$\scriptstyle c_2'$};
\node[red] at (.2,-.2) {$\scriptstyle \xi'$};
}
\right\rangle
=
\left(\prod_{i=1,2}
\frac{\delta_{a_i=a_i'}
\delta_{b_i=b_i'}
\delta_{c_i=c_i'}}{\sqrt{d_{a_i}d_{b_i}d_{c_i}}}
\right)
\cdot
\tr_\cA(\xi^\dag\circ \xi').
$$
For the setup with the $\cA$-enriched UFC $\cX$-boundary, each boundary vertex is assigned a different local Hilbert space
\begin{equation}
\label{eq:BoundaryVertexSpaceOrientation}
\tikzmath{
\draw[thick, red] (-.5,-.5) node[left]{$\scriptstyle b_1$} -- (0,0);
\draw[thick] (-.5,0) node[left]{$\scriptstyle x_1$} -- (.5,0) node[right]{$\scriptstyle x_2$};
\draw[thick] (0,-.5) node[below]{$\scriptstyle y_1$} -- (0,.5) node[above]{$\scriptstyle y_2$};
\filldraw (0,0) circle (.05cm);
\draw[blue!50, very thin] (-.5,.5) -- (.5,-.5);
{\draw[blue!50, very thin, -stealth] ($ (-.3,.3) - (.1,.1)$) to ($ (-.3,.3) + (.1,.1)$);}
{\draw[blue!50, very thin, -stealth] ($ (.3,-.3) - (.1,.1)$) to ($ (.3,-.3) + (.1,.1)$);}
}
\qquad
\longleftrightarrow
\qquad
\cK_v 
:= 
\bigoplus_{
a_i,b_i,c_i\in \Irr(\cA)
}
\cX(x_1\Phi(b_1)y_1\to y_2x_2)
\end{equation}
equipped with the enriched skein module inner product \cite[\S3.1]{2305.14068}:
$$
\left\langle
\tikzmath{
\draw[thick, red] (-.5,-.5) node[left]{$\scriptstyle b$} -- (0,0);
\draw[thick] (-.5,0) node[left]{$\scriptstyle x_1$} -- (.5,0) node[right]{$\scriptstyle x_2$};
\draw[thick] (0,-.5) node[below]{$\scriptstyle y_1$} -- (0,.5) node[above]{$\scriptstyle y_2$};
\node at (.2,-.2) {$\scriptstyle \eta$};
}
\middle|
\tikzmath{
\draw[thick, red] (-.5,-.5) node[left]{$\scriptstyle b'$} -- (0,0);
\draw[thick] (-.5,0) node[left]{$\scriptstyle x_1'$} -- (.5,0) node[right]{$\scriptstyle x_2'$};
\draw[thick] (0,-.5) node[below]{$\scriptstyle y_1'$} -- (0,.5) node[above]{$\scriptstyle y_2'$};
\node at (.2,-.2) {$\scriptstyle \eta'$};
}
\right\rangle
=
\frac{\delta_{b=b'}}{\sqrt{d_{b}}}
\left(\prod_{i=1,2}
\frac{\delta_{x_i=x_i'}
\delta_{y_i=y_i'}}{\sqrt{d_{x_i}d_{y_i}}}
\right)
\cdot
\tr_\cX(\eta^\dag\circ \eta').
$$
We refer the reader to \cite[\S3.1]{2305.14068} for the full definition of the enriched skein module, which is entirely similar to the skein module.

Consider now a 3D rectangle $\Lambda$ in our lattice $\cL$.
We consider the canonical spin system from this setup, i.e., we define local algebras
$\fA(\Lambda):= \bigotimes_{v\in \Lambda} B(\cH_v)$,
and we set $\fA:= \varinjlim \fA(\Lambda) = \bigotimes_v B(\cH_v)$.
In the presence of $\cX$-boundary sites, we replace the necessary copies of $\cH_v$ with $\cK_v$.
As before, there are edge terms $A_\ell$ which enforce that the simple labels on connected edges match, and there are plaquette terms $B_p$ for every square in the lattice which are defined on the edge-restricted subspace (and zero on the orthogonal complement of this space).
These plaquette terms incorporate the braiding in $\cA$ and the half-braiding in $Z(\cX)$:
$$
\frac{1}{D_\cA}
\sum_{a\in \Irr(\cA)}d_a\cdot
\tikzmath{
\draw[thick, red] (.65,-.35) -- (1.35,.35);
\draw[thick, red] (-.35,.65) -- (.35,1.35);
\draw[thick, red] (.65,.65) -- (1.35,1.35);
\filldraw[knot, thick, orange, rounded corners=5pt, fill=orange!60, opacity=.5] (.15,.15) rectangle (.85,.85);
\node[orange] at (.25,.5) {$\scriptstyle a$};
\draw[thick, red, knot] (-.35,-.35) -- (.35,.35);
\draw[thick, red] (-.5,-.5) grid (1.5,1.5);
}\,,
\qquad
\frac{1}{D_\cA}
\sum_{a\in \Irr(\cA)}d_a\cdot
\tikzmath{
\draw[thick, red] (.6,-.9) -- (.6,1.1);
\draw[thick, red] (.1,-.4) -- (1.1,-.4);
\draw[thick, red] (.1,.6) -- (1.1,.6);
\filldraw[knot, thick, orange, rounded corners=5pt, fill=orange!60, opacity=.5] (.7,0) -- (.7,-.2) -- (.9,0) -- (.9,.8) -- (.7,.6) -- (.7,0);
\node[orange] at (.8,.3) {$\scriptstyle a$};
\draw[thick, knot] (.5,-.5) grid (1.5,1.5);
\draw[thick, red] (.25,-.75) -- (1,0);
\draw[thick, red] (.25,.25) -- (1,1);
}\,,
\qquad
\text{or}\qquad
\frac{1}{D_\cX}
\sum_{x\in \Irr(\cX)}d_x\cdot
\tikzmath{
\draw[thick, red] (.65,-.35) -- (1,0);
\draw[thick, red] (-.35,.65) -- (0,1);
\draw[thick, red] (.65,.65) -- (1,1);
\filldraw[knot, thick, blue, rounded corners=5pt, fill=blue!60, opacity=.5] (.15,.15) rectangle (.85,.85);
\node[blue] at (.25,.5) {$\scriptstyle x$};
\draw[thick, red] (-.35,-.35) -- (0,0);
\draw[thick, black] (-.5,-.5) grid (1.5,1.5);
}
$$
In the middle picture above, one uses $\Phi:\cA\to \cX$ to resolve the plaquette operator into the vertices which meet the black $\cX$ edges. 
There is a similar formula for the matrix coefficients $C(\xi,\xi')$ for the $B_p$ which can be obtained by modifying \cite[Prop.~3.2]{2305.14068} which shows that the $B_p$ are self-adjoint projectors.
There are no terms for cubes in the lattice.

As before, we define
$p^A_\Lambda := \prod_{\ell\subset \Lambda}A_\ell$,
$p^B_\Lambda:=\prod_{p\subset \Lambda} B_p$, and $p_\Lambda:=p^A_\Lambda p^B_\Lambda$.
We have the following theorem for the operator $p^B_\Lambda$ on $\im(p^A_\Lambda) =p^A_\Lambda \bigotimes_{v\in \Lambda}\cH_v$
(again replacing $\cH_v$ for $\cK_v$ in the presence of a topological $\cX$ boundary.)

\begin{lem}[{\cite[Thm. 3.4]{2305.14068}}]
For a rectangle $\Lambda$, on $\im(p^A_\Lambda)$,
$$
p^B_\Lambda = 
D_{\cA}^{-\#p\subset \Lambda^\circ + \#c\subset \Lambda^\circ} 
D_{\cX}^{-\#p\subset \partial\Lambda} 
\cdot \eval^\dag\circ \eval,
$$
where 
$\Lambda^\circ$ is the part of $\Lambda$ which does not touch the $\cX$ boundary,
$\#p\subset \Lambda^\circ$ is the number of plaquettes in $\Lambda^\circ$ and $\#c\subset \Lambda^\circ$ is the number of cubes in $\Lambda$, including plaquettes and cubes which meet but are not strictly contained in $\partial\Lambda$, $\#p\subset \partial\Lambda$ is the number of plaquettes contained where $\Lambda$ intersects the $\cX$ boundary,
and $\eval$ is the evaluation map from the image of $p^A_\Lambda$ to the enriched skein module.
Hence $\im(p_\Lambda)$ is unitarily isomorphic to the enriched skein module via an appropriately scaled evaluation map.
\end{lem}

\begin{defn}
We now cut our 3D region along a vertical 2D hyperplane $\cK$ as follows:
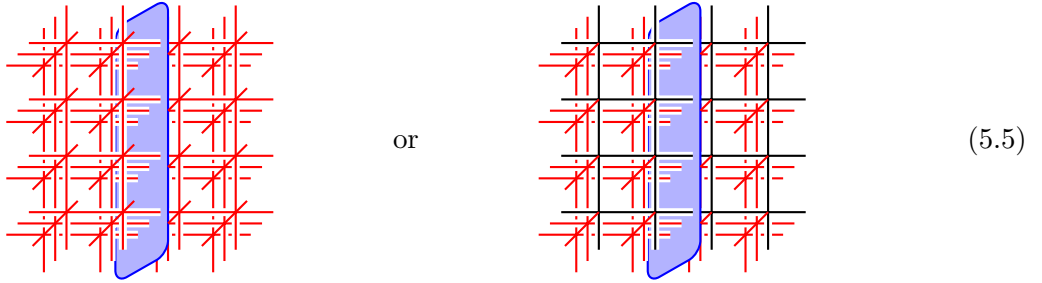
\begin{equation}
\label{eq:BoundaryRectangleWW}
\tikzmath{
\draw[thick, red, step=.75, xshift=-.3cm, yshift=-.3cm] (1.75,0.25) grid (3.5,3.5);
\draw[thick, red, step=.75, knot, xshift=-.15cm, yshift=-.15cm] (1.75,0.25) grid (3.5,3.5);
\draw[thick, red, step=.75, knot] (1.75,0.25) grid (3.5,3.5);
\filldraw[draw=blue, thick, fill=blue!30, rounded corners=5pt]
(2.1,2) -- (2.1,3.6) -- (1.4,3) -- (1.4,-.2) -- (2.1,.4) -- (2.1,2);
\draw[thick, red, step=.75, knot, xshift=-.3cm, yshift=-.3cm] (0.25,0.25) grid (2,3.5);
\draw[thick, red, step=.75, knot, xshift=-.15cm, yshift=-.15cm] (0.25,0.25) grid (2,3.5);
\draw[thick, red, step=.75, knot] (0.25,0.25) grid (2,3.5);
\foreach \x in {.75,1.5,3}{
\foreach \y in {.75,1.5,2.25,3}{
\draw[thick, red] ($ (\x,\y) + (.15,.15) $) -- ($ (\x,\y) + (-.45,-.45) $);
}}
\foreach \y in {.75,1.5,2.25,3}{
\draw[thick, red] ($ (2.25,\y) + (.15,.15) $) -- ($ (2.25,\y) + (-.15,-.15) $);
}
}
\qquad\qquad\text{or}\qquad\qquad
\tikzmath{
\draw[thick, red, step=.75, xshift=-.3cm, yshift=-.3cm] (1.75,0.25) grid (3.5,3.5);
\draw[thick, red, step=.75, knot, xshift=-.15cm, yshift=-.15cm] (1.75,0.25) grid (3.5,3.5);
\draw[thick, step=.75, knot] (1.75,0.25) grid (3.5,3.5);
\filldraw[draw=blue, thick, fill=blue!30, rounded corners=5pt]
(2.1,2) -- (2.1,3.6) -- (1.4,3) -- (1.4,-.2) -- (2.1,.4) -- (2.1,2);
\draw[thick, red, step=.75, knot, xshift=-.3cm, yshift=-.3cm] (0.25,0.25) grid (2,3.5);
\draw[thick, red, step=.75, knot, xshift=-.15cm, yshift=-.15cm] (0.25,0.25) grid (2,3.5);
\draw[thick, step=.75, knot] (0.25,0.25) grid (2,3.5);
\foreach \x in {.75,1.5,3}{
\foreach \y in {.75,1.5,2.25,3}{
\draw[thick, red] (\x,\y) -- ($ (\x,\y) + (-.45,-.45) $);
}}
\foreach \y in {.75,1.5,2.25,3}{
\draw[thick, red] (2.25,\y) -- ($ (2.25,\y) + (-.15,-.15) $);
}
}
\end{equation}
We choose a distinguished half-space $\bbH$ bounded by $\cK$, which is assumed to be to the \emph{left} of the slice in \eqref{eq:BoundaryRectangleWW} above.
For a 3D rectangle $\Lambda\subset\bbH$, we let $I=I(\Lambda)$ be the intersection $\partial\Lambda\cap \cK$.
Observe that $I=\emptyset$ if $\Lambda\cap \cK=\emptyset$.
Since we the sites in $\Lambda$ correspond to the Hilbert spaces from \eqref{eq:VertexSpaceOrientation}, one should view the red strings passing through $\cK$ as oriented from \emph{left-to-right}.
This means gluing by postcomposition will happen on the right-hand side of $\cK$.
In order to make the connection with the braided categorical nets defined in \S\ref{sec:BraidedCategoricalNets} above, one can look at the intersection points in $I$ of the strands with $\cK$ from the right hand side.
For example, for the diagram on the right side of \eqref{eq:BoundaryRectangleWW}, we can shift our perspective by `looking from the right' and then rotating as follows.
\[
\tikzmath{
\filldraw[draw=blue, thick, fill=blue!30, rounded corners=5pt]
(2.1,2) -- (2.1,3.6) -- (1.4,3) -- (1.4,-.2) -- (2.1,.4) -- (2.1,2);
\draw[thick, red, step=.75, knot, xshift=-.3cm, yshift=-.3cm] (0.25,0.25) grid (2,3.5);
\draw[thick, red, step=.75, knot, xshift=-.15cm, yshift=-.15cm] (0.25,0.25) grid (2,3.5);
\draw[thick, step=.75, knot] (0.25,0.25) grid (2,3.5);
\foreach \x in {.75,1.5}{
\foreach \y in {.75,1.5,2.25,3}{
\draw[thick, red] (\x,\y) -- ($ (\x,\y) + (-.45,-.45) $);
}}
}
\qquad\qquad
\rightsquigarrow
\qquad\qquad
\tikzmath[scale=-1]{
\draw[thick, red, step=.75, xshift=-.3cm, yshift=-.3cm] (1.9,0.25) grid (3.5,3.5);
\draw[thick, red, step=.75, knot, xshift=-.15cm, yshift=-.15cm] (1.9,0.25) grid (3.5,3.5);
\draw[thick, step=.75, knot] (1.9,0.25) grid (3.5,3.5);
\filldraw[draw=blue, thick, fill=blue!30, opacity=.5, rounded corners=5pt]
(2.1,2) -- (2.1,3.6) -- (1.4,3) -- (1.4,-.2) -- (2.1,.4) -- (2.1,2);
\foreach \y in {.75,1.5,2.25,3}{
\draw[thick, red] (3,\y) -- ($ (3,\y) + (-.45,-.45) $);
\draw[thick, red] (2.25,\y) -- ($ (2.25,\y) + (-.15,-.15) $);
}
\foreach \y in {.75,1.5,2.25,3}{
\filldraw (1.9,\y) circle (.05cm);
\foreach \x in {.15,.3}{
\filldraw[red] ($ (1.9,\y) - \x*(1,1) $) circle (.05cm);
}
}
}
\qquad\qquad
\rightsquigarrow
\qquad\qquad
\tikzmath[scale=.75]{
\draw[blue, thick, rounded corners=5pt]
(-2,0) rectangle (2,3);
\foreach \x in {-1.5,-.5,.5,1.5}{
\filldraw (\x,.5) circle (.067cm);
\foreach \y in {1.5,2.5}{
\filldraw[red] (\x,\y) circle (.067cm);
}}
}
\]

We now define our nets of algebras which act as gluing operators.
For the case without boundary as in the left hand side of \eqref{eq:BoundaryRectangleWW}, when $I\neq \emptyset$, we define
$$
\fF(I)
:=
\End_\cA(A^{\otimes I})
$$
where $A^{\otimes I}$ is the tensor product of the 
generator $A:=\bigoplus_{a\in\Irr(\cA)} a$ over the sites in $I$.
In the presence of the $\cX$ topological boundary located on sites in the boundary $\partial\cL$ of our lattice $\cL$, we modify the above definition to account for sites in $I$ which meet $\partial \cL$:
$$
\fG(I)
:=
\begin{cases}
\End_\cA(A^{\otimes I})
&
\text{if }I\cap \partial \cL=\emptyset
\\
\End_\cX(X^{\otimes I \cap \partial \cL}\lhd A^{\otimes I\setminus \partial \cL})   
&
\text{if }I\cap \partial \cL\neq\emptyset
\end{cases}
$$
where $X=\bigoplus_{x\in\Irr(\cX)}x$ is our strong tensor generator of $\cX$.

\end{defn}

Just as \cite[Def.~4.6]{MR4945955}, for $\psi \in \fF(I)$, we can define a gluing operator $\Gamma_\psi$ acting on $\bigotimes_{v\in \Lambda} \cH_v$. 
When $I$ is as in \eqref{eq:BoundaryRectangleWW} and
$\psi \in \cA(\bigotimes_I a_i \to \bigotimes_I b_i)$, $\Gamma_\psi$ first applies the projector $p_\Lambda$ so that the edge labels of the vectors match and plaquettes are effectively filled.
This allows us to view a vector in $p_\Lambda\bigotimes_{v\in\Lambda} \cH_v$ 
as an element of the enriched skein module \cite[\S3.1]{2305.14068} for $\cA$, giving an element in $\cA(1_\cA \to A^{\otimes \partial_E \Lambda})$.
We can then take our desired morphism $\psi$ and glue it on, first projecting to make sure the labels $a_i$ in $\bigotimes_I a_i$ match in order so composition is well-defined.
We have a similar definition for $\Gamma_\psi$ in the presence of the $\cX$ toplogical boundary for $\psi\in \fG(I)$.
We must also multiply by ratios of fourth roots of quantum dimensions in order to make this a $*$-action, which are omitted in the cartoons below. 
$$
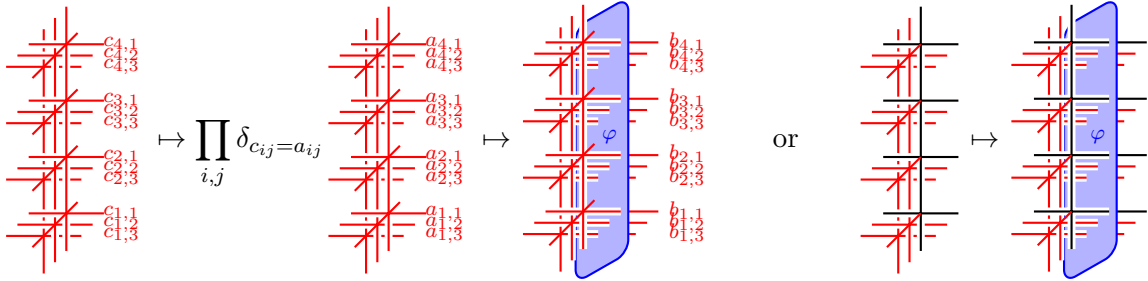
\tikzmath{
\draw[thick, red, step=.75, knot, xshift=-.3cm, yshift=-.3cm] (1,0.25) grid (2,3.5);
\draw[thick, red, step=.75, knot, xshift=-.15cm, yshift=-.15cm] (1,0.25) grid (2,3.5);
\draw[thick, red, step=.75, knot] (1,0.25) grid (2,3.5);
\foreach \x in {1.5}{
\foreach \y in {.75,1.5,2.25,3}{
\draw[thick, red] ($ (\x,\y) + (.15,.15) $) -- ($ (\x,\y) + (-.45,-.45) $);
}}
\foreach \y in {1,2,3,4}{
\foreach \s in {1,2,3}{
\draw[thick, red] ($ (1.85,0) + \y*(0,.75) + (0,.15) - \s*(0,.15) $) node[right]{$\scriptstyle c_{\y,\s}$};
}}
}
\mapsto
\prod_{i,j} \delta_{c_{ij}=a_{ij}}
\tikzmath{
\draw[thick, red, step=.75, knot, xshift=-.3cm, yshift=-.3cm] (1,0.25) grid (2,3.5);
\draw[thick, red, step=.75, knot, xshift=-.15cm, yshift=-.15cm] (1,0.25) grid (2,3.5);
\draw[thick, red, step=.75, knot] (1,0.25) grid (2,3.5);
\foreach \x in {1.5}{
\foreach \y in {.75,1.5,2.25,3}{
\draw[thick, red] ($ (\x,\y) + (.15,.15) $) -- ($ (\x,\y) + (-.45,-.45) $);
}}
\foreach \y in {1,2,3,4}{
\foreach \s in {1,2,3}{
\draw[thick, red] ($ (1.85,0) + \y*(0,.75) + (0,.15) - \s*(0,.15) $) node[right]{$\scriptstyle a_{\y,\s}$};
}}
}
\mapsto
\tikzmath{
\foreach \y in {.75,1.5,2.25,3}{
\foreach \s in {0,.15,.3}{
\draw[thick, red] ($ (1.75,\y) - (0,\s) $) -- ($ (2.5,\y) - (\s,\s) $);
}}
\filldraw[draw=blue, thick, fill=blue!30, rounded corners=5pt]
(2.1,2) -- (2.1,3.6) -- (1.4,3) -- (1.4,-.2) -- (2.1,.4) -- (2.1,2);
\node[blue] at (1.85,1.75) {$\scriptstyle\psi$};
\draw[thick, red, step=.75, knot, xshift=-.3cm, yshift=-.3cm] (1,0.25) grid (2,3.5);
\draw[thick, red, step=.75, knot, xshift=-.15cm, yshift=-.15cm] (1,0.25) grid (2,3.5);
\draw[thick, red, step=.75, knot] (1,0.25) grid (2,3.5);
\foreach \x in {1.5}{
\foreach \y in {.75,1.5,2.25,3}{
\draw[thick, red] ($ (\x,\y) + (.15,.15) $) -- ($ (\x,\y) + (-.45,-.45) $);
}}
\foreach \y in {1,2,3,4}{
\foreach \s in {1,2,3}{
\draw[thick, red] ($ (2.5,0) + \y*(0,.75) + (0,.15) - \s*(0,.15) $) node[right]{$\scriptstyle b_{\y,\s}$};
}}
}
\qquad\text{or}\qquad
\tikzmath{
\draw[thick, red, step=.75, knot, xshift=-.3cm, yshift=-.3cm] (1,0.25) grid (2,3.5);
\draw[thick, red, step=.75, knot, xshift=-.15cm, yshift=-.15cm] (1,0.25) grid (2,3.5);
\draw[thick, step=.75, knot] (1,0.25) grid (2,3.5);
\foreach \x in {1.5}{
\foreach \y in {.75,1.5,2.25,3}{
\draw[thick, red] ($ (\x,\y) $) -- ($ (\x,\y) + (-.45,-.45) $);
}}
}
\mapsto
\tikzmath{
\foreach \y in {.75,1.5,2.25,3}{
\draw[thick] ($ (1.75,\y) $) -- ($ (2.5,\y) $);
\foreach \s in {.15,.3}{
\draw[thick, red] ($ (1.75,\y) - (0,\s) $) -- ($ (2.5,\y) - (\s,\s) $);
}}
\filldraw[draw=blue, thick, fill=blue!30, rounded corners=5pt]
(2.1,2) -- (2.1,3.6) -- (1.4,3) -- (1.4,-.2) -- (2.1,.4) -- (2.1,2);
\node[blue] at (1.85,1.75) {$\scriptstyle\psi$};
\draw[thick, red, step=.75, knot, xshift=-.3cm, yshift=-.3cm] (1,0.25) grid (2,3.5);
\draw[thick, red, step=.75, knot, xshift=-.15cm, yshift=-.15cm] (1,0.25) grid (2,3.5);
\draw[thick, step=.75, knot] (1,0.25) grid (2,3.5);
\foreach \x in {1.5}{
\foreach \y in {.75,1.5,2.25,3}{
\draw[thick, red] ($ (\x,\y) $) -- ($ (\x,\y) + (-.45,-.45) $);
}}
}
$$
This new morphism must be resolved into a vector in $\bigotimes_{v\in\Lambda} \cH_v$ in the usual way using semisimplicity.

Arguing as in \cite[Lem.~4.7 and Thm.~4.8]{MR4945955}, but using a variant of the \emph{sphere algebra} of $\cA$ \cite[Def.~3.12]{2305.14068} instead of the tube algebra, we get the following theorem.

\begin{thm}
\label{thm:LTQO-WW}
The Walker-Wang model for $\cA$ satisfies \ref{LTO:QECC}--\ref{LTO:Injective} for $s=1$,
and the boundary algebra
$$
\fB(\Lambda\Subset \Delta)=\set{\Gamma_\psi}{\psi \in \End_\cA(A^{\otimes I})}\cong \fF(I)
\qquad\qquad\qquad
I=\partial\Lambda\cap \partial\Delta\cap \cK.
$$
\end{thm}

Using a variant of the \emph{dome algebra} of the $\cA$-enriched UFC $\cX$ \cite[Def.~3.6]{2305.14068}, we get the following theorem.

\begin{thm}
\label{thm:BoundaryWWBoundaryAlgebra}
The boundary Walker-Wang model for the $\cA$-enriched UFC $\cX$ satisfies \ref{BoundaryLTO:QECC}--\ref{BoundaryLTO:Injective} for $s=1$, and the boundary algebra
$$
\fB(\Lambda\Subset^\partial \Delta)=\set{\Gamma_\psi}{\psi \in \End_\cX(X^{\otimes I \cap \partial \cL}\lhd A^{\otimes I\setminus \partial \cL})   )}
\cong 
\fG(I)
\qquad\qquad
I=\partial\Lambda\cap \partial\Delta\cap \cK.
$$
\end{thm}

\subsection{DHR bimodules for the Walker-Wang model}

In order to compute the DHR bimodules for the boundary algebra of the Walker-Wang model for $\cA$, we first analyze the boundary DHR bimodules for the boundary Walker-Wang model for the $\cA$-enriched UFC $\cX$.
We saw in Theorem~\ref{thm:BoundaryWWBoundaryAlgebra} above that the boundary algebra for this boundary LTO is given by the braided categorical net with boundary $\fG$ associated to the $\cA$-enriched UFC $\cX$.
In Theorem \ref{thm:BoundaryDHRForBraidedEnrichedNet} above, we showed that $\DHR^\partial(\fG)\cong Z^\cA(\cX)$, the enriched center of $\cX$, a.k.a., the M\"uger centralizer $\cA'\subset Z(\cX)$.

With this theorem in hand, we can now compute the DHR bimodules for the boundary algebra $\fB\cong\fF$ of the $\cA$ Walker-Wang model.
We apply a 3D folding trick along a 2D hyperplane $\cH$ which is orthogonal to our cut $\cK$.
Folding effectively doubles each lattice site away from $\cH$ and views sites in $\cH$ as the boundary $\partial \cL$ of the lattice.
The resulting folded lattice now has $\cA$-sites along the boundary $\cH$, and a pair of lattice points for each site in the bulk, one associated to $\cA$ and the other to $\cA^{\rm rev}$.\footnote{Reflecting the graphical calculus for a UFC $\cC$ gives the arrow opposite $\cC^{\rm op}$ or the monoidal opposite $\cC^{\rm mp}$ depending on if you reflect along the $x$ or $y$-axis respectively.
These categories are (unitarily) monoidally equivalent since $\cC$ is rigid. 
When $\cA$ is a UBFC, both $\cA^{\rm op}$ and $\cA^{\rm mp}$ are (unitarily) equivalent to $\cA^{\rm rev}$, which has the reverse braiding.
}
(which is equivalent to what one obtains from $\cA$ under any reflection in the graphical calculus).
The resulting bulk-boundary system can then be described as a boundary Walker-Wang model,
namely the $\cH=\partial\cL$ boundary is an $\cA$-boundary for the 3D Walker-Wang model for $\cA^{\rm rev}\boxtimes \cA$, i.e., we view the underlying UFC of $\cA$ as an $\cA^{\rm rev}\boxtimes \cA$-enriched UFC.

Now since we started with the 3D $\cA$-Walker-Wang model which satisfies the LTO axioms, the folded system satisfies the boundary LTO axioms.
(Explicitly, this is just Theorem \ref{thm:BoundaryWWBoundaryAlgebra} for the model built from $\cA$ as an $\cA^{\rm rev}\boxtimes \cA$-enriched UFC.)
Now with our cut $\cK$ perpendicular to our fold $\cH=\partial \cL$, the boundary algebra $\fG$ as in Theorem \ref{thm:BoundaryWWBoundaryAlgebra} is obtained from the boundary algebra $\fF$ for $\cA$ in Theorem \ref{thm:LTQO-WW} via folding.

Given a DHR-bimodule for $\fF$, we may localize it in a rectangle $I$ which is symmetric about our fold $\cH$, so that the resulting folded rectangle is a sufficiently large boundary rectangle for $\fG$.
(Since these are braided categorical nets, this sufficient largeness is quite small by the discussion in Construction \ref{construct:BoundaryDHR} or \eqref{eq:PartitionOfUnityGivesProjectiveBasis}.)
We thus get a fully faithful monoidal functor 
\begin{equation}
\label{eq:RestrictDHRToBoundary}
\operatorname{Restrict}:\DHR(\fF)\longrightarrow \DHR^\partial(\fG)
\end{equation}
which is actually braided as we may define the braiding for $\DHR(\fF)$ using disjoint such rectangles.
By Theorem \ref{thm:BoundaryDHRForBraidedEnrichedNet}, we have $\DHR^\partial(\fG)\cong Z^{\cA^{\rm rev}\boxtimes \cA}(\cA)\cong Z_2(\cA)$, the \emph{M\"uger center} of $\cA$, as UBFCs \cite{MR1990929}.
Finally, to show \eqref{eq:RestrictDHRToBoundary} is an equivalence, similar to Construction \ref{const:BoundaryDHRFOrEnrichedNet}, we can construct a fully-faithful braided tensor functor $Y:Z_2(\cA)\hookrightarrow \DHR(\fF)$ such that the following diagram commutes:
\[
\begin{tikzcd}
\DHR(\fF) 
\arrow[rr,hook,"\operatorname{Restrict}"]
&&
\DHR^\partial(\fG)
\\
&
Z_2(\cA) 
\arrow[hook, ul,"Y"]
\arrow[ur, "\cong","\text{(Thm.~\ref{thm:BoundaryDHRForBraidedEnrichedNet})}"']
\end{tikzcd}
\]
This is entirely similar to Construction \ref{const:BoundaryDHRFOrEnrichedNet} by taking inductive limits of the spaces
$$
Y^z_n := 
\Hom_\cA( 
A^{R_n}
\to 
A^{R_n}
\otimes z
)
$$
for $z\in Z_2(\cA)$, where $(R_n)$ is a sequence of rectangles symmetric about $\cH$ which exhaust $\bbZ^2$ as $n\to\infty$.
That $Y$ is braided follows similarly to the proof of Theorem \ref{thm:BoundaryDHRForBraidedEnrichedNet}.

\begin{rem}
For the Levin-Wen model from a UFC $\cC$, one can also use the folding trick to get a 2D boundary LTO.
The resulting boundary algebra $\fM$ is the fusion module boundary algebra for $\cC$ as a $\cC^{\rm mp}\boxtimes \cC$-module.
We then have $\End(\cC_{\cC^{\rm mp}\boxtimes \cC})\cong Z(\cC)$, so we get the correct category of boundary DHR-bimodules as a UFC.
However, we no longer have a braiding on $\DHR^\partial(\fM)$ in this setting. 
\end{rem}

\appendix
\section{AF actions of unitary multifusion categories}
\label{appendix:AFActionsOfUmFCs}

In this section, we review how one can get an AF action of a unitary multifusion category on an AF $\rmC^*$-algebra by adapting the results of \cite{MR4916103}.
We begin by reviewing these results.

Let $\cC$ be a unitary multifusion category.
Fix a \emph{faithful 2-shading} on $\cC$, i.e., a decomposition $1_\cC=1_+\oplus 1_-$ such that $\cC_{+-}:=1_+\otimes \cC\otimes 1_- \neq 0$.
Observe that the similarly defined $\cC_{++},\cC_{-+},\cC_{--}$ are all nonzero, and we may think of $\cC$ as a $2\times 2$ block of categories
$$
\cC=
\begin{pmatrix}
\cC_{++} & \cC_{+-}
\\
\cC_{-+} & \cC_{--}
\end{pmatrix}.
$$

We represent $1_+$ by an unshaded region and $1_-$ by a shaded region.
Pick a generator $X\in \cC_{+-}$ (which generates $\cC$ under tensor product $\otimes$, algebraic duals $\vee$, direct sum $\oplus$, and idempotent completion), and we equip $\cC$ with the \emph{standard unitary dual functor} $\vee$ with respect to $X$ \cite{MR4133163} such that the loop parameters for $X$ are scalars:
$$
\tikzmath{
\fill[rounded corners=5pt, gray!50] (-.8,-.5) rectangle (.8,.5);
\filldraw[fill=white] (0,0) circle (.25cm);
\node at (.5,0) {$\scriptstyle X$};
\node at (-.5,0) {$\scriptstyle X^\vee$};
}
=
\lambda\cdot\,
\tikzmath{
\fill[rounded corners=5pt, gray!50] (-.5,-.5) rectangle (.5,.5);
}
\qquad
\tikzmath{
\draw[rounded corners=5pt, dotted] (-.8,-.5) rectangle (.8,.5);
\filldraw[fill=gray!50] (0,0) circle (.25cm);
\node at (.5,0) {$\scriptstyle X^\vee$};
\node at (-.5,0) {$\scriptstyle X$};
}
=
\lambda\cdot\,
\tikzmath{
\draw[rounded corners=5pt, dotted] (-.5,-.5) rectangle (.5,.5);
}
\qquad\qquad
\tikzmath{
\draw[rounded corners=5pt, dotted] (-.3,-.3) rectangle (.3,.3);
}
=
\cC_{++}
\qquad
\tikzmath{
\fill[rounded corners=5pt, gray!50] (-.3,-.3) rectangle (.3,.3);
}
=
\cC_{--}.
$$
Observe that $X\in\cC_{+-}$ being a generator means that the map
$\End_\cC(1_{+})\otimes \End_\cC(1_-)\to Z(\End_\cC(X))$ given by 
\begin{equation}
\label{eq:ZEndXIso}
\tikzmath{
\draw[rounded corners=5pt, dotted] (-.5,-.5) rectangle (.5,.5);
\roundNbox{}{(0,0)}{.3}{0}{0}{$x$}
}\,
\otimes 
\,\tikzmath{
\fill[rounded corners=5pt, gray!50] (-.5,-.5) rectangle (.5,.5);
\roundNbox{fill=white}{(0,0)}{.3}{0}{0}{$y$}
}\,
\longmapsto
\,\tikzmath{
\begin{scope}
\clip[rounded corners=5pt] (-1,-.5) rectangle (1,.5);
\fill[gray!50] (0,-.5) rectangle (1,.5);    
\end{scope}
\draw[rounded corners=5pt,dotted] (0,-.5) -- (-1,-.5) -- (-1,.5) --(0,.5);
\roundNbox{fill=white}{(-.5,0)}{.3}{0}{0}{$x$}
\roundNbox{fill=white}{(.5,0)}{.3}{0}{0}{$y$}
\draw (0,-.5) node[below]{$\scriptstyle X$} -- (0,.5) node[above]{$\scriptstyle X$};
}\,
\end{equation}
is an isomorphism.

Since $\cC$ is unitary multifusion, by \cite{MR4133163}, there is a faithful \emph{spherical weight} $\psi: \cC(1_\cC\to 1_\cC)\to \bbC$ such that
$$
\psi\left(
\tikzmath{
\draw (0,.3) arc(180:0:.3cm) --node[right]{$\scriptstyle c^\vee$} (.6,-.3) arc(0:-180:.3cm);
\roundNbox{}{(0,0)}{.3}{0}{0}{$f$}
}
\right)
=
\psi\left(
\tikzmath{
\draw (0,.3) arc(0:180:.3cm) --node[left]{$\scriptstyle c^\vee$} (-.6,-.3) arc(-180:0:.3cm);
\roundNbox{}{(0,0)}{.3}{0}{0}{$f$}
}
\right)
\qquad\qquad\qquad
\forall\,f:c\to c,
\quad
\forall\,c\in\cC,
$$
which is unique subject to the normalization
$\psi(\id_{1_+})=\psi(\id_{1_-})=1$.
This means that $\Tr^\cC:=\psi\circ \tr_L^\vee = \psi\circ \tr_R^\vee$ is a unitary trace on $\cC$ which equips it with the structure of a \emph{2-Hilbert space} in the sense of \cite{MR1448713}.
Each hom space $\cC(a\to b)$ is equipped with a Hilbert space structure given by
$$
\langle f| g\rangle_{a\to b}:= 
\Tr^\cC_a(f^\dag\circ g)
=
\Tr^\cC_b(g\circ f^\dag)
$$
which satisfies that for all $f:a\to b$, $g:b\to c$, and $h: a\to c$,
$$
\langle g| h\circ f^\dag\rangle_{b\to c}
=
\langle g\circ f| h\rangle_{a\to c}
=
\langle f| g^\dag\circ h\rangle_{a\to b}.
$$

\begin{construction}
One now constructs two AF $\rmC^*$-algebras as follows.
First, we write
$$
X^{\alt\otimes n}:= \underbrace{X\otimes X^\vee \otimes X\otimes \cdots X^?}_{n\text{ tensorands}}
\qquad\qquad\text{and}\qquad\qquad
(X^\vee)^{\alt\otimes n}:= \underbrace{X^\vee\otimes X \otimes X^\vee\otimes \cdots X^?}_{n\text{ tensorands}}
$$
where the question marks are determined by parity of $n$.
We define
$$
A_n:= \End_\cC(X^{\alt \otimes n})
\qquad\qquad\text{and}\qquad\qquad
B_n:=\End_\cC((X^\vee)^{\alt\otimes n}).
$$
We have the inductive limit AF algebras $A:=\varinjlim A_n$ and $B:=\varinjlim B_n$,
where the inclusions are always given by tensoring by identities on the right, i.e., taking inductive limit over the map
$$
-\otimes \id:
\tikzmath{
\draw[cyan, dashed] (0,.7) -- (0,-.7);
\draw[cyan,->] (-.15,.5) -- (.15,.5);
\draw[cyan,->] (-.15,-.5) -- (.15,-.5);
\draw (-.7,0) node[left]{$\scriptstyle n$} -- (-.3,0);
\draw (.7,0) node[right]{$\scriptstyle n$} -- (.3,0);
\roundNbox{fill=white}{(0,0)}{.3}{0}{0}{$\eta$}
}
\,
\longmapsto
\,
\tikzmath{
\draw[cyan, dashed] (0,.7) -- (0,-.7);
\draw[cyan,->] (-.15,.5) -- (.15,.5);
\draw (-.7,-.5) -- (.7,-.5);
\draw (-.7,0) node[left]{$\scriptstyle n$} -- (-.3,0);
\draw (.7,0) node[right]{$\scriptstyle n$} -- (.3,0);
\roundNbox{fill=white}{(0,0)}{.3}{0}{0}{$\eta$}
}
$$
Here, a strand labeled $n$ denotes $X^{\alt\otimes n}$ or $(X^\vee)^{\alt\otimes n}$,
and the single strand represents either $X$ or $X^\vee$ depending on the parity of $n$.
The top string is always $X$ for $A$ and $X^\vee$ for $B$.
As the map \eqref{eq:ZEndXIso} is an isomorphism, we see that 
$Z(A)\cong \End_\cC(1_+)$ and $Z(B)\cong \End_\cC(1_-)$.
This also relies on the fact that eventually the Bratteli diagrams for $A,B$ are stationary \cite[Thm.~4.1.4]{MR999799} (see also \cite[\S3.1]{MR4227743}).
The inclusion $A\subset B$ is obtained by taking the inductive limit of the map $\id_{X^\vee}\otimes -:A_n\hookrightarrow B_{n+1}$ given by
$$
\tikzmath{
\draw[cyan, dashed] (0,.7) -- (0,-.7);
\draw[cyan,->] (-.15,.5) -- (.15,.5);
\draw[cyan,->] (-.15,-.5) -- (.15,-.5);

\draw (-.7,0) node[left]{$\scriptstyle n$} -- (-.3,0);
\draw (.7,0) node[right]{$\scriptstyle n$} -- (.3,0);
\roundNbox{fill=white}{(0,0)}{.3}{0}{0}{$\eta$}
}
\,
\longmapsto
\,
\tikzmath{
\draw[cyan, dashed] (0,.7) -- (0,-.7);
\draw[cyan,->] (-.15,-.5) -- (.15,-.5);
\draw (-.7,.5) -- (.7,.5);
\draw (-.7,0) node[left]{$\scriptstyle n$} -- (-.3,0);
\draw (.7,0) node[right]{$\scriptstyle n$} -- (.3,0);
\roundNbox{fill=white}{(0,0)}{.3}{0}{0}{$\eta$}
}\,.
$$

Now on $A_n$, we can define a faithful tracial state by $\tr_n:= d_X^{-n}\Tr^\cC$, and similarly for $B_n$.
(The normalization for $\psi$ ensures that $\tr_n(1)=1$ for all $n$ for both $A$ and $B$.)
Note that the inclusions 
$A_n\hookrightarrow A_{n+1}$
and
$A_n\hookrightarrow B_{n+1}$
are both trace preserving.
We can thus complete $A,B$ to $\rm II_1$ \emph{multifactors} (direct sums of $\rm II_1$ factors) $\cA,\cB$ by using the GNS construction with respect to $\tr:=\varinjlim \tr_n$,
which yields a homogeneous connected finite index inclusion whose standard invariant $(\cC(\cA\subset \cB), {}_\cA L^2\cB_\cB)$ is unitarily equivalent to $(\cC,X)$ \cite[Prop.~6.5]{MR4916103}.
These results rely on Popa's generating tunnel for finite depth inclusions \cite{MR1055708}
and the Ocneanu Compactness Theorem \cite{MR1473221}, \cite[\S11.4]{MR1642584}, generalized to the multifactor setting \cite[Appendices B and C]{MR4916103}. 
To see this, one performs Jones' basic construction to get a tower of multifactors equipped with traces
$$
\cA_0:=\cA \subset \cB:=\cA_1 
\overset{e_1}{\subset} 
\cA_2 
\overset{e_2}{\subset} 
\cA_3
\subset \cdots
$$
and the \emph{higher relative commutants} can be identified with endomorphism spaces in $\cC$ as follows \cite{MR1424954}:
\begin{equation}
\label{eq:HigherRelativeCommutants}
\begin{aligned}
\cA_0' \cap \cA_{2k} 
&\cong 
\End_{\cA-\cA}(L^2\cB^{\boxtimes_\cA k})
\cong 
\End_\cC((X\otimes X^\vee)^{\otimes k})
\\
\cA_0' \cap \cA_{2k+1} 
&\cong 
\End_{\cA-\cB}(L^2\cB^{\boxtimes_\cA k+1})
\cong 
\End_\cC((X\otimes X^\vee)^{\otimes k}\otimes X)
\\
\cA_1' \cap \cA_{2k+1} 
&\cong 
\End_{\cB-\cB}(L^2\cB^{\boxtimes_\cA k})
\cong 
\End_\cC((X^\vee\otimes X)^{\otimes k+1})
\\
\cA_1' \cap \cA_{2k+2} 
&\cong 
\End_{\cB-\cA}(L^2\cB^{\boxtimes_\cA k+1})
\cong 
\End_\cC((X^\vee\otimes X)^{\otimes k}\otimes X^\vee).
\end{aligned}
\end{equation}
We thus get a unitary tensor functor into $\Bim(\cA\oplus \cB)$ from the full subcategory $\cC_0\subset \cC$ whose objects are of the form $X^{\alt\otimes n}$ and $(X^\vee)^{\alt\otimes n}$.
We then use the universal property of unitary idempotent completion to uniquely extend the functor to
$$
\cF:\cC
=
\begin{pmatrix}
\cC_{++} & \cC_{+-}  
\\
\cC_{-+} & \cC_{--}
\end{pmatrix}
\hookrightarrow
\begin{pmatrix}
\Bim(\cA) & \Bim(\cA,\cB)
\\
\Bim(\cB,\cA) & \Bim(\cB)
\end{pmatrix}
=
\Bim(\cA\oplus \cB).
$$

There are also uniqueness results in \cite{MR4916103}, but these are not necessary for the results in this article.
\end{construction}

\begin{rem}
\label{rem:MoreDetailForFF-AF-Action}
The unitary tensor functor $\cF$ can be described in more detail using the \emph{multistep basic construction}.
Our analysis follows the exposition from \cite[\S6.3]{MR4753059}.
The main point is that the inclusion of the form 
$$
\cA_j \subset \cA_{j+k} \overset{f^{j+k}_j}\subset \cA_{j+2k}
\qquad\qquad\qquad
f^{j+k}_j
:=
\frac{1}{d_X^j}\cdot\,
\tikzmath{
\draw (0,-.5) --node[left]{$\scriptstyle j$} (0,.5);
\draw (.3,.5) node[above]{$\scriptstyle k$} arc(-180:0:.3cm);
\draw (.3,-.5) node[below]{$\scriptstyle k$} arc(180:0:.3cm);
}
$$
is also a Jones basic construction \cite{MR965748}.
This means that $\cA_0'\cap \cA_{2k}=\End_{\cA-\cA}(L^2\cA_k)$, and similarly for the other 3 cases.
Iterating the Jones basic construction gives the isomorphism $L^2\cA_k \cong L^2\cB^{\boxtimes_\cA k}$, leading to the isomorphisms \eqref{eq:HigherRelativeCommutants} above.

Now one effectively performs the unitary idempotent completion of $\cC_0\subset \Bim(\cA\oplus \cB)$ by splitting projections in these higher relative commutants. 
We explain this in detail for an arbitrary $p\in \cA_0'\cap \cA_{2k}$, and the other three cases are similar.
The corresponding $\cA-\cA$ bimodule is given by
$pL^2\cA_k \cong pL^2\cB^{\boxtimes_\cA k}$.
Now as $p\in \cA_0'\cap \cA_{2k}$, in order for $p$ to act on $\cA_k$, we can embed $\cA_k$ into $\cA_{2k}$ where the action is just left multiplication.
But in order to preserve the endomorphism space of the bimodule, one should also multiply by the Jones projection $f^{2k}_k$ on the right, which commutes with the $\cA-\cA$ action.
In diagrams, the map $\cA_k\cong \cA_k f^{2k}_k = \cA_{2k}f^{2k}_k$ (by \cite[Pull Down Lemma~1.2]{MR860811}) is given by taking the inductive limit of the map
$$
\tikzmath{
\draw[cyan, dashed] (0,.8) -- (0,-.8);
\draw[cyan,->] (-.15,.6) -- (.15,.6);
\draw[cyan,->] (-.15,-.6) -- (.15,-.6);
\draw (-.7,.2) node[left]{$\scriptstyle k$} -- (-.3,.2);
\draw (-.7,-.2) node[left]{$\scriptstyle n$} -- (-.3,-.2);
\draw (.7,.2) node[right]{$\scriptstyle k$} -- (.3,.2);
\draw (.7,-.2) node[right]{$\scriptstyle n$} -- (.3,-.2);
\roundNbox{fill=white}{(0,0)}{.4}{-.1}{-.1}{$\eta$}
}
\,
\longmapsto
\,
\frac{1}{d_X^k}
\tikzmath{
\draw[cyan, dashed] (0,.8) -- (0,-.8);
\draw[cyan,->] (-.15,-.6) -- (.15,-.6);
\draw (.7,.6) node[right]{$\scriptstyle k$} -- (-.3,.6) arc (90:270:.2cm);
\draw (-1,.6) node[left]{$\scriptstyle k$} -- (-.9,.6) arc(90:-90:.2cm) -- (-1,.2) node[left]{$\scriptstyle k$};
\draw (-1,-.2) node[left]{$\scriptstyle n$} -- (-.3,-.2);
\draw (.7,.2) node[right]{$\scriptstyle k$} -- (.3,.2);
\draw (.7,-.2) node[right]{$\scriptstyle n$} -- (.3,-.2);
\roundNbox{fill=white}{(0,0)}{.4}{-.1}{-.1}{$\eta$}
}
$$
under $-\otimes \id$.
(Recall that the left action is post-composition, and the right action is pre-composition.)
It is important to note that the $\cA-\cA$ actions on $\cA_k$ only interact with the $n$ strands over which we take the inductive limit, and not the top $k$ strings;
similarly for $\cA_{2k}$ and the top $2k$ strands. 
Left multiplication by $p$ is then given by
$$
\frac{1}{d_X^k}
\tikzmath{
\draw[cyan, dashed] (0,.8) -- (0,-.8);
\draw[cyan,->] (-.15,-.6) -- (.15,-.6);
\draw (.7,.6) node[right]{$\scriptstyle k$} -- (-.3,.6) arc (90:270:.2cm);
\draw (-1,.6) node[left]{$\scriptstyle k$} -- (-.9,.6) arc(90:-90:.2cm) -- (-1,.2) node[left]{$\scriptstyle k$};
\draw (-1,-.2) node[left]{$\scriptstyle n$} -- (-.3,-.2);
\draw (.7,.2) node[right]{$\scriptstyle k$} -- (.3,.2);
\draw (.7,-.2) node[right]{$\scriptstyle n$} -- (.3,-.2);
\roundNbox{fill=white}{(0,0)}{.4}{-.1}{-.1}{$\eta$}
}
\,\longmapsto\,
\frac{1}{d_X^k}
\tikzmath{
\draw[cyan, dashed] (0,.8) -- (0,-.8);
\draw[cyan,->] (-.15,-.6) -- (.15,-.6);
\draw (.7,.6) -- node[above]{$\scriptstyle k$} (.3,.6) -- (-.3,.6) arc (90:270:.2cm);
\draw (-1,.6) node[left]{$\scriptstyle k$} -- (-.9,.6) arc(90:-90:.2cm) -- (-1,.2) node[left]{$\scriptstyle k$};
\draw (-1,-.2) node[left]{$\scriptstyle n$} -- (-.3,-.2);
\draw (.7,.2) --node[above]{$\scriptstyle k$}  (.3,.2);
\draw (1.7,.6) node[right]{$\scriptstyle k$} -- (1.3,.6);
\draw (1.7,.2) node[right]{$\scriptstyle k$} -- (1.3,.2);
\draw (1.7,-.2) node[right]{$\scriptstyle n$} -- (.3,-.2);
\roundNbox{fill=white}{(0,0)}{.4}{-.1}{-.1}{$\eta$}
\roundNbox{fill=white}{(1,.4)}{.4}{-.1}{-.1}{$p$}
}
$$
Taking inductive limits gives the corresponding bimodule.

Since the $\cA-\cA$ action only interacts with the $n$ strands below under which we take the inductive limit, writing  $c\in\cC$ for the image of $p$,
we have an obvious $\cA-\cA$ bimodule isomorphism with $\cF(c):=\varinjlim \cF(c)_n$ where
$$
\cF(c)_n:=
F(c)_n:=\cC(X^{\alt\otimes n} \to c\otimes X^{\alt\otimes n})=
\left\{
\tikzmath{
\draw[cyan, dashed] (-.7,.7) -- (.7,-.7);
\draw[cyan,->] (-.6,.4) -- (-.4,.6);
\draw[cyan,->] (.4,-.6) -- (.6,-.4);
\draw (-.7,0) node[left]{$\scriptstyle n$} -- (-.3,0);
\draw (.7,0) node[right]{$\scriptstyle n$} -- (.3,0);
\draw[thick, red] (0,.3) --node[right]{$\scriptstyle c$} (0,.7);
\roundNbox{fill=white}{(0,0)}{.3}{0}{0}{$\eta$}
}
\right\}.
$$
Given another $d\in\cC$, 
the tensorator $\cF(c)_n \boxtimes_{A_n} \cF(d)_n \cong \cF(c\otimes d)_n$ is given graphically by
$$
\tikzmath{
\draw (-.7,0) node[left]{$\scriptstyle n$} -- (-.3,0);
\draw (.7,0) node[right]{$\scriptstyle n$} -- (.3,0);
\draw[thick, red] (0,.3) --node[right]{$\scriptstyle c$} (0,.7);
\roundNbox{}{(0,0)}{.3}{0}{0}{$\eta$}
}
\,\boxtimes\,
\tikzmath{
\draw (-.7,0) node[left]{$\scriptstyle n$} -- (-.3,0);
\draw (.7,0) node[right]{$\scriptstyle n$} -- (.3,0);
\draw[thick, blue] (0,.3) --node[right]{$\scriptstyle d$} (0,.7);
\roundNbox{}{(0,0)}{.3}{0}{0}{$\xi$}
}
\,\longmapsto\,
\tikzmath{
\draw (-.7,0) node[left]{$\scriptstyle n$} -- (-.3,0);
\draw (.7,0) --node[above]{$\scriptstyle n$} (.3,0);
\draw[thick, red] (0,.3) --node[right]{$\scriptstyle c$} (0,.7);
\roundNbox{}{(0,0)}{.3}{0}{0}{$\eta$}
\draw (1.7,0) node[right]{$\scriptstyle n$} -- (1.3,0);
\draw[thick, blue] (1,.3) --node[right]{$\scriptstyle d$} (1,.7);
\roundNbox{}{(1,0)}{.3}{0}{0}{$\xi$}
}
\,.
$$
We simply observe that the $\cA$-valued inner product on $F(c)_n$ coming from the trace on $\cA$ lands in $A_n$ and is given by
$$
\langle \eta|\xi\rangle_{A_n}:= \eta^\dag\circ \xi \in \End((A_n)_{A_n})=A_n.
$$
Moreover, the inclusions $-\otimes \id: F(c)_n\hookrightarrow F(c)_{n+1}$
are compatible with the inclusions $-\otimes \id : A_n\hookrightarrow A_{n+1}$, and
$$
\langle \eta\otimes\id|\xi\otimes\id\rangle_{A_{n+1}}
=
(\eta\otimes \id)^\dag\circ (\xi\otimes \id)
=
(\eta^\dag\circ \xi) \otimes \id
=
\langle \eta|\xi\rangle_{A_n}\otimes \id.
$$
Finally, this map is compatible with the associator and unitors in $\cC_{++}$, yielding a tensorator for our fully faithful tensor functor $\cF:\cC_{++}\to\Bim(\cA)$.
\end{rem}

\begin{rem}
\label{rem:NoNeedForvNA}
One does not need to pass to the von Neumann algebraic closures in order for the results on standard invariants and unitary tensor functors to hold.
Indeed, these results hold on the level of the AF $\rmC^*$-algebras and their Hilbert $\rmC^*$-bimodules/correspondences.
For the standard invariant for $A\subset B$, where again the inclusion is tensor product on the left with $\id_{X^\vee}$, we again have a unitary equivalence $(\cC(A\subset B), {}_AB_B)\cong (\cC,X)$.
One argues this by showing we have canonical inclusions
\begin{align*}
&\End_\cC((X{\otimes} X^\vee)^{{\otimes} n})
&\subseteq&
\End_{A-A}(B^{\boxtimes_A n})
&\subseteq&
\End_{\cA-\cA}(L^2\cB^{\boxtimes_A n})
&\cong&
\End_\cC((X{\otimes} X^\vee)^{{\otimes} n})
\\
&\End_\cC((X{\otimes} X^\vee)^{{\otimes} n}{\otimes} X)
&\subseteq&
\End_{A-B}(B^{\boxtimes_A n+1})
&\subseteq&
\End_{\cA-\cB}(L^2\cB^{\boxtimes_A n+1})
&\cong&
\End_\cC((X{\otimes} X^\vee)^{{\otimes} n}{\otimes} X)
\\
&\End_\cC((X^\vee{\otimes} X)^{{\otimes} n})
&\subseteq&
\End_{B-B}(B^{\boxtimes_A n+1})
&\subseteq&
\End_{\cB-\cB}(L^2\cB^{\boxtimes_A n+1})
&\cong&
\End_\cC((X^\vee{\otimes} X)^{{\otimes} n})
\\
&\End_\cC((X^\vee{\otimes} X)^{{\otimes} n}{\otimes} X^\vee)
&\subseteq&
\End_{B-A}(B^{\boxtimes_A n+1})
&\subseteq&
\End_{\cB-\cA}(L^2\cB^{\boxtimes_A n+1})
&\cong&
\End_\cC((X^\vee{\otimes} X)^{{\otimes} n}{\otimes} X^\vee)
\end{align*}
which are thus isomorphisms by the finite dimensionality of endomorphism spaces of $\cC$.
We thus obtain a fully faithful unitary tensor functor
$$
F:\cC
=
\begin{pmatrix}
\cC_{++} & \cC_{+-}  
\\
\cC_{-+} & \cC_{--}
\end{pmatrix}
\hookrightarrow
\begin{pmatrix}
\Bim(A) & \Bim(A,B)
\\
\Bim(B,A) & \Bim(B)
\end{pmatrix}
=
\Bim(A\oplus B).
$$
such that $X\mapsto {}_A B_B$.
Similar to before, each $F(c):=\varinjlim F(c)_n$ now taken with respect to the $A$-valued inner product $\varinjlim \langle \,\cdot\, |\,\cdot\,\rangle_{A_n}$ gives an explicit description of an
inductive limit $A-A$ $\rmC^*$-Hilbert bimodule.
\end{rem}

\begin{ex}
Let $\cC$ be a unitary multifusion category, and let $\cM_\cC$ be a unitary right $\cC$-module category.
For any choice of generator $W\in\cM$, we can run the above construction to obtain a fully faithful unitary tensor functor
$$
\cE = 
\begin{pmatrix}
\End(\cM_\cC) & \cM
\\
\cM^{\rm op} & \cC
\end{pmatrix}
\hookrightarrow
\begin{pmatrix}
\Bim(A) & \Bim(A,B)
\\
\Bim(B,A) & \Bim(B)
\end{pmatrix}
=
\Bim(A\oplus B)
$$
such that $W\mapsto {}_A B_B$.
\end{ex}

We now turn to the explicit example that appears in \S\ref{sec:BoundaryDHRforBoundaryLW} above, and we describe how to employ the results of \cite{MR4916103} described above in order to prove the corresponding construction is fully faithful.

\begin{ex}
Suppose now that we want to work with two generators: some $W\in\cM$ and $X\in \cC$ rather than a single generator.
We impose the requirement that $X\in\cC$ is a strong tensor generator, i.e., there is an $n\in\bbN$ such that every $c\in\Irr(\cC)$ appears as a summand of $X^{\otimes n}$.

First, we use the above construction for our generator $X\in\cC$ to obtain a fully faithful unitary tensor functor $\cC\to \Bim(A)$.
(We do this by first replacing $\cC$ with $\operatorname{Mat}_2(\cC)$ and viewing our generator $X$ as living off the diagonal.)

We now employ the \emph{Q-system completion} technique from \cite{MR4419534}.
The basic idea is that given a fully faithful tensor functor $F: \cC\to \Bim(A)$ for a $\rmC^*$-algebra $A$ and a Q-system ($\rmC^*$ Frobenius algebra object) $Q\in \cC$, one can form the \emph{Q-system realization} $|Q|_F$, which is the $\rmC^*$-algebra $F(Q)\in \Bim(A)$, equipped with its own multiplication and involution. 
There are many utilities for $|Q|_F$; the one we are most interested in is \cite[Cor.~C]{MR4419534}, which states that we have a fully faithful composite 2-functor
$$
\mathsf{QSys}(\cC) 
\xrightarrow{\mathsf{QSys}(F)}
\mathsf{QSys}(\Bim(A))
\xrightarrow{|\cdot|}
\mathsf{C^*Alg}.
$$
As a corollary, for any Q-system $Q\in\cC$, we get a fully faithful unitary tensor functor from the unitary (multi)tensor category ${}_Q\cC_Q$ into $\Bim(|Q|_F)$.

There is still one technical subtlety to overcome to deal with our specific example.
Given our module generator $W\in\cM$, we would like that the inductive limit algebra
$$
B:= \varinjlim \End_\cM(W\lhd X^{\otimes n})
$$
is the Q-system realization of $\underline{\End}_\cC(W)\in\cC$.
In more detail, one uses duality in $\cE$ so that
$$
B= \varinjlim \Hom_\cE(X^{\otimes n}\to W^\vee\otimes W\otimes X^{\otimes n})
=
\varinjlim \Hom_\cC(X^{\otimes n}\to \underline{\End}_\cC(W)\otimes X^{\otimes n})
$$
where multiplication is given by 
$$
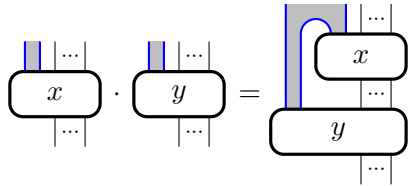
\tikzmath{
\fill[gray!50] (-.4,.3) rectangle (-.2,.7);
\draw[thick, blue] (-.4,.3) -- (-.4,.7);
\draw[thick, blue] (-.2,.3) -- (-.2,.7);
\draw (0,.3) -- (0,.7);
\node at (.23,.5) {$\scriptstyle \cdots$};
\draw (.4,.3) -- (.4,.7);
\draw (0,-.3) -- (0,-.7);
\node at (.23,-.5) {$\scriptstyle \cdots$};
\draw (.4,-.3) -- (.4,-.7);
\roundNbox{}{(0,0)}{.3}{.3}{.3}{$x$}
}
\,\cdot\,
\tikzmath{
\fill[gray!50] (-.4,.3) rectangle (-.2,.7);
\draw[thick, blue] (-.4,.3) -- (-.4,.7);
\draw[thick, blue] (-.2,.3) -- (-.2,.7);
\draw (0,.3) -- (0,.7);
\node at (.23,.5) {$\scriptstyle \cdots$};
\draw (.4,.3) -- (.4,.7);
\draw (0,-.3) -- (0,-.7);
\node at (.23,-.5) {$\scriptstyle \cdots$};
\draw (.4,-.3) -- (.4,-.7);
\roundNbox{}{(0,0)}{.3}{.3}{.3}{$y$}
}
\,=
\tikzmath{
\fill[gray!50] (-.2,.7) -- (-.2,.3) -- (-.4,.3) arc(0:180:.2cm) -- (-.8,-.7) -- (-1,-.7) -- (-1,.7);
\draw[thick, blue] (-.4,.3) arc(0:180:.2cm) -- (-.8,-.7);
\draw[thick, blue] (-.2,.3) -- (-.2,.7);
\draw[thick, blue] (-1,-.7) -- (-1,.7);
\draw (0,.3) -- (0,.7);
\node at (.23,.5) {$\scriptstyle \cdots$};
\draw (.4,.3) -- (.4,.7);
\draw (0,-.3) -- (0,-.7);
\node at (.23,-.5) {$\scriptstyle \cdots$};
\draw (.4,-.3) -- (.4,-.7);
\draw (0,-1.3) -- (0,-1.7);
\node at (.23,-1.5) {$\scriptstyle \cdots$};
\draw (.4,-1.3) -- (.4,-1.7);
\roundNbox{}{(0,0)}{.3}{.3}{.3}{$x$}
\roundNbox{}{(0,-1)}{.3}{.9}{.3}{$y$}
}
$$
The main problem here is that the choice of cap above depends on a choice of unitary duality for $\cE$.
Moreover, one must use the notion of \emph{unitary internal end} in order for $\underline{\End}_\cC(W)$ to define an honest element of $\cC$ and not just the underlying multifusion category $\cC^{\natural}$ where we have forgotten the dagger structure. 
Both of these goals can be accomplished by endowing $\cM$ with a carefully chosen \emph{unitary module trace} as in \S\ref{sec:BoundaryLTOforBoundaryLW}, which allows one to use \emph{unitary adjunction} \cite[\S2.1]{MR4750417} to constuct the unitary internal hom as in \cite[\S3.3]{2410.05120}.
Indeed, when $\cM$ is indecomposable,  $\End(\cM_\cC)$ is fusion, so the bubble for $W$ is always a positive scalar.
There is thus a unique unitary module trace such that the $W$-bubble is equal to one.\footnote{This is entirely analogous to the situation of an operator $T: H\to K$ such that $T^*T=r 1_H$ for some $r>0$.
We can uniquely rescale the inner product on $H$ so that $T$ becomes unitary.}
For an arbitrary module, there is a unique faithful convex combination of traces on the indecomposable summands so that the bubble for $W$ in $\End(\cM_\cC)$ equals one. 
This immediately implies that $|\underline{\End}_\cC(W)|$ is a (special) Q-system
$$
\tikzmath{
\fill[gray!50] (-.2,-.7) to[out=90,in=-90] (-.6,0) to[out=90,in=-90] (-.2,.7) -- (.2,.7) to[out=-90,in=90] (.6,0) to[out=-90,in=90] (.2,-.7);
\draw[thick, blue] (-.2,-.7) to[out=90,in=-90] (-.6,0) to[out=90,in=-90] (-.2,.7);
\draw[thick, blue] (.2,-.7) to[out=90,in=-90] (.6,0) to[out=90,in=-90] (.2,.7);
\filldraw[thick, blue, fill=white] (0,0) circle (.3cm);
}
\,\,
=
\,\,
\tikzmath{
\fill[gray!50] (-.2,-.7) rectangle (.2,.7);
\draw[thick, blue] (-.2,-.7) -- (-.2,.7);
\draw[thick, blue] (.2,-.7) -- (.2,.7);
}
\,\,
=
\id_{\underline{\End}_\cC(W)},
$$
and we thus get an isomorphism $B\cong |\underline{\End}_\cC(W)|$.

Now by \cite[Cor.~C]{MR4419534}, we have a fully faithful tensor functor $\End(\cM_\cC)\hookrightarrow \Bim(B)$.
Unpacking as in Remark \ref{rem:MoreDetailForFF-AF-Action} above, 
for $F\in \End(\cM_\cC)$,
we obtain exactly the DHR bimodules $Y^F:=\varinjlim Y^F_n$ where
$$
Y^F_n := \cM(W\lhd X^{\otimes n} \to FW\lhd \otimes X^{\otimes n})
$$
for the boundary algebra $\fB=B=\varinjlim \End_\cM(W\lhd X^{\otimes n})$ from \S\ref{sec:BoundaryDHRforBoundaryLW} above.

One can also repeat the above analysis for the Q-system $1\oplus |\underline{\End}_\cC(W)| \in\cC$ in order to obtain a fully faithful tensor functor 
$$
\cE = 
\begin{pmatrix}
\End(\cM_\cC) & \cM
\\
\cM^{\rm op} & \cC
\end{pmatrix}
\hookrightarrow
\begin{pmatrix}
\Bim(B) & \Bim(B,A)
\\
\Bim(A,B) & \Bim(A)
\end{pmatrix}
=
\Bim(B\oplus A).
$$
\end{ex}

\section*{Acknowledgements}
The authors would like to thank
Gregory Faurot,
Peter Huston,
Charlton Li,
Emeline Sherman,
Ryan Thorngren, 
and
Dominic Williamson
for helpful conversations.

\section*{Declarations}

\noindent\textbf{Funding.}
Corey Jones was supported by NSF DMS-2247202.
David Penneys was supported by NSF DMS-2154389 and 2554723.
Additionally, this work was supported in part by NSF DMS-1928930 while David Penneys was in residence at the Mathematical Sciences Research Institute/SLMath in Berkeley, California, during Summer 2024,
and in part by grant no.~NSF PHY-2309135 to the Kavli Institute for Theoretical Physics (KITP) while Corey Jones and David Penneys attended the GENSYM25 program in Spring 2025.
\vspace{\baselineskip}

\noindent\textbf{Conflicts of interest/Competing interests.}
The authors have no competing interests to declare that are relevant to the content of this article.
\vspace{\baselineskip}

\noindent\textbf{Data availability statement.}
Not applicable, as no new data were created or analyzed.
\vspace{\baselineskip}

\noindent\textbf{Code availability statement.}
Not applicable, as no code was created or analyzed.

\bibliographystyle{alpha}
\bibliography{bibliography}

\end{document}